%% file: nautytraces2b.tex
\documentclass[amsthm]{elsart}

\pdfoutput=1
\usepackage{yjsco}
\usepackage{natbib_fixed}
\usepackage{pgfplotstable}

\usepackage{amsmath,amssymb}

\makeatletter
\newif\if@restonecol
\makeatother

\usepackage[noline,boxed]{algorithm2e}

\date{--- {\jobname} --- {\today} ---}

\def\nauty{\texttt{nauty}} \def\Nnauty{\texttt{Nauty}}
\def\traces{\texttt{Traces}}
\def\saucy{\texttt{saucy}} \def\Ssaucy{\texttt{Saucy}}
\def\bliss{\texttt{Bliss}}
\def\conauto{\texttt{conauto}} \def\Cconauto{\texttt{Conauto}}

\def\dfrac#1#2{\lower0.15ex\hbox{\large$\frac{#1}{#2}$}}
\def\({\bigl(}
\def\){\bigr)}
\def\isom{\cong}
\def\cat{\mathbin{\Vert}}

\def\G{{\cal G}}
\def\T{{\cal T}}

\def\nullseq{(\,)}

\def\Aut{\operatorname{Aut}}

\def\abs#1{\lvert#1\rvert} \let\card=\abs

\def\nicebreak{\vskip 0pt plus 100pt\penalty-300\vskip 0pt plus -100pt }

\begin{document}
\input{plotsmacros}
\begin{frontmatter}

\title {Practical graph isomorphism, II}

\author{Brendan~D.~McKay}
\address{Research School of Computer Science,
Australian National University,
Canberra ACT 0200, Australia\thanksref{ARC}}
\ead{bdm@cs.anu.edu.au}
\thanks[ARC]{Supported by the Australian Research Council.}

\author{Adolfo Piperno}
\address{Dipartimento di Informatica,
Sapienza Universit\`a di Roma,
Via Salaria 113, I-00198 Roma, \\Italy}
\ead{piperno@di.uniroma1.it}

\begin{abstract}
   We report the current state of the graph isomorphism problem
   from the practical point of view.  After describing the general principles
   of the refinement-individualization paradigm and proving its validity,
   we explain how it is
   implemented in several of the key programs. In particular, we
   bring the description of the best known program {\nauty} up to date
   and describe
   an innovative approach called {\traces} that outperforms the
   competitors for many difficult graph classes.  Detailed comparisons
   against {\saucy}, {\bliss} and {\conauto} are presented.
\end{abstract}

\end{frontmatter}

\nicebreak
\section{Introduction}\label{s:intro}

An \textit{isomorphism} between two graphs is a bijection between 
their vertex sets that preserves adjacency.
An \textit{automorphism} is an isomorphism from a graph to itself.
The set of all automorphisms of a graph $G$ form a group under
composition called the \textit{automorphism group} $\Aut(G)$.

The graph isomorphism problem (GI) is that of determining whether
there is an isomorphism between two given graphs.
GI has long been a favorite target of
algorithm designers---so much so that it was already described as
a ``disease'' in 1976~\citep{RandC}.

Though it is not the focus of this paper, we summarize the current state
of the theoretical study of graph isomorphism.
It is obvious that $\mathrm{GI}\in \mathrm{NP}$ but unknown whether 
$\mathrm{GI}\in \text{co-NP}$.  As that implies, no polynomial time
algorithm is known (despite many published claims),
but neither is $\mathrm{GI}$ known to be NP-complete.
NP-completeness is considered unlikely since it would imply collapse of the
polynomial-time hierarchy~\citep{Goldreich}.
The fastest proven running time for GI has stood for three decades at
$e^{O(\sqrt{n\log n})}$~\citep{Babai}.

On the other hand, polynomial time algorithms are known for many 
special classes
of graphs.  The most general such classes are those with a forbidden
minor \citep{Ponomarenko,GroheLICS} and those with a forbidden topological
minor \citep{GroheSODA}.
These classes include many earlier classes such as graphs of bounded
degree \citep{Luks}, bounded genus \citep{Filotti,Miller} and bounded
tree-width \citep{Bodlaender}.
The algorithms resulting from this theory
are most unlikely to be useful in practice.  Only for a very few important
graph classes, such as trees~\citep{Aho} and planar graphs~\citep{Kellog} are
there practical approaches which are sure to outperform general methods
such as described in this paper.

Testing two graphs for isomorphism directly can have the advantage that
an isomorphism might be found long before an exhaustive search is
complete.  On the other hand, it is poorly suited for the common problems
of rejecting isomorphs from a collection of graphs or identifying a graph
in a database of graphs.  For this reason, the most common practical
approach is ``canonical labelling'', a process in which a graph is relabeled
in such a way that isomorphic graphs are identical after relabelling.
When we have an efficient canonical labelling procedure, we can use
a sorting algorithm for removing isomorphs from a large collection and
standard data structures for database retrieval.


\medskip
It is impossible to comprehensively survey the history of this problem
since there are at least a few hundred published algorithms.  However,
a clear truth of history is that the most successful approach has involved
fixing of vertices together with refinement of partitions of the vertex set.
This ``individualization-refinement'' paradigm was introduced by
\citet{Parris} and developed by \citet{CandG} and
\citet{Arlazarov}. However, the first
program that could handle both structurally regular graphs with
hundreds of vertices and graphs with large automorphism groups
was that of \citet{nauty1978,PGI1}, that later became known as
{\nauty}.  The main advantage of {\nauty} over earlier programs was
its innovative use of automorphisms to prune the search. 
Although there were some worthy competitors \citep{Leon,Kocay},
{\nauty} dominated the field for the next several decades.

This situation changed when
\citet{saucy1} introduced {\saucy}, which at that stage was
essentially a reimplementation of the automorphism group subset of
{\nauty} using sparse data structures.  This gave it a very large advantage
for many graphs of practical interest, prompting the first author to release
a version of {\nauty} for sparse graphs.  {\Ssaucy} has since introduced
some important innovations, such as the ability to detect some types of
automorphism (such as those implied by a locally tree-like structure)
very early \citep{saucy2}.  
Soon afterwards \citet{Bliss1,Bliss2} introduced {\bliss}, which also used the same
algorithm but had some extra ideas that helped its performance on 
difficult graphs.  In particular, it allowed refinement operations to be
aborted early in some cases.  The latter idea reached its full
expression in {\traces}, which we introduce in this paper.
More importantly, {\traces} pioneered a major revision of the way
the search tree is scanned, which we will demonstrate to produce great
efficiency gains.

Another program worthy of consideration is {\conauto}~\citep{conauto1,conauto2}.
It does not feature canonically labelling, though it can compare two graphs
for isomorphism.

\medskip

In Section~\ref{s:generic}, we provide a description of algorithms
based on the individualization-refinement paradigm.  It is sufficiently
general to encompass the primary structure of all of the most successful 
algorithms.
In Section~\ref{s:implementation}, we flesh out the details of how
{\nauty} and {\traces} are implemented, with emphasis on how they 
differ from differ.
In Section~\ref{s:performance}, we compare the performance of
{\nauty} and {\traces} with {\bliss}, {\saucy} and {\conauto} when applied
to a variety of families of graphs ranging from those traditionally
easy to the most difficult known.
Although none of the programs is the fastest in all cases, we
will see that {\nauty} is generally the fastest for small graphs
and some easier families, while {\traces} is better, sometimes in
dramatic fashion, for most of the difficult graph families.

\nicebreak
\section{Generic Algorithm}\label{s:generic}

In this section, we give formal definitions of colourings (partitions),
invariants, and group actions.  We then define the search tree which
is at the heart of most recent graph isomorphism algorithms and
explain how it enables computation of automorphism groups and
canonical forms.
This section is intended to be a self-contained introduction to the
overall strategy and does not contain new features.

Let $\G=\G_n$ denote the set of graphs with vertex set $V=\{1,2,\ldots,n\}$.

\subsection{Colourings}\label{ss:colourings}

A \textit{colouring} of $V$ (or of $G\in\G$) is a surjective 
function $\pi$ from $V$
onto $\{1,2,\ldots,k\}$ for some~$k$.
The number of colours, i.e. $k$, is denoted by~$\abs\pi$.
A \textit{cell} of $\pi$ is the set of vertices with
some given colour; that is, $\pi^{-1}(j)$ for some~$j$ with
$1\le j\le \abs\pi$.
A~\textit{discrete colouring} is a colouring in which each cell is a singleton,
in which case $\abs\pi=n$.
Note that a discrete colouring is a permutation of~$V$.

If $\pi,\pi'$ are colourings, then $\pi'$ is \textit{finer than or equal to} $\pi$,
written $\pi'\preceq\pi$, if 
$\pi(v)<\pi(w)\Rightarrow\pi'(v)<\pi'(w)$ for all $v,w\in V$.
(This implies that each cell of $\pi'$ is a subset of a cell of~$\pi$,
but the converse is not true.)

Since a colouring partitions $V$ into cells, it is frequently called a
\textit{partition}.  However, note that the colours come in a
particular order and this
matters when defining concepts like ``finer''.

A pair $(G,\pi)$, where $\pi$ is a colouring of $G$, is called a
\textit{coloured graph}.

\subsection{Group actions and isomorphisms}\label{ss:groupacts}

Let $S_n$ denote the symmetric group acting on $V$.  We indicate the
action of elements of $S_n$ by exponentiation.  That is, for $v\in V$
and $g\in S_n$, $v^g$ is the image of~$v$ under~$g$.
The same notation indicates the induced action on complex structures
derived from $V$; in particular:
\begin{enumerate}
\item[(a)] If $W\subseteq V$, then $W^g=\{w^g : w\in W\}$, and similarly
  for sequences.
\item[(b)] If $G\in\G$, then $G^g\in\G$ has
  $v^g$ adjacent to $w^g$ exactly when $v$ and $w$
  are adjacent in~$G$. As a special case, a discrete colouring $\pi$ 
  is a permutation on $V$ so we can write $G^\pi$.
\item[(c)] If $\pi$ is a colouring of $V$, then $\pi^g$ is the colouring
 with $\pi^g(v)=\pi(v^g)$ for each $v\in V$.
\item[(d)] If $(G,\pi)$ is a coloured graph, then
  $(G,\pi)^g=(G^g,\pi^g)$.
\end{enumerate}

Two coloured graphs $(G,\pi), (G',\pi')$ are \textit{isomorphic} if
there is $g\in S_n$ such that $(G',\pi')=(G,\pi)^g$,
in which case we write $(G,\pi)\isom (G',\pi')$.
Such a $g$ is called an \textit{isomorphism}.
The \textit{automorphism group} $\Aut(G,\pi)$
is the group of isomorphisms of the coloured graph $(G,\pi)$ to
itself; that is,
\[  \Aut(G,\pi) = \{ g\in S_n : (G,\pi)^g=(G,\pi)\}. \]
A \textit{canonical form} is a function
\[  C : \G\times\varPi \to \G\times\varPi \]
such that, for all $G\in\G$, $\pi\in\varPi$ and $g\in S_n$,
\begin{enumerate}
\item[(C1)] $C(G,\pi) \isom (G,\pi)$,
\item[(C2)] $C(G^g,\pi^g) = C(G,\pi)$.
\end{enumerate}
In other words, it assigns to each coloured graph an isomorphic
coloured graph that is a unique representative of its isomorphism
class.  It follows from the definition that $(G,\pi)\isom (G',\pi')
\Leftrightarrow C(G,\pi)=C(G',\pi')$.

Property (C2) is an important property that must be satisfied by 
many functions we define.  It says that if the elements of $V$ 
appearing in
the inputs to the function are renamed in some manner, the
elements of $V$ appearing in the function value are renamed
in the same manner.  We call this \textit{label-invariance}.

\subsection{Search tree}\label{ss:searchtree}

Now we define a rooted tree whose nodes correspond to sequences
of vertices, with the empty sequence at the root of the tree.
The sequences become longer as we move down the tree.
Each sequence corresponds to a colouring of the graph obtained
by giving the vertices in the sequence unique colours then
inferring in a controlled fashion a colouring of the other vertices.
Leaves of the tree correspond to sequences for which the
derived colouring is discrete.

To formally define the tree, we first define a ``refinement function''
that specifies the colouring that corresponds to a sequence.
Let $V^*$ denote the set of finite sequences of vertices.
For $\nu\in V^*$, $\abs\nu$ denotes the number of components of~$\nu$.
If $\nu=(v_1,\ldots,v_k)\in V^*$ and $w\in V$, then $\nu\cat w$
denotes $(v_1,\ldots,v_k,w)$.
Furthermore, for $0\le s\le k$, $[\nu]_s=(v_1,\ldots,v_s)$.
The ordering $\le$ on finite sequences is the lexicographic order: 
If $\nu=(v_1,\ldots,v_k)$ and $\nu'=(v'_1,\ldots,v'_\ell)$, then
$\nu\le\nu'$ if $\nu$ is a prefix of $\nu'$ or there is some
$j\le\min\{k,\ell\}$ such that $v_i=v'_i$ for $i<j$ and $v_j<v'_j$.

A \textit{refinement function} is a function
\[ R : \G\times\varPi\times V^*  \to \varPi \]
such that for any $G\in\G$, $\pi\in\varPi$ and $\nu\in V^*$,
\begin{enumerate}
\item[(R1)] $R(G,\pi,\nu) \preceq\pi$;
\item[(R2)] if $v\in \nu$, then $\{v\}$ is a cell of $R(G,\pi,\nu)$;
\item[(R3)] for any $g\in S_n$, we have
   $R(G^g,\pi^g,\nu^g)=R(G,\pi,\nu)^g$.
\end{enumerate}

To complete the definition of the tree, we need to specify what are
the children of each node.  We do this by choosing one non-singleton
cell of the colouring, called the \textit{target cell}, and appending
an element of it to the sequence.

A \textit{target cell selector} chooses a non-singleton
cell of a colouring, if there is one.  Formally, it is a function
\[  T : \G\times\varPi\times V^* \to 2^V \]
such that for any $\pi_0\in\varPi$, $G\in\G$ and $\nu\in V^*$,
\begin{enumerate}
\item[(T1)] if $R(G,\pi_0,\nu)$ is discrete, then $T(G,\pi_0,\nu)=\emptyset$;
\item[(T2)] if $R(G,\pi_0,\nu)$ is not discrete, then $T(G,\pi_0,\nu)$ is
  a non-singleton cell of $R(G,\pi_0,\nu)$;
\item[(T3)] for any $g\in S_n$, we have
   $T(G^g,\pi^g,\nu^g)=T(G,\pi,\nu)^g$.
\end{enumerate}

\bigskip

Now we can define the \textit{search tree}
$\T(G,\pi_0)$ depending on an initially-specified coloured graph $(G,\pi_0)$.
The nodes of the tree are elements of $V^*$. \\[0.5ex]
(a) The root of $\T(G,\pi_0)$ is the empty sequence $\nullseq$. \\
(b) If $\nu$ is a node of $\T(G,\pi_0)$, let
  $W=T(G,\pi_0,\nu)$. Then the children of $\pi$ are
  \[ \{ \nu\cat w : w\in W \}. \]
This definition implies by (T2) that a node $\nu$ of $\T(G,\pi_0)$ is
a leaf iff $R(G,\pi_0,\nu)$ is discrete.

For any node $\nu$ of $\T(G,\pi_0)$, define $\T(G,\pi_0,\nu)$ to be the
subtree of $\T(G,\pi_0)$ consisting of $\nu$ and all its descendants.
The following lemmas are easily derived using induction
from the definition of the search
tree and the properties of the functions $R$, $T$ and~$I$.

\begin{lem}\label{l:invartrees}
   For any $G\in\G, \pi_0\in\varPi, g\in S_n$, we have
   $\T(G^g,\pi_0^g)= \T(G,\pi_0)^g$. 
\end{lem}
\begin{proof}
  Let $\nu=(v_1,\ldots,v_k)$ be a node of $\T(G,\pi_0)$. It is easily
  proved by induction on $s$ that 
  $[\nu^g]_s$ is a node of $\T(G^g,\pi_0^g)$ for $0\le s\le k$.
  Therefore, $\T(G,\pi_0)^g\subseteq \T(G^g,\pi_0^g)$.  The reverse
  inclusion follows on considering $g^{-1}$ instead, so the lemma
  is proved.
\end{proof}

\begin{cor}\label{c:subtrees}
   Let $\nu$ be a node of $\T(G,\pi_0)$ and let $g\in\Aut(G,\pi_0)$.
   Then $\nu^g$ is a node of $\T(G,\pi_0)$ and
   $\T(G,\pi_0,\nu^g)=\T(G,\pi_0,\nu)^g$. 
\end{cor}
\begin{proof}
  This follows from Lemma~\ref{l:invartrees} on noticing that
  $(G,\pi_0)^g=(G,\pi)$ if $g\in\Aut(G,\pi_0)$.
\end{proof}

\begin{lem}\label{l:stabilizer}
   Let $\nu$ be a node of $\T(G,\pi_0)$ and let $\pi=R(G,\pi_0,\nu)$.
   Then $\Aut(G,\pi)$ is the point-wise stabilizer of $\nu$ in
   $\Aut(G,\pi_0)$.
\end{lem}
\begin{proof}
   By condition (R2), every element of $\Aut(G,\pi)$ stabilizes~$\nu$.
   Conversely, suppose $g\in \Aut(G,\pi_0)$ stabilizes~$\nu$.  Then by (R3),
   $\pi^g = R(G,\pi_0,\nu)^g = R(G,\pi_0,\nu)=\pi$, so
   $g\in \Aut(G,\pi)$.
\end{proof}

\subsection{Automorphisms and canonical forms}\label{ss:automorphisms}

Now we describe how the search tree $\T(G,\pi_0)$, defined as
in the previous subsection, can be used to compute $\Aut(G,\pi_0)$
and a canonical form.

Let $\varOmega$ be some totally ordered set.  A \textit{node invariant}
is a function
\[  \phi : \G\times\varPi\times V^* \to \varOmega, \]
such that for any $\pi_0\in\varPi$, $G\in\G$,
  and distinct $\nu,\nu'\in\T(G,\pi_0)$,
\begin{enumerate}
\item[($\phi1$)] if $\abs\nu=\abs{\nu'}$
  and $\phi(G,\pi_0,\nu)<\phi(G,\pi_0,\nu')$, then for every leaf
  $\nu_1\in\T(G,\pi_0,\nu)$ and leaf $\nu'_1\in\T(G,\pi_0,\nu')$ we have
  $\phi(G,\pi_0,\nu_1) < \phi(G,\pi_0,\nu'_1)$;
\item[($\phi2$)] if $\pi=R(G,\pi_0,\nu)$ and $\pi'=R(G,\pi_0,\nu')$ are discrete, then
 $\phi(G,\pi_0,\nu)=\phi(G,\pi_0,\nu') \allowbreak\Leftrightarrow G^\pi=G^{\pi'}$
 (note that the last relation is equality, not isomorphism);
\item[($\phi3$)] for any $g\in S_n$, we have
   $\phi(G^g,\pi_0^g,\nu^g)=\phi(G,\pi_0,\nu)$.
\end{enumerate}

Say that leaves $\nu,\nu'$ are \textit{equivalent} if
$\phi(G,\pi_0,\nu)=\phi(G,\pi_0,\nu')$.  If this is the case, there is
a unique $g\in\Aut(G,\pi_0)$ such that $\nu^g=\nu'$, namely
$g=R(G,\pi_0,\nu') R(G,\pi_0,\nu)^{-1}$.
(Recall that $R(G,\pi_0,\nu)$ is a permutation if $\nu$ is a leaf.)

According to Corollary~\ref{c:subtrees}, if $\nu$ is a leaf of $\T(G,\pi_0)$,
then so is $\nu^g$ for any $g\in\Aut(G,\pi_0)$.
Moreover, by the properties of $\phi$ these leaves (over $g\in\Aut(G,\pi_0)$)
have the same value of $\phi$ and no other leaf has that value.
Consequently, for any leaf $\nu$,
\begin{align*}
   \Aut(G,\pi_0) &= \{ R(G,\pi_0,\nu') R(G,\pi_0,\nu)^{-1}  \\
    &{\qquad} : \nu' \text{ is a leaf of $\T(G,\pi_0)$ with }
     \phi(G,\pi_0,\nu')=\phi(G,\pi_0,\nu) \}.
\end{align*}

To define a canonical form, let
\[ 
 \phi^*(G,\pi_0) = \max\{ \phi(G,\pi_0,\nu) : \nu \text{ is a leaf of }\T(G,\pi_0)\},
\]
and let $\nu^*$ be any leaf of $\T(G,\pi_0)$ that achieves the maximum.
Now define $C(G,\pi_0) = (G,\pi_0)^{R(G,\pi_0,\nu^*)}$.
By the properties of $\phi$,
$C(G,\pi_0)$ thus defined is independent of the choice of~$\nu^*$.
In particular, we have:
\begin{lem}\label{l:iscanon}
  The function 
  \[  C : \G \times \varPi \to \G \times \varPi \]
  as just defined is a canonical form.
\end{lem}

These observations provide an algorithm for computing $\Aut(G,\pi_0)$
and $C(G,\pi_0)$, once we have defined $T$ and $\phi$.
In practice it is not of much use, since the search tree can be extremely
large and the group is found element by element rather than as a set
of generators.  However, in practice we can
dramatically improve the performance by judicious pruning of the tree.

When we refer to a \textit{leaf} of $\T(G,\pi_0)$, we always mean a node
$\nu$ of $\T(G,\pi_0)$ for which $R(G,\pi_0,\nu)$ is discrete, even if our
pruning of the tree results in additional nodes having no children.

We define three types of pruning operation on the search tree.
\begin{enumerate} 
 \item[(A)] Suppose $\nu,\nu'$ are distinct nodes of $\T(G,\pi_0)$ with
 $\abs\nu=\abs\nu'$ and $\phi(G,\pi_0,\nu)>\phi(G,\pi_0,\nu')$.
 \textit{Operation $P_A(\nu,\nu')$} is to remove $\T(G,\pi_0,\nu')$. 
 \item[(B)] Suppose $\nu,\nu'$ are distinct nodes of $\T(G,\pi_0)$ with
 $\abs\nu=\abs\nu'$ and $\phi(G,\pi_0,\nu)\ne\phi(G,\pi_0,\nu')$.
 \textit{Operation $P_B(\nu,\nu')$} is to remove $\T(G,\pi_0,\nu')$.
 \item[(C)] Suppose $g\in\Aut(G,\pi_0)$ and suppose $\nu<\nu'$ are
 nodes of $\T(G,\pi_0)$ such that $\nu^g=\nu'$.
 \textit{Operation $P_C(\nu,g)$} is to remove $\T(G,\pi_0,\nu')$.
\end{enumerate}

\begin{thm}\label{t:pruning}
  Consider any $G\in\G$ and $\pi_0\in\varPi$.
  \begin{enumerate} 
  \item[(a)]  Suppose any sequence of operations of the form
    $P_A(\nu,\nu')$ or $P_C(\nu,g)$ are performed.
    Then there remains at least one leaf $\nu_1$ with
     $\phi(G,\pi_0,\nu_1)=\phi^*(G,\pi_0)$.
  \item[(b)] Let $\nu_0$ be some fixed leaf of $\T(G,\pi_0)$.
     Suppose any sequence of operations of the form $P_B(\nu,\nu')$
    or $P_C(\nu'',g)$ are performed, where $\phi(G,\pi_0,\nu'')
    \ne\phi(G,\pi_0,[\nu_0]_{\card{\nu''}})$.
     Let $g_1,\ldots,g_k$ be the automorphisms used in the
     operations $P_C$ that were performed, and let
     \[ A = \{ g \in \Aut(G,\pi_0) : \nu_0^g \text{ is a remaining leaf\/}\}. \]
     Then $\Aut(G,\pi_0)$ is generated by $\{g_1,\ldots,g_k\}\cup A$. 
  \end{enumerate}
\end{thm}

\begin{proof}
 To prove claim (a), note that the lexicographically least leaf $\nu_1$
 with $\phi(G,\pi_0,\nu_1)=\phi^*(G,\pi_0)$ is never removed.
 
 For claim (b), note that the lexicographically least leaf $\nu_1$ equivalent
 to~$\nu_0$ is not removed by the allowed operations.
 Choose an arbitrary $g\in\Aut(G,\pi_0)$.  
 By Corollary~\ref{c:subtrees}, $\nu_0^g$ is a leaf of $\T(G,\pi_0)$.
 If it has been removed, that
 must have been by some $P_C(\nu'',g_i)$ with $\nu''<\nu^g$,
 since operation $P_B(\nu,\nu')$ only removes leaves inequivalent
 to~$\nu_0$.
 Note that $\nu_0^{gg_i^{-1}}$ is a leaf descended from $\nu''$ and
 $\nu_0^{gg_i^{-1}}<\nu_0^g$.
 If $\nu_0^{gg_i^{-1}}$, has been
 removed, that must have been due to some $P_C(\nu''',g_j)$ with
 $\nu'''<\nu_0^{gg_i^{-1}}$, so consider the leaf
 $\nu_0^{gg_i^{-1}g_j^{-1}}<\nu_0^{gg_i^{-1}}$.
 Continuing in this way we must eventually find a leaf that has not been
 removed, since the leaf $\nu_1$ is still present.
 That is, there is some $h\in\langle g_1,\ldots,g_k\rangle$ such that
 leaf $\nu_0^{gh}$ has not been removed.  This proves $g$ belongs to
 the group generated by $\{g_1,\ldots,g_k\}\cup A$, as we wished to prove.
\end{proof}

The theorem leaves unspecified where the automorphisms 
for $P_C(\nu,g)$ operations come from.  They might be provided in
advance, detected by noticing two leaves are equivalent, or 
otherwise.  This is discussed in the following section.

\nicebreak
\section{Implementation strategies}\label{s:implementation}

In this section, we describe two implementations of the generic
algorithm, which are distributed together as
{\nauty} and {\traces}~\citep{nautypage}.

\subsection{Refinement}\label{ss:refinement}

Let $G\in\G$.
A colouring of $G$ is called \textit{equitable}
if any two vertices of the same colour are adjacent to the
same number of vertices of each colour.\footnote{Unfortunately,
``equitable colouring'' also has another meaning in graph theory.
More commonly, our concept is called an \textit{equitable partition}.}

It is well known that for every colouring $\pi$ there is a
coarsest equitable colouring $\pi'$ such that $\pi'\preceq\pi$,
and that $\pi'$ is unique up to the order of its cells.
An algorithm for computing $\pi'$ appears in~\citet{PGI1}.
We summarize it in Algorithm~\ref{a:refinement}.

\begin{algorithm}[ht]
\KwData{$\pi$ is the input colouring and $\alpha$ is a sequence of some cells of $\pi$}
\KwResult{the final value of $\pi$ is the output colouring}
\smallskip
 \While{$\alpha$ is not empty {\bf and} $\pi$ is not discrete}{
    Remove some element $W$ from $\alpha$.\\
    \For{each cell $X$ of $\pi$}{
      Let $X_1,\ldots,X_k$ be the fragments of $X$ distinguished
      according\\ to \Indp the number of edges from each vertex to $W$.\\
      \Indm Replace $X$ by $X_1,\ldots,X_k$ in $\pi$.\\
      \eIf{$X\in\alpha$}{
         Replace $X$ by $X_1,\ldots,X_k$ in $\alpha$.
      }{
        Add all but one of the largest of $X_1,\ldots,X_k$ to $\alpha$.
      }
    }
 }
\caption{\strut Refinement algorithm $F(G,\pi,\alpha)$ \label{a:refinement}}
\end{algorithm}

\begin{figure}[ht]
\centering
\includegraphics[scale=0.20]{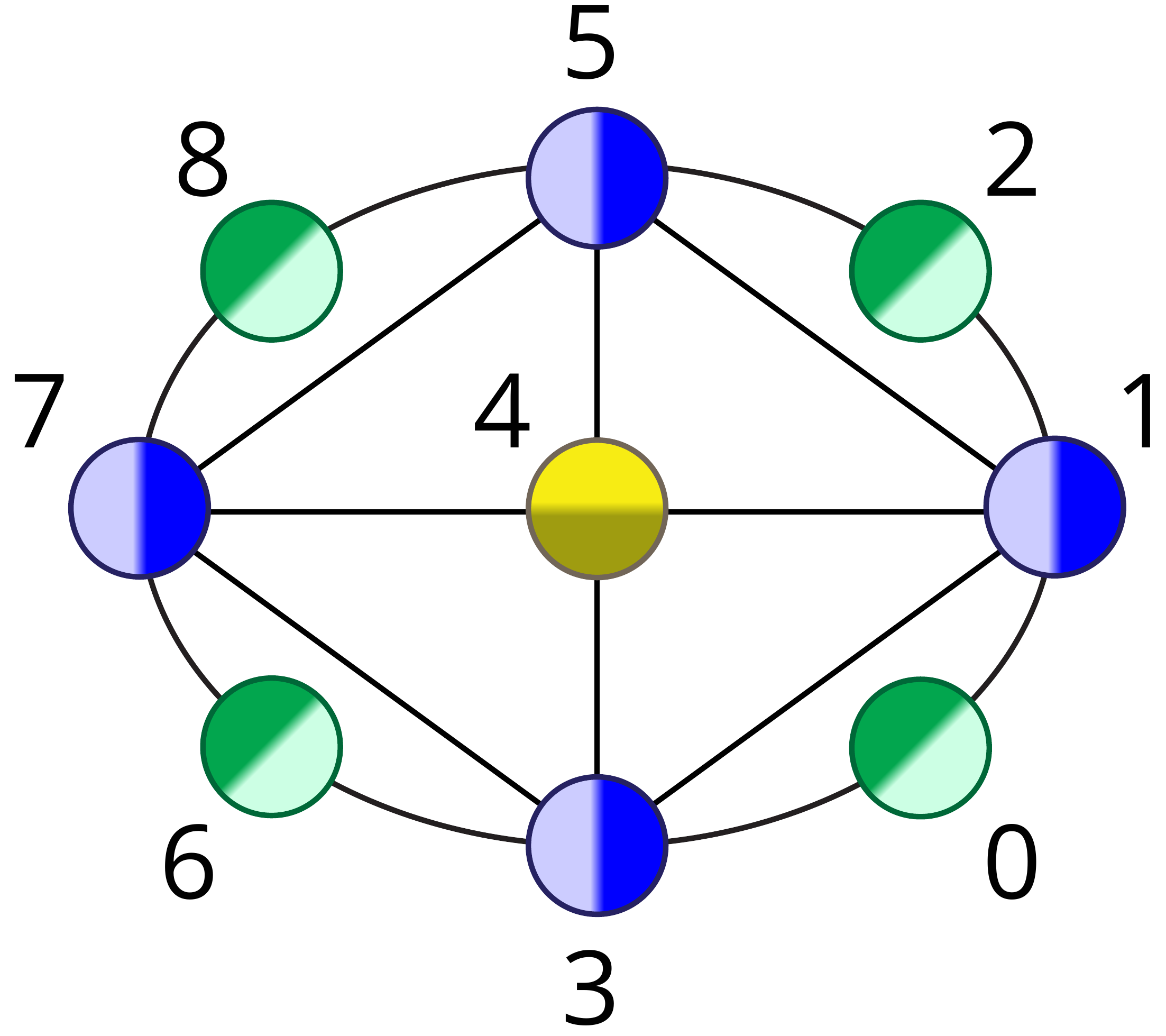}
\caption{Example of an equitable colouring\label{f:equitable}}
\end{figure}

Let $F(G,\pi,\alpha)$ be the function defined by Algorithm~\ref{a:refinement},
which we assume to be implemented in a label-invariant manner.  Now
define the  function
\[ I : \varPi\times V \to \varPi, \]
such that, if $v$ is a vertex in a non-singleton cell of $\pi$
and $\pi' = I(\pi,v)$, then for $w\in V$,
\[
   \pi'(w) = \begin{cases}  
                   \pi(w), & \text{if }\pi(w)<\pi(v)\text{ or }w=v; \\
                   \pi(w)+1, & \text{otherwise}.
               \end{cases}
\]
We see that $I(\pi,v)$ differs from $\pi$ in that a unique colour
has been given to vertex~$v$.
Now we can define a refinement function.
For a sequence of vertices $v_1,v_2,\ldots\,$, define
\begin{align*}
   R(G,\pi_0,\nullseq) &= F(G,\pi_0,\textit{a list of all the cells of $\pi_0$}), \\
   R(G,\pi_0,(v_1)) &= F\(G,I(R(G,\pi_0,\nullseq),v_1),(\{v_1\})\), \displaybreak[0]\\
  R(G,\pi_0,(v_1,v_2)) &= F\(G,I(R(G,\pi_0,(v_1)),v_2),(\{v_2\}))\), \\
  R(G,\pi_0,(v_1,v_2,v_3)) &= F\(G,I(R(G,\pi_0,(v_1,v_2)),v_3),(\{v_3\}))\),
\end{align*}
and so on.  According to Theorem~2.7 and Lemma~2.8 of~\citet{PGI1},
$R$ satisfies (R1)--(R3) and, moreover, $R(G,\pi_0,\nu)$ is equitable.

In practice most of the execution time of the whole algorithm is devoted
to refining colourings, so the implementation is critical. Since the
splitting of $X$ into fragments can be coded more efficiently if
$W$ is a singleton, we have found it advantageous to choose
singletons out of $\alpha$ in preference to larger cells.

\begin{figure}[ht]
\centering
\includegraphics[scale=0.18]{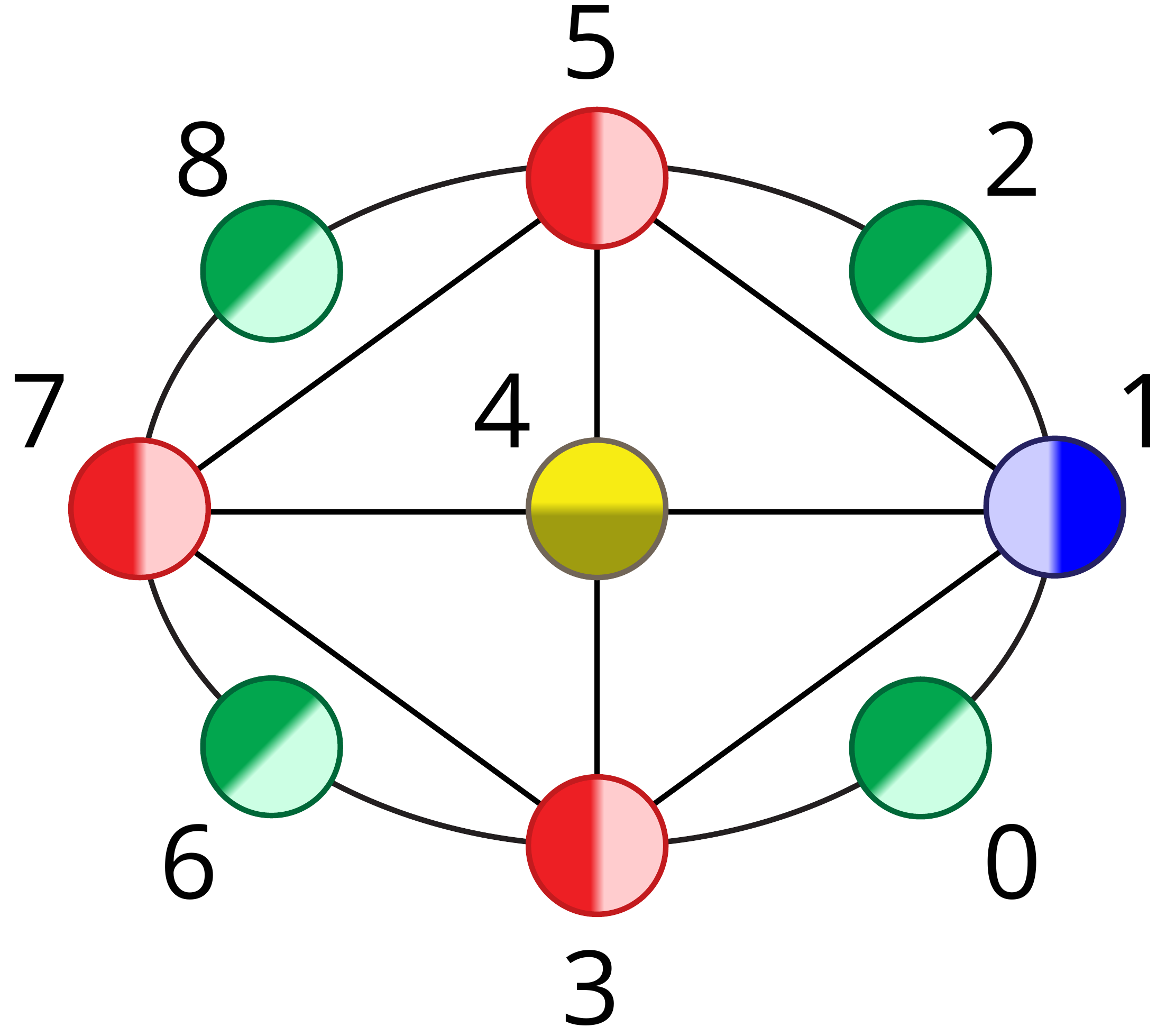}
\qquad\quad
\includegraphics[scale=0.18]{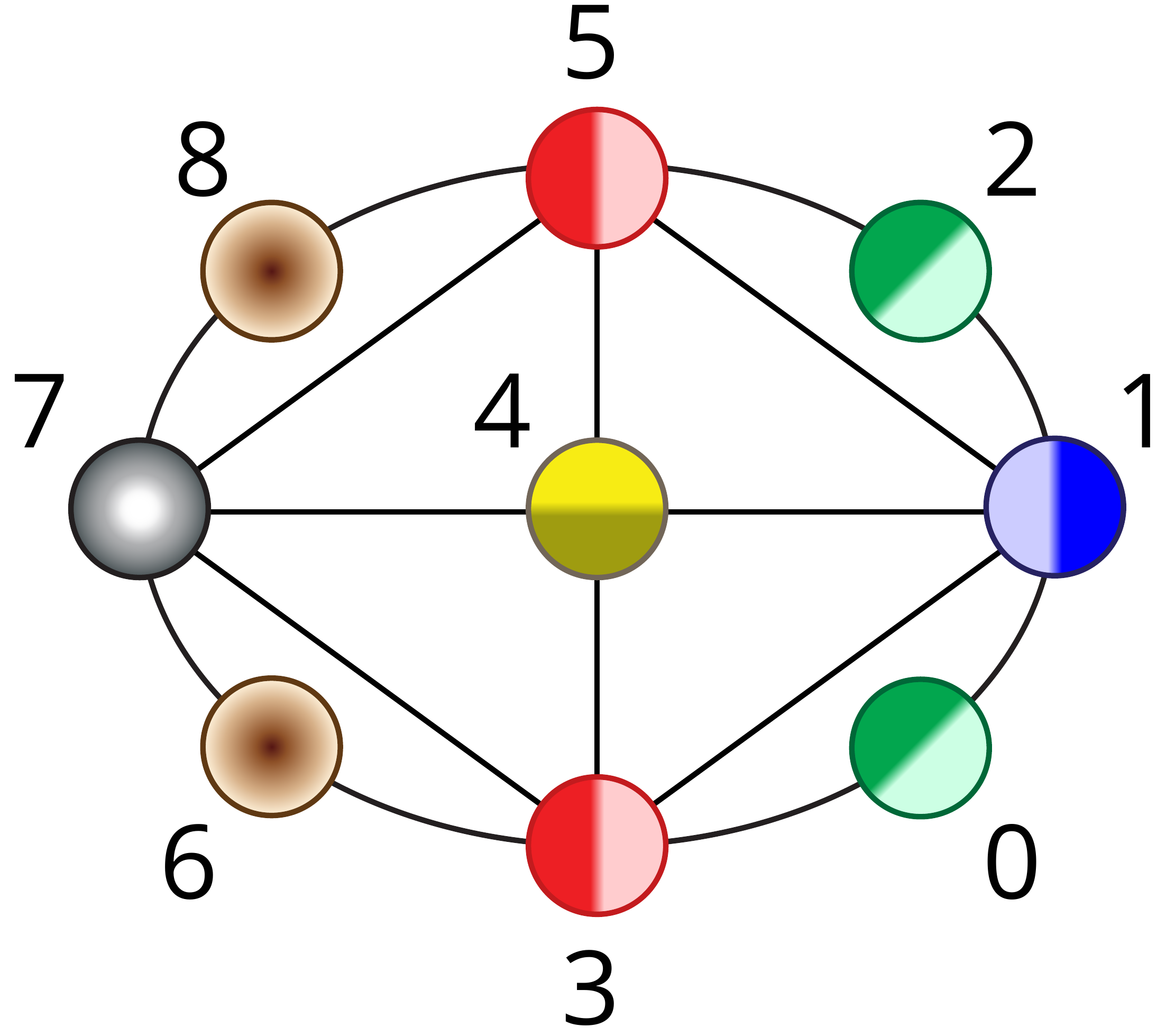}
\caption{Individualization of vertex 1 and subsequent refinement\label{f:refine}}
\end{figure}

\medskip

While the function $R$ defined above is sufficient for many graphs,
there are difficult classes (see Section~\ref{s:performance}) for
which it does not adequately separate inequivalent vertices.
Regular graphs are the simplest example, since the colouring with
only one colour is equitable.
A simple way of doing better is to count the number of triangles
incident to each vertex.  In choosing such a strategy, there is a
trade-off between the partitioning power and the cost. 
{\nauty} provides a small library of
stronger partitioning functions, some of them designed for
particular classes of difficult graphs.
The improvement in performance can be very dramatic.
On the other hand, choice of which partitioning function to employ
is left to the user and requires skill, which is not very
satisfactory.

{\traces} has a different approach to this problem, as we will see
in Section~\ref{ss:invariants}.

\subsection{Target cell selection}\label{ss:targetcell}

The choice of target cell has a significant effect on the shape of the
search tree, and thus on performance.  A small target cell may 
perhaps have a greater chance
of being an orbit of the group which fixes the current stabilizer sequence.
For this reason, \citet{PGI1} recommended using the first smallest
non-singleton cell.  However, \citet{Kocay} found (without realizing
it) that using the first non-singleton cell regardless of size was
better for most test cases, as confirmed by~\citet{Kirk}.
The current version of {\nauty} has two strategies.
One is to use the first non-singleton cell, and the other is to choose
the first cell which is joined in a non-trivial fashion to the largest
number of cells, where a non-trivial join between two cells means
that there is more than 0 edges and less than the maximum possible.

{\traces}, on the other hand prefers large target cells, as they tend
to make the tree less deep.  A strategy developed by experiment is
to use the first largest non-singleton cell that is a subset of the
target cell in the parent node.  If there are no such non-singleton cells,
the target cell in the grandparent node is used, and so on, with the
first largest cell altogether being the last possibility.

\subsection{Node invariants}\label{ss:invariants}

Information useful for computing node invariants can come from two
related sources.  At each node $\nu$ there is a colouring $R(G,\pi_0,\nu)$ and
we can use properties of this colouring such as the number and size of the
cells, as well as combinatorial properties of the coloured graph.
Another source is the intermediate states of the computation of a colouring
from that of the parent node, such as the order, position and size of the cells
produced by the refinement procedure and various counts of edges that are
determined during the computation.

If $f(\nu)$ is some function of this information, computed during the computation
of $R(G,\pi_0,\nu)$ and from the resulting coloured graph, the vector
$\(f([\nu]_0),f([\nu]_1),\ldots,f(\nu)\)$, with lexicographic ordering, satisfies
Conditions ($\phi$1) and ($\phi$3) for a node invariant.
If $\nu$ is a leaf, we can append $G^\pi$, where $\pi$ is the discrete
colouring $R(G,\pi_0,\nu)$, to the vector so as to satisfy ($\phi$2) as well.

In {\nauty}, the value of $f(\nu)$ is an integer, and the pruning rules are applied
as each node is computed.
{\traces} introduced a major improvement, defining each $f(\nu)$ as a vector 
itself.  The primary components of $f(\nu)$ are the sizes and positions of the
cells in the order that they are created by the refinement procedure. 
$\phi(G,\pi_0,\nu)$ thus becomes a vector of vectors, called the \textit{trace}
(and hence the name ``{\traces}'').
The advantage is that it often enables the comparison of $f(\nu)$ and
$f(\nu')$ to be made while the computation of $\nu'$ is only partly complete.
A limited form of this idea appeared in {\bliss} \citep{Bliss1}, and also appears in 
a recent version of {\saucy}~\citep{saucy2}.
For many difficult graph families, only a fraction of all refinement operations
need to be completed.
A practical consequence is that the stronger refinements used by {\nauty}
(see Section~\ref{ss:refinement}) are rarely needed.  This makes good performance
in {\traces} less dependent on user expertise than is the case with {\nauty}.

If $\pi$ is an equitable colouring of a graph $G$, we can define a
the \textit{quotient graph} $Q(G,\pi)$ as follows.  The vertices of $Q(G,\pi)$
are the cells of $\pi$, labelled with the cell number and size.  For any two
cells $C_1,C_2\in\pi$, possibly equal, the corresponding vertices of $Q(G,\pi)$
are joined by an edge labelled with the number of edges of $G$ between
$C_1$ and $C_2$.

The node invariant $\phi(G,\pi_0,\nu)$ computed by {\traces}, and also
by {\nauty} if the standard refinement process Algorithm~\ref{a:refinement}
is used, is a deterministic function of the sequence of quotient graphs
$Q\(G,R(G,\pi_0,[\nu]_i)\)$ for $i=0,\ldots,\card\nu$.
We could in fact use that sequence of quotient graphs, but that would be
expensive in both time and space.
Our experience is that the information we do use, which is essentially
information about the quotient matrices collected during the refinement
process, rarely has less pruning power than the 
quotient matrices themselves would have.

\begin{figure}[htp]
\centering
\includegraphics[scale=0.62]{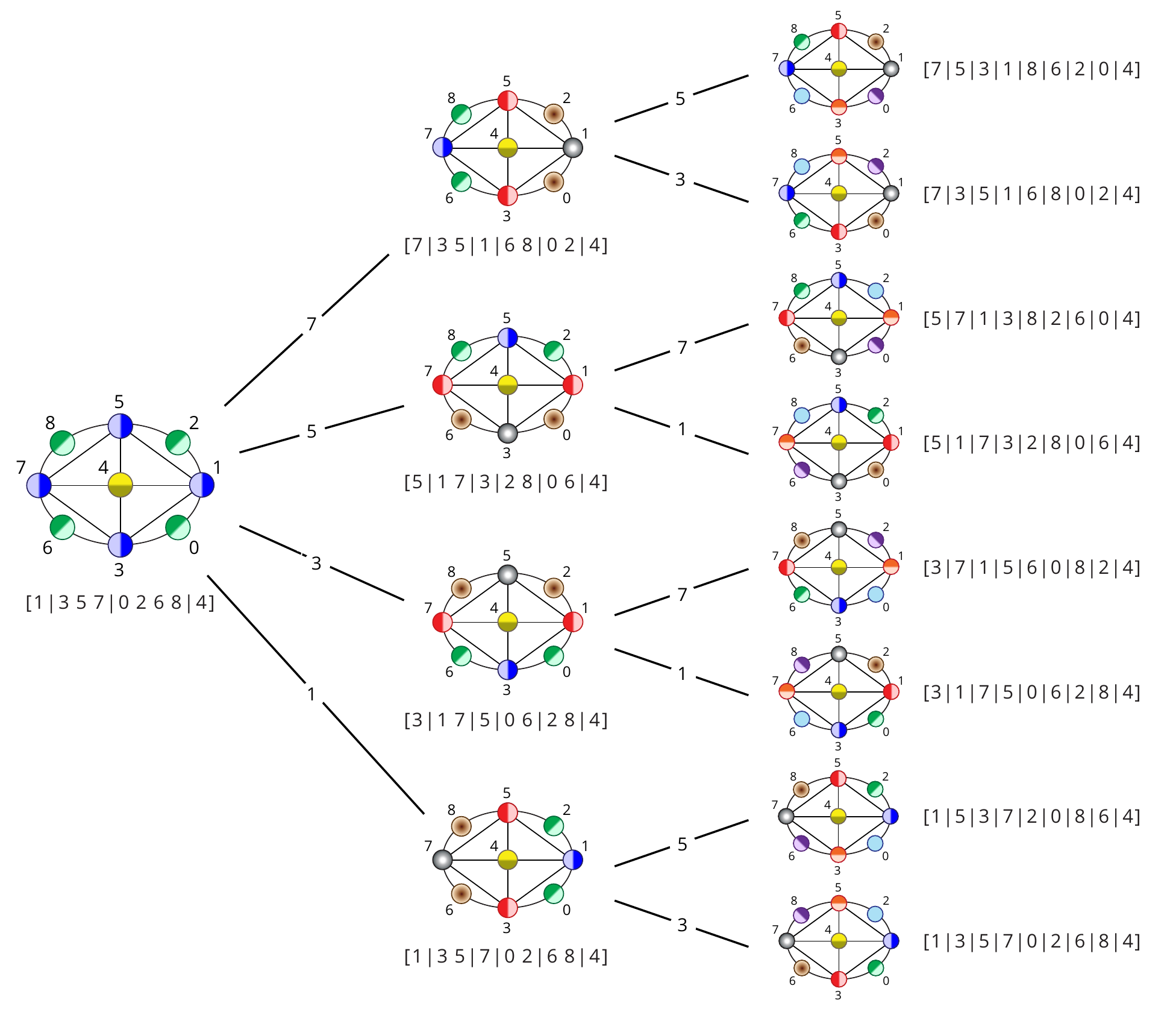}
\caption{Example of a search tree for the graph of
 Fig.~\ref{f:equitable}\label{f:Tree}}
\end{figure}

\subsection{Strategies for tree generation}\label{ss:strategies}

Now we have described the search tree $\T(g,\pi_0)$ as defined by
{\nauty} and {\traces}.
In general only a fraction of the search tree is actually generated, since the
pruning rules of Section~\ref{ss:automorphisms} are applied.  These pruning
rules utilise both node invariants, as described in Section~\ref{ss:invariants},
and automorphisms, which are mainly discovered by noticing that two discrete
colourings give the same coloured graph.  Now we will describe order of
generation of the tree, which is fundamentally different for {\nauty}
and {\traces}.

\medskip

In {\nauty}, the tree is generated in depth-first order.
The lexicographically least leaf $\nu_1$ is kept.
If the canonical labelling is sought (rather than just the automorphism group),
the leaf $\nu^*$ with the greatest invariant discovered so far is also kept.
A non-leaf node $\nu$ is pruned if neither
$\phi(G,\pi_0,\nu)=\phi(G,\pi_0,[\nu_1]_{\card{\nu}})$
or $\phi(G,\pi_0,\nu)\ge\phi(G,\pi_0,[\nu^*]_{\card{\nu}})$.
Such operations have both type $P_A(\nu^*,\nu)$ and
$P_B(\nu_1,\nu)$, so Theorem~\ref{t:pruning} applies.
Automorphisms are found by discovering leaves equivalent to $\nu_1$ or
$\nu^*$, and also to a limited extent from the properties of equitable
colourings.
Pruning operation $P_C$ is performed wherever possible, as high in the
tree as possible (i.e., at the children of the nearest common ancestor
of the two leaves found to be equivalent).

Until a recent version of {\nauty}, the only automorphisms used for
pruning operation $P_C$ were those directly discovered, without any
attempt to compose them.
Now we use the random Schreier method~\citep{schreier} to perform more
complete pruning.  By Lemma~\ref{l:stabilizer}, nodes $\nu\cat v_1$
and $\nu\cat v_2$ are equivalent if $v_1,v_2$ belong to the same orbit of
the point-wise stabiliser of $\nu$ in $\varGamma$, where $\varGamma$ 
is the group generated by the automorphisms found so far.
This stabiliser could be computed with a deterministic algorithm
as proposed by \citet{Kocay} and \citet{Butler}, but we have
found the random Schreier method \citep{schreier} to be more efficient
and it doesn't matter if occasionally (due to its probabilistic nature)
it computes smaller orbits.
The usefulness of this for {\nauty}'s efficiency with
some classes of difficult graph was demonstrated in 1985 by
\citet{Kirk} but only made it into the distributed edition of
{\nauty} in 2011.

\begin{figure}[ht]
\centering
\includegraphics[scale=0.98]{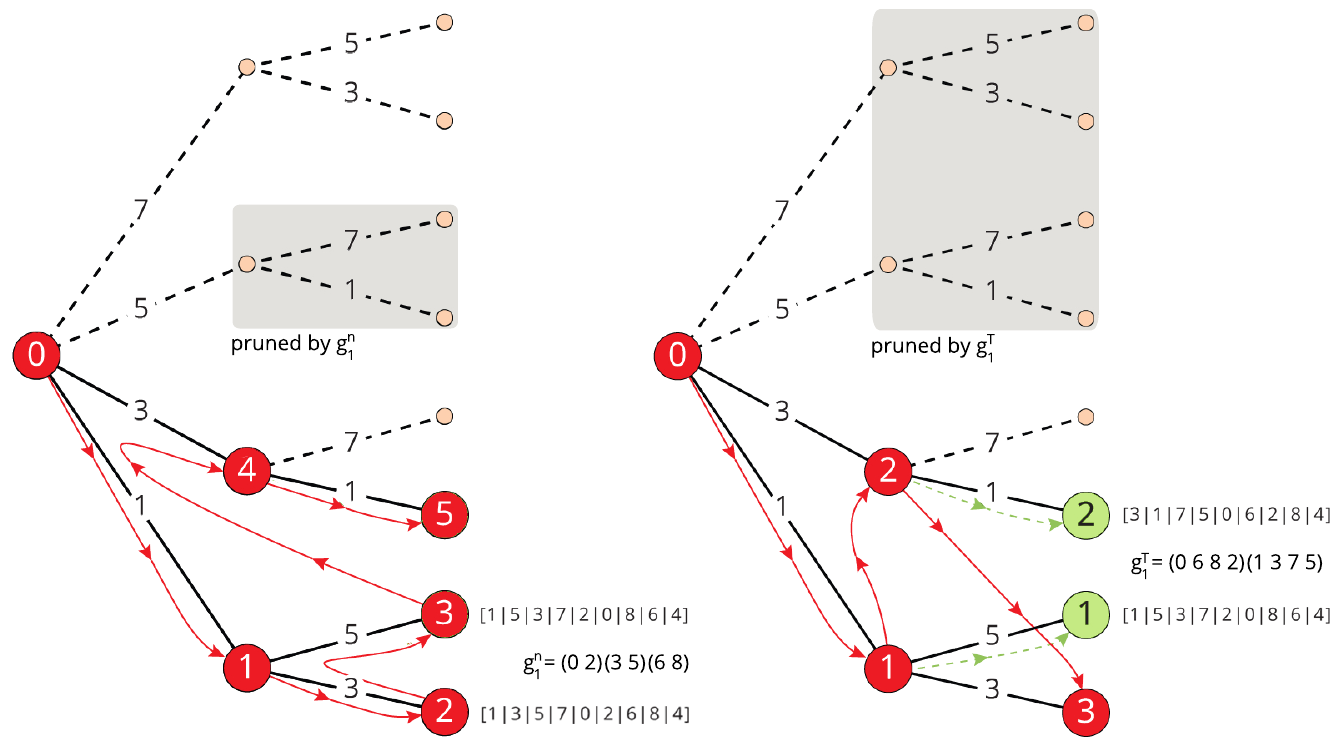}
\caption{Search tree order for {\nauty} (left) and {\traces} (right) \label{f:order}}
\end{figure}

\smallskip
{\Nnauty}'s basic depth-first approach is also followed by {\bliss} and
{\saucy}. However, {\traces} introduces an entirely different order of
generating the tree.
Some variations are possible but we will first describe the normative
method, which is based on a breadth-first search.
Define \textit{level}~$k$ to be the set of nodes $\nu$ with $\card\nu=k$.
In the $k$-th phase, {\traces} computes those nodes $\nu$ in level~$k$
which have the greatest value of $\phi(G,\pi_0,\nu)$ on that level.
By property ($\phi$1), such nodes are the children of the nodes
with greatest $\phi$ on the previous level, so no backtracking is needed.
This order of tree generation has the big advantage that pruning operation
$P_A$ is used to the maximum possible extent.

As mentioned in Section~\ref{ss:invariants}, the node invariant
$\phi(G,\pi_0,\nu)$ is computed incrementally during the refinement
process, so that pruning operation $P_A$ can often be applied when
the refinement is only partly complete.

An apparent disadvantage of breadth-first order
is that pruning by automorphisms (operation $P_C$) is
only possible when automorphisms are known, which in general requires
leaves of the tree.  To remedy this problem, for every node a single
path, called an ``experimental path'', is generated from that node down
to a leaf of the tree.  Automorphisms are found by comparing the
labelled graphs that correspond to those leaves, with the value of
$\phi(G,\pi_0,\nu)$ at the leaf being used to avoid most unnecessary
comparisons.
We have found experimentally that generating experimental paths randomly
tends to find automorphisms that generate larger subgroups, so that the
group requires fewer generators altogether and more of the group is
available early for pruning. 

The group generated by the automorphisms found so far is maintained
using the random Schreier method. Some features of the Schreier method
are turned on and off in {\traces} when it is possible to 
heuristically infer their computational weight.

Figure \ref{f:order} 
continues the example of Figure \ref{f:Tree}, showing the portion of the 
search tree traversed by {\nauty} (left) and {\traces} (right).  Node labels indicate
the order in which nodes are visited, and edge labels indicate which 
vertex is individualized.
During its backtrack search, {\nauty}
stores the first leaf ($2$) for comparison with subsequent leaves.
Leaves $2$ and $3$ provide the generator $g_1^n = (0\,2)(3\,5)(6\,8)$,
which for example allows pruning of the greyed subtree formed by
individualizing vertex 5 at the root.
{\traces} executes a breadth-first search, storing with each visited node the
discrete partition obtained by a randomly chosen experimental path
(shown by green arrow).
After processing node $2$ of the tree, the experimental leaves $1$ and $2$ are
compared, revealing the generator $g_1^T = (0\,6\,8\,2) (1\,3\,7\,5)$, which
allows for pruning the greyed subtrees formed by individualizing
vertices $5$ and $7$ at the root.

\subsection{Detection of automorphisms}\label{ss:detection}

The primary way that automorphisms are detected, in all the
programs under consideration, is to compare the labelled graphs
corresponding to leaves of the search tree as described above.

An important innovation of {\saucy}~\citep{saucy2} was to detect
some types of automorphism higher in the tree.
Suppose that $\pi,\pi'$ are equitable colourings with the same
number of vertices of each colour.  Any automorphism of $(G,\pi_0)$
that takes $\pi$ onto $\pi'$ has known action on the fixed
vertices of $\pi$: it maps them to the fixed vertices of $\pi'$ with
the same colours.
In some cases that {\saucy} can detect very quickly, this 
partial mapping is an automorphism when extended as the
identity on the non-fixed vertices.
This happens, for example, when a component of $G$ is
completely fixed by two different but equivalent
stabilization sequences. This is one of the
main reasons {\saucy} can be very fast on graphs with many
automorphisms that move few vertices.

\begin{figure}[ht]
\centering
\includegraphics[scale=0.26]{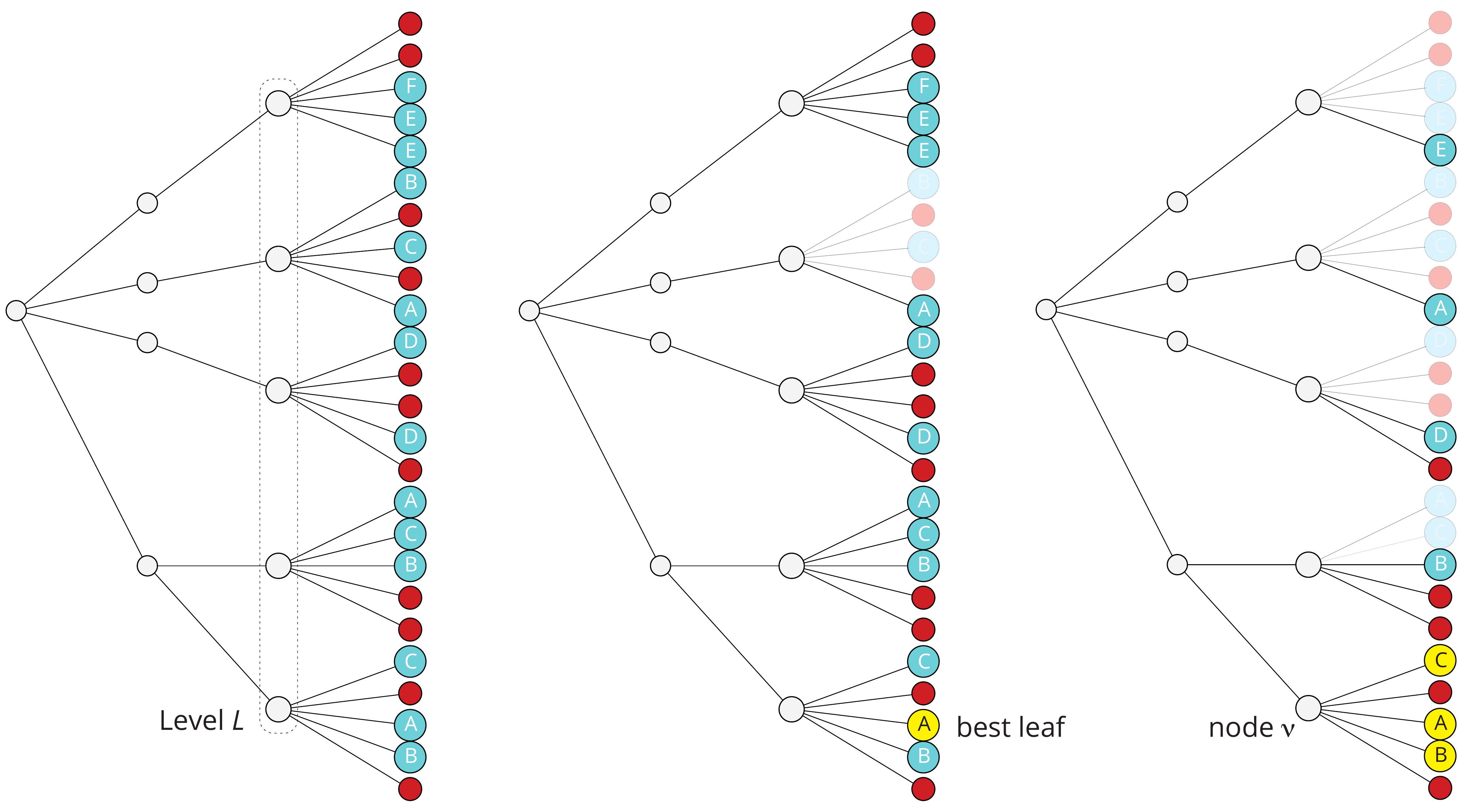}
\caption{{\traces} search strategies for canonical labelling or
  automorphism group\label{f:final}}
\end{figure}

{\traces} extends this idea by finding many automorphisms that do
not require the identity mapping on the non-trivial vertices.
It does this by a heuristic that extends the mapping from the fixed
vertices to the non-fixed vertices, which is applied in certain
situations where it is more likely to succeed.

When {\traces} is only looking for the automorphism group, and not
for a canonical labelling, it employs another strategy which is 
sometimes much faster.  Suppose that
while generating the nodes on some level $L$, it notices (during
experimental path generation) that one of them, say $\nu$, has a
child which is discrete.
At this point, {\traces} determines and keeps all the discrete children
of $\nu$ (modulo the usual automorphism pruning).  Now, for all nodes
$\nu'$ on level $L$, a single discrete child $\nu''$ is found, if any, and
an automorphism is discovered if it is equivalent to any child of $\nu$.
The validity of this approach follows from Theorem~\ref{t:pruning}
with the role of $\nu_0$ played by the first discrete child of~$\nu$.

Figure \ref{f:final} (left) shows the whole tree up to level $L{+}1$, where
a node labelled by $X$ represents a discrete partition corresponding
to labelled graph~$X$,
while an unlabelled (and smaller) node stands for a non-discrete partition.
Figure \ref{f:final} (center) shows the part of the tree which is traversed by
{\traces} during the search for a canonical labelling. Only the best leaf is
kept for comparison with subsequent discrete partitions.

Figure \ref{f:final} (right) shows the part of the tree which is traversed by
{\traces} during an automorphism group computation.
All the discrete children of $\nu$ are kept for comparison with subsequent
discrete partitions. When the first discrete partition is found as a child of
a node $\nu'$ at level $L$, either it has the same labelled graph as one
of those stored, or the whole subtree rooted at $\nu'$ has no
leaf with one of the stored graphs. In the first case, an automorphism
is found. In both cases, the computation is resumed from the next
node at level~$L$.

\subsection{Low degree vertices}\label{ss:sparse}

Graphs in some applications, such as constraint satisfaction problems described
by \citet{saucy1} have many small components with vertices of low
degree, vertices with common neighborhoods, and so on.
{\Ssaucy} handles them efficiently by a refinement procedure tuned to
this situation plus early detection of sparse automorphisms.
{\traces} employs another method. Recall that after the first refinement
vertices with equal colours also have equal degrees.
The target cell selector never selects cells containing vertices of
degree 0, 1, 2 or $n{-}1$, and nodes whose non-trivial cells are only
of those degrees are not expanded further.  Special-purpose code
then produces generators for the automorphism group fixed by
the node and, if necessary, a unique discrete colouring that refines
the node.

This technique is quite successful.  However, in our opinion, graphs
of this type ought to be handled by preprocessing.  For example, 
sets of vertices with the same neighborhoods ought to be replaced by 
single vertices with a colour that encodes the multiplicity.  All
tree-like appendages, long paths of degree~2 vertices, and similar
easy subgraphs, could be efficiently factored out in this manner.


\nicebreak
\section{Performance}\label{s:performance}

In the following figures, we present some comparisons between programs
for a variety of graphs ranging from very easy to very difficult.  We made
an effort to include graphs that are easy and difficult for each of the
programs tested.

Most of the graphs are taken from the {\bliss} collection, but for the record we
provide all of our test graphs at the {\nauty} and {\traces}
website~\citep{nautypage}.

The times given are for a Macbook Pro with 2.66 GHz Intel i7
processor, compiled using gcc 4.7 and running in a single thread.
Easy graphs were processed multiple times to give more precise times.
In order to avoid non-typical behaviour due to the input labelling, all the
graphs were randomly labelled before processing.
In some classes, such as the ``combinatorial graphs'', the processing
time can depend a lot on the initial labelling; the plots show
whatever happened in our tests.

\medskip
The following programs were included. Programs (c)--(e) reflect their
distributed versions at the end of October 2012.
\nicebreak
\begin{enumerate}
 \item[(a)] {\nauty} version 2.5
 \item[(b)] {\traces} version 2.0
 \item[(c)] {\saucy} version 3.0
 \item[(d)] {\bliss} version 7.2
 \item[(e)] {\conauto} version 2.0.1
\end{enumerate}
The first column of plots in each figure is for computation of the
automorphism group alone.  The second column is for computation 
of a canonical labelling, which for all the programs here includes
an automorphism group computation.

For {\nauty} we used the dense or sparse version consistently within each
class, depending on whether the class is inherently dense or sparse.
We did not use an invariant except where indicated, even though it
would often help.

{\Ssaucy} does not have a canonical labelling option.  Version
3.0, which was released just as this paper neared completion, has an
amalgam of {\saucy} and {\bliss} that can do canonical labelling, but
we have not tested it much.

{\Cconauto} features automorphism group computation and the ability
for testing two graphs for isomorphism.  We decided that the latter
is outside the scope of this study.
For the same reason we did not include the program of~\citet{Italians}
in our comparisons.

Another excellent program, that we were unfortunately unable to include
for technical reasons, is due to~\citet{Stoichev}.
Many more experiments and comments can be found at
\texttt{http://pallini.di.uniroma1.it}.

%

\nicebreak
\section{Conclusions}\label{s:conclusions}

We have brought the published description of {\nauty} up to date and
introduced the program {\traces}. 
In particular, we have shown that the highly innovative tree scanning algorithm
introduced by {\traces} can have a remarkable effect on the processing
power. 
Although none of the programs tested have the best performance on all
graph classes, it is clear that {\traces} is currently the leader on the majority
of difficult graph classes tested, while {\nauty} is still preferred for mass
testing of small graphs.
An exception is provided by some classes of graphs consisting of disjoint or
minimally-overlapping components, here represented by non-disjoint
unions of tripartite graphs.  Conauto and Bliss~\citep{Bliss2} have
special code for such graphs, but as yet {\nauty} and {\traces} do not.

We wish to thank Gordon Royle for many useful test graphs.  We also
thank the authors of {\saucy}, {\bliss} and {\conauto} for many useful discussions.
The second author is indebted to Riccardo Silvestri for his 
strong encouragement and valuable suggestions.

\nicebreak


\def\oneplot#1#2{\multicolumn{2}{l}{#1} \\[1ex]
   \input{fams/#2.tex} & \input{fams/#2table.tex}}
\def\wideplot#1#2{\multicolumn{2}{l}{#1} \\[1ex]
   \multicolumn{2}{l}{\input{fams/#2.tex}}}
\def\twoplots#1#2{\multicolumn{2}{l}{#1} \\
   \input{fams/#2.tex} & \input{fams/#2-can.tex}}
\def\picScale{.72}


\begin{figure}[p]
\begin{center}
  \begin{tabular}{@{} cc @{}}
    Automorphism group & Canonical label \\
    \hline
   \multicolumn{2}{l}{Random graphs with $p=\tfrac12$} \\
   \input{fams/ran2.tex} & \input{fams/ran2-can.tex} \\ 
   \multicolumn{2}{l}{Random graphs with $p=n^{-1/2}$} \\
   \input{fams/ransqrt.tex} & \input{fams/ransqrt-can.tex} \\ 
   \multicolumn{2}{l}{Random cubic graphs ({\nauty} invariant \textit{distances}(2))} \\
   \input{fams/ran3reg.tex} & \input{fams/ran3reg-can.tex} \\ 
    \hline
   \multicolumn{2}{l}{\showBliss\hspace{5pt}{\bliss}\hspace{10pt}\showSaucy\hspace{5pt}{\saucy}\hspace{10pt}
   \showConauto\hspace{5pt}{\conauto}\hspace{10pt}\showNauty\hspace{5pt}{\nauty}\hspace{10pt}\showNautyLight
   \hspace{5pt}{\nauty} with invariant\hspace{5pt}\showTraces\hspace{5pt}{\traces}} \\
  \end{tabular}
\end{center}
\caption{Performance comparison
  (horizontal: number of vertices; vertical: time in seconds)}
\label{perf1}
\end{figure}


\begin{figure}[p]
\begin{center}
  \begin{tabular}{@{} cc @{}}
    Automorphism group & Canonical label \\
    \hline
   \multicolumn{2}{l}{Hypercubes (vertex-transitive)} \\
   \input{fams/hypercubes.tex} & \input{fams/hypercubes-can.tex} \\ 
   \multicolumn{2}{l}{Miscellaneous vertex-transitive graphs} \\
   \input{fams/tran.tex} & \input{fams/tran-can.tex} \\ 
   \multicolumn{2}{l}{(Non-disjoint) union of tripartite graphs} \\
   \input{fams/tnn.tex} & \input{fams/tnn-can.tex} \\ 
    \hline
   \multicolumn{2}{l}{\showBliss\hspace{5pt}{\bliss}\hspace{10pt}\showSaucy\hspace{5pt}{\saucy}\hspace{10pt}
   \showConauto\hspace{5pt}{\conauto}\hspace{10pt}\showNauty\hspace{5pt}{\nauty}\hspace{10pt}\showTraces\hspace{5pt}{\traces}} \\
  \end{tabular}
\end{center}
\caption{Performance comparison
  (horizontal: number of vertices; vertical: time in seconds)}
\label{perf2}
\end{figure}


\begin{figure}[p]
\begin{center}
  \begin{tabular}{@{} cc @{}}
    Automorphism group & Canonical label \\
    \hline
   \multicolumn{2}{l}{Small strongly-regular graphs} \\
   \input{fams/srg-small.tex} & \input{fams/srg-small-can.tex} \\ 
   \multicolumn{2}{l}{Large strongly-regular graphs} \\
   \input{fams/srg.tex} & \input{fams/srg-can.tex} \\ 
   \multicolumn{2}{l}{Hadamard matrix graphs} \\
   \input{fams/had.tex} & \input{fams/had-can.tex} \\ 
    \hline
   \multicolumn{2}{l}{\showBliss\hspace{5pt}{\bliss}\hspace{10pt}\showSaucy\hspace{5pt}{\saucy}\hspace{10pt}
   \showConauto\hspace{5pt}{\conauto}\hspace{10pt}\showNauty\hspace{5pt}{\nauty}\hspace{10pt}\showTraces\hspace{5pt}{\traces}} \\
  \end{tabular}
\end{center}
\caption{Performance comparison
  (horizontal: number of vertices; vertical: time in seconds)}
\label{perf3}
\end{figure}


\begin{figure}[p]
\begin{center}
  \begin{tabular}{@{} cc @{}}
    Automorphism group & Canonical label \\
    \hline
   \multicolumn{2}{l}{Random trees} \\
   \input{fams/rantree.tex} & \input{fams/rantree-can.tex} \\ 
   \multicolumn{2}{l}{Cai-F\"urer-Immerman graphs} \\
   \input{fams/cfi.tex} & \input{fams/cfi-can.tex} \\
   \multicolumn{2}{l}{Miyazaki graphs} \\
   \input{fams/mz-aug2.tex} & \input{fams/mz-aug2-can.tex} \\ 
    \hline
   \multicolumn{2}{l}{\showBliss\hspace{5pt}{\bliss}\hspace{10pt}\showSaucy\hspace{5pt}{\saucy}\hspace{10pt}
   \showConauto\hspace{5pt}{\conauto}\hspace{10pt}\showNauty\hspace{5pt}{\nauty}\hspace{10pt}\showTraces\hspace{5pt}{\traces}} \\
  \end{tabular}
\end{center}
\caption{Performance comparison
  (horizontal: number of vertices; vertical: time in seconds)}
\label{perf4}
\end{figure}


\begin{figure}[p]
\begin{center}
  \begin{tabular}{@{} cc @{}}
   \multicolumn{2}{l}{\parbox{0.9\textwidth}{Automorphisms groups of projective planes of order 16\\
     (regular bipartite graphs of order 546 and degree 17)}} \\[1ex]
   \input{fams/planes.tex} & \input{fams/planestable.tex} \\
   \multicolumn{2}{l}{Automorphisms of some combinatorial graphs} \\[1ex]
   \input{fams/comb.tex} & \input{fams/combtable.tex} \\ 
   \multicolumn{2}{l}{Canonical labelling of the above graphs} \\[1ex]
   \multicolumn{2}{l}{\input{fams/planescomb-can.tex}} \\
    \hline
   \multicolumn{2}{l}{\showBliss\hspace{5pt}{\bliss}\hspace{10pt}\showSaucy\hspace{5pt}{\saucy}\hspace{10pt}
   \showConauto\hspace{5pt}{\conauto}\hspace{10pt}\showNauty\hspace{5pt}{\nauty}\hspace{10pt}\showNautyLight
   \hspace{5pt}{\nauty} with invariant \textit{cellfano2}\hspace{8pt}\showTraces\hspace{5pt}{\traces}} \\
  \end{tabular}
\end{center}
\caption{Performance comparison
  (horizontal: graph number; vertical: time in seconds)}
\label{perf5}
\end{figure}

\end{document}

%% file: plotsmacros.tex
\newcommand{\Bliss}{\addplot[mark=triangle*, mark options={scale=1.4, draw=red!80!black, fill=red, opacity=1},draw=red!80!black]}
\newcommand{\Saucy}{\addplot[mark=diamond*, mark options={scale=1.2, draw=green!50!black,  fill=green,opacity=1},draw=green!50!black]}
\newcommand{\Conauto}{\addplot[mark=oplus*, mark options={scale=1.4, very thin, draw=yellow!50!black, fill=yellow!50!white,opacity=1},draw=yellow!50!black]}
\newcommand{\Nauty}{\addplot[mark=*, mark options={scale=1.4, very thin, draw=black, fill=blue, opacity=1}, draw=blue!80!black]}
\newcommand{\NautyLight}{\addplot[mark=*, mark options={scale=1.4, very thin, draw=black, fill=black!60!white, opacity=1}, draw=black!80!white]}
\newcommand{\Traces}{\addplot[mark=square*, mark options={scale=1, thin, draw=cyan!50!black,fill=cyan!50!white, opacity=1},cyan!50!white, opacity=1]}

\newcommand{\BlissOnlyMarks}{\addplot[mark=triangle*, mark options={scale=1.4, draw=red!80!black, fill=red, opacity=1},only marks]}
\newcommand{\SaucyOnlyMarks}{\addplot[mark=diamond*, mark options={scale=1.2, draw=green!50!black,  fill=green,opacity=1},only marks]}
\newcommand{\ConautoOnlyMarks}{\addplot[mark=oplus*, mark options={scale=1.4, very thin, draw=yellow!50!black, fill=yellow!50!white,opacity=1},only marks]}
\newcommand{\NautyOnlyMarks}{\addplot[mark=*, mark options={scale=1.4, very thin, draw=black, fill=blue, opacity=1},only marks]}
\newcommand{\NautyLightOnlyMarks}{\addplot[mark=*, mark options={scale=1.4, very thin, draw=black, fill=black!60!white, opacity=1},only marks]}
\newcommand{\TracesOnlyMarks}{\addplot[mark=square*, mark options={scale=1, thin, draw=cyan!50!black,fill=cyan!50!white, opacity=1},only marks]}

\newcommand{\BlissMed}{\addplot[mark=triangle*, mark options={scale=1.2, draw=red!80!black, fill=red, opacity=1},only marks]}
\newcommand{\SaucyMed}{\addplot[mark=diamond*, mark options={scale=.9, draw=green!50!black,  fill=green,opacity=1},only marks]}
\newcommand{\ConautoMed}{\addplot[mark=oplus*, mark options={scale=1.2, very thin, draw=yellow!50!black, fill=yellow!50!white,opacity=1},only marks]}
\newcommand{\NautyMed}{\addplot[mark=*, mark options={scale=1.2, very thin, draw=black, fill=blue, opacity=1},only marks]}
\newcommand{\TracesMed}{\addplot[mark=square*, mark options={scale=.75, thin, draw=cyan!50!black,fill=cyan!50!white, opacity=1},only marks]}

\newcommand{\BlissSmall}{\addplot[mark=triangle*, mark options={scale=.8, draw=red!80!black, fill=red, opacity=1},draw=red!80!black]}
\newcommand{\SaucySmall}{\addplot[mark=diamond*, mark options={scale=.6, draw=green!50!black, very thin, fill=green,opacity=1},draw=green!50!black]}
\newcommand{\ConautoSmall}{\addplot[mark=oplus*, mark options={scale=.8, very thin, draw=yellow!50!black, fill=yellow!50!white,opacity=1},draw=yellow!50!black]}
\newcommand{\NautySmall}{\addplot[mark=*, mark options={scale=.8, very thin, draw=black, fill=blue, opacity=1}, draw=blue!80!black]}
\newcommand{\TracesSmall}{\addplot[mark=square*, mark options={scale=.5, thin, draw=cyan!50!black,fill=cyan!50!white, opacity=1},cyan!50!white, opacity=1]}

\newcommand{\SaucyThree}{\addplot[mark=diamond*, mark options={scale=.75, draw=green!50!black,  fill=green,opacity=1},draw=green!50!black]}
\newcommand{\ConautoThree}{\addplot[mark=oplus*, mark options={scale=.8, very thin, draw=yellow!50!black, fill=yellow!50!white,opacity=1},draw=yellow!50!black]}

\newcommand{\timeoutbox}[1]{timeout: $#1$ secs}

\def\showBliss{%
	\tikz\draw[%
		mark options={scale=1.4, draw=red!80!black, fill=red, opacity=1},
		mark=triangle*]
	plot coordinates {(0,0)};%
}%
\def\showSaucy{%
	\tikz\draw[%
		mark options={scale=1.2, draw=green!50!black,  fill=green,opacity=1},
		x=0.8cm,y=0.3cm,
		mark=diamond*]
	plot coordinates {(0,0)};%
}%
\def\showConauto{%
	\tikz\draw[%
		mark options={scale=1.4, very thin, draw=yellow!50!black, fill=yellow!50!white,opacity=1},
		x=0.8cm,y=0.3cm,
		mark=oplus*]
	plot coordinates {(0,0)};%
}%
\def\showNauty{%
	\tikz\draw[%
		mark options={scale=1.4, very thin, draw=black, fill=blue, opacity=1},
		x=0.8cm,y=0.3cm,
		mark=*]
	plot coordinates {(0,0)};%
}%
\def\showNautyLight{%
	\tikz\draw[%
		mark options={scale=1.4, very thin, draw=black!80!white, fill=black!60!white, opacity=1},
		x=0.8cm,y=0.3cm,
		mark=*]
	plot coordinates {(0,0)};%
}%
\def\showTraces{%
	\tikz\draw[%
		mark options={scale=1, thin, draw=cyan!50!black,fill=cyan!50!white, opacity=1},
		x=0.8cm,y=0.3cm,
		mark=square*]
	plot coordinates {(0,0)};%
}%

%% file: fams/ran2.tex
\begin{tikzpicture}[transform shape, scale=\picScale]
\def \timeout{120}
\def \xmin{10}
\def \xmax{5000}
\def \ymin{-0}
\def \ymax{0.12}
\begin{loglogaxis}[
width=9cm,
height=8cm,
xmin=\xmin,  xmax=\xmax, 
ymin=\ymin,  ymax=\ymax, 
enlargelimits=0.01, 
grid=major]
\Bliss coordinates {
(10,0.0000037)
(15,0.0000046)
(20,0.0000053)
(25,0.0000064)
(30,0.0000075)
(40,0.0000107)
(100,0.0000395)
(200,0.0001298)
(300,0.0002769)
(400,0.0004742)
(1000,0.0029087)
(2000,0.0117248)
(3000,0.0260897)
(4000,0.0491325)
(5000,0.0799354)};
\Nauty coordinates {
(10,0.0000006)
(15,0.0000008)
(20,0.0000012)
(25,0.0000015)
(30,0.0000020)
(40,0.0000030)
(100,0.0000115)
(200,0.0000359)
(300,0.0000767)
(400,0.0001390)
(1000,0.0007769)
(2000,0.0028965)
(3000,0.0063508)
(4000,0.0112401)
(5000,0.0175833)};
\Conauto coordinates {
(10,0.0000073)
(15,0.0000092)
(20,0.0000091)
(25,0.0000104)
(30,0.0000122)
(40,0.0000133)
(100,0.0000271)
(200,0.0000564)
(300,0.0000882)
(400,0.0001246)
(1000,0.0003771)
(2000,0.0007484)
(3000,0.0010640)
(4000,0.0016105)
(5000,0.0020418)};
\Saucy coordinates {
(10,0.0000012)
(15,0.0000015)
(20,0.0000020)
(25,0.0000024)
(30,0.0000030)
(40,0.0000044)
(100,0.0000194)
(200,0.0000584)
(300,0.0001197)
(400,0.0002012)
(1000,0.0012355)
(2000,0.0049921)
(3000,0.0105150)
(4000,0.0184980)
(5000,0.0285488)};
\Traces coordinates {
(10,0.0000060)
(15,0.0000062)
(20,0.0000065)
(25,0.0000065)
(30,0.0000067)
(40,0.0000069)
(100,0.0000093)
(200,0.0000136)
(300,0.0000182)
(400,0.0000234)
(1000,0.0000768)
(2000,0.0002024)
(3000,0.0003121)
(4000,0.0004360)
(5000,0.0005605)};
\end{loglogaxis} 
\end{tikzpicture}%

%% file: fams/ran2-can.tex
\begin{tikzpicture}[transform shape, scale=\picScale]
\def \timeout{120}
\def \xmin{10}
\def \xmax{5000}
\def \ymin{-0}
\def \ymax{1}
\begin{loglogaxis}[
width=9cm,
height=8cm,
xmin=\xmin,  xmax=\xmax, 
ymin=\ymin,  ymax=\ymax, 
enlargelimits=0.01, 
grid=major]
\Bliss coordinates {
(10,0.0000076)
(15,0.0000115)
(20,0.0000163)
(25,0.0000221)
(30,0.0000283)
(40,0.0000427)
(100,0.0002223)
(200,0.0008247)
(300,0.0018637)
(400,0.0032002)
(1000,0.0205349)
(2000,0.0877391)
(3000,0.2100000)
(4000,0.3813333)
(5000,0.6235000)};
\Nauty coordinates {
(10,0.0000007)
(15,0.0000010)
(20,0.0000014)
(25,0.0000018)
(30,0.0000024)
(40,0.0000037)
(100,0.0000153)
(200,0.0000493)
(300,0.0001065)
(400,0.0001921)
(1000,0.0011152)
(2000,0.0045197)
(3000,0.0097500)
(4000,0.0176512)
(5000,0.0277333)};
\Traces coordinates {
(10,0.0000061)
(15,0.0000063)
(20,0.0000067)
(25,0.0000067)
(30,0.0000071)
(40,0.0000076)
(100,0.0000129)
(200,0.0000264)
(300,0.0000478)
(400,0.0000776)
(1000,0.0004008)
(2000,0.0017592)
(3000,0.0037853)
(4000,0.0065547)
(5000,0.0102600)};
\end{loglogaxis} 
\end{tikzpicture}%

%% file: fams/ransqrt.tex
\begin{tikzpicture}[transform shape, scale=\picScale]
\def \timeout{120}
\def \xmin{10}
\def \xmax{10000}
\def \ymin{-0}
\def \ymax{.12}
\begin{loglogaxis}[
width=9cm,
height=8cm,
xmin=\xmin,  xmax=\xmax, 
ymin=\ymin,  ymax=\ymax, 
enlargelimits=0.01, 
grid=major]
\Bliss coordinates {
(10,0.0000041)
(15,0.0000040)
(20,0.0000046)
(25,0.0000052)
(30,0.0000057)
(40,0.0000070)
(50,0.0000084)
(60,0.0000097)
(70,0.0000110)
(80,0.0000129)
(90,0.0000149)
(100,0.0000161)
(200,0.0000363)
(300,0.0000618)
(400,0.0000894)
(500,0.0001192)
(600,0.0001505)
(700,0.0001849)
(800,0.0002200)
(900,0.0002590)
(1000,0.0003030)
(2000,0.0007724)
(3000,0.0013892)
(4000,0.0020700)
(5000,0.0028820)
(6000,0.0038111)
(7000,0.0048903)
(8000,0.0061051)
(9000,0.0075728)
(10000,0.0089862)};
\Nauty coordinates {
(10,0.0000006)
(15,0.0000007)
(20,0.0000010)
(25,0.0000011)
(30,0.0000014)
(40,0.0000020)
(50,0.0000025)
(60,0.0000031)
(70,0.0000036)
(80,0.0000044)
(90,0.0000050)
(100,0.0000057)
(200,0.0000133)
(300,0.0000218)
(400,0.0000379)
(500,0.0000526)
(600,0.0000713)
(700,0.0000900)
(800,0.0001111)
(900,0.0001292)
(1000,0.0001471)
(2000,0.0003573)
(3000,0.0006111)
(4000,0.0008949)
(5000,0.0012173)
(6000,0.0015740)
(7000,0.0019644)
(8000,0.0023836)
(9000,0.0028629)
(10000,0.0034121)};
\Conauto coordinates {
(10,0.0000069)
(15,0.0000084)
(20,0.0000096)
(25,0.0000152)
(30,0.0000145)
(40,0.0000202)
(50,0.0000269)
(60,0.0000324)
(70,0.0000382)
(80,0.0000474)
(90,0.0000571)
(100,0.0000666)
(200,0.0001699)
(300,0.0003068)
(400,0.0004944)
(500,0.0007596)
(600,0.0009267)
(700,0.0014197)
(800,0.0015075)
(900,0.0019851)
(1000,0.0025038)
(2000,0.0077347)
(3000,0.0167246)
(4000,0.0269219)
(5000,0.0475250)
(6000,0.0546200)
(7000,0.0689346)
(8000,0.0851359)
(9000,0.1169632)
(10000,0.1189592)};
\Saucy coordinates {
(10,0.0000012)
(15,0.0000014)
(20,0.0000018)
(25,0.0000020)
(30,0.0000025)
(40,0.0000032)
(50,0.0000039)
(60,0.0000046)
(70,0.0000053)
(80,0.0000061)
(90,0.0000068)
(100,0.0000081)
(200,0.0000174)
(300,0.0000302)
(400,0.0000437)
(500,0.0000578)
(600,0.0000747)
(700,0.0000923)
(800,0.0001092)
(900,0.0001299)
(1000,0.0001517)
(2000,0.0003802)
(3000,0.0006468)
(4000,0.0009881)
(5000,0.0013389)
(6000,0.0018127)
(7000,0.0023710)
(8000,0.0028934)
(9000,0.0035808)
(10000,0.0047164)};
\Traces coordinates {
(10,0.0000065)
(15,0.0000062)
(20,0.0000066)
(25,0.0000066)
(30,0.0000068)
(40,0.0000070)
(50,0.0000075)
(60,0.0000078)
(70,0.0000077)
(80,0.0000088)
(90,0.0000091)
(100,0.0000089)
(200,0.0000126)
(300,0.0000161)
(400,0.0000204)
(500,0.0000254)
(600,0.0000304)
(700,0.0000330)
(800,0.0000386)
(900,0.0000463)
(1000,0.0000556)
(2000,0.0001508)
(3000,0.0002322)
(4000,0.0003110)
(5000,0.0004167)
(6000,0.0004836)
(7000,0.0006017)
(8000,0.0006809)
(9000,0.0007934)
(10000,0.0008652)};
\end{loglogaxis} 
\end{tikzpicture}%

%% file: fams/ransqrt-can.tex
\begin{tikzpicture}[transform shape, scale=\picScale]
\def \timeout{120}
\def \xmin{10}
\def \xmax{10000}
\def \ymin{-0}
\def \ymax{.1}
\begin{loglogaxis}[
width=9cm,
height=8cm,
xmin=\xmin,  xmax=\xmax, 
ymin=\ymin,  ymax=\ymax, 
enlargelimits=0.01, 
grid=major]
\Bliss coordinates {
(10,0.0000075)
(15,0.0000094)
(20,0.0000119)
(25,0.0000148)
(30,0.0000176)
(40,0.0000243)
(50,0.0000304)
(60,0.0000371)
(70,0.0000469)
(80,0.0000545)
(90,0.0000616)
(100,0.0000701)
(200,0.0001803)
(300,0.0003210)
(400,0.0004727)
(500,0.0006403)
(600,0.0007963)
(700,0.0009980)
(800,0.0012142)
(900,0.0013667)
(1000,0.0016008)
(2000,0.0043866)
(3000,0.0073695)
(4000,0.0113022)
(5000,0.0172054)
(6000,0.0222372)
(7000,0.0281808)
(8000,0.0361374)
(9000,0.0400424)
(10000,0.0466586)};
\Nauty coordinates {
(10,0.0000007)
(15,0.0000008)
(20,0.0000011)
(25,0.0000013)
(30,0.0000016)
(40,0.0000023)
(50,0.0000031)
(60,0.0000035)
(70,0.0000042)
(80,0.0000051)
(90,0.0000058)
(100,0.0000066)
(200,0.0000155)
(300,0.0000266)
(400,0.0000449)
(500,0.0000627)
(600,0.0000840)
(700,0.0001056)
(800,0.0001295)
(900,0.0001505)
(1000,0.0001754)
(2000,0.0004467)
(3000,0.0007628)
(4000,0.0011202)
(5000,0.0015242)
(6000,0.0019705)
(7000,0.0024976)
(8000,0.0030868)
(9000,0.0037921)
(10000,0.0045776)};
\Traces coordinates {
(10,0.0000065)
(15,0.0000064)
(20,0.0000067)
(25,0.0000067)
(30,0.0000070)
(40,0.0000075)
(50,0.0000078)
(60,0.0000083)
(70,0.0000083)
(80,0.0000094)
(90,0.0000098)
(100,0.0000099)
(200,0.0000148)
(300,0.0000204)
(400,0.0000267)
(500,0.0000348)
(600,0.0000442)
(700,0.0000480)
(800,0.0000580)
(900,0.0000700)
(1000,0.0000839)
(2000,0.0002375)
(3000,0.0003846)
(4000,0.0005442)
(5000,0.0007591)
(6000,0.0009783)
(7000,0.0013531)
(8000,0.0017522)
(9000,0.0022191)
(10000,0.0026765)};
\end{loglogaxis} 
\end{tikzpicture}%

%% file: fams/ran3reg.tex
\begin{tikzpicture}[transform shape, scale=\picScale]
\def \timeout{200}
\def \xmin{1000}
\def \xmax{10000}
\def \ymin{-0}
\def \ymax{140}
\begin{semilogyaxis}[
width=9cm,
height=8cm,
xmin=\xmin,  xmax=\xmax, 
ymin=\ymin,  ymax=\ymax, 
enlargelimits=0.01, 
grid=major,
scaled x ticks = false]
\Bliss coordinates {
(1000,0.0099299)
(2000,0.0351648)
(3000,0.1029709)
(4000,0.1779351)
(5000,0.2720194)
(6000,0.3554646)
(7000,0.5053737)
(8000,0.7374459)
(9000,0.7722576)
(10000,0.8881061)};
\Nauty coordinates {
(1000,0.1715152)
(2000,0.9878788)
(3000,2.9054545)
(4000,6.5836364)
(5000,12.1945455)
(6000,20.3290909)
(7000,33.0054545)
(8000,50.1045455)
(9000,71.0463636)
(10000,97.3918182)};
\NautyLight coordinates {
(1000,0.0004479)
(2000,0.0013073)
(3000,0.0024964)
(4000,0.0039897)
(5000,0.0058162)
(6000,0.0079447)
(7000,0.0106326)
(8000,0.0133442)
(9000,0.0166080)
(10000,0.0201111)};
\Conauto coordinates {
(1000,0.1064396)
(2000,0.5396667)
(3000,1.7637879)
(4000,3.3090909)
(5000,5.1009091)
(6000,7.1072727)
(7000,10.4645455)
(8000,17.2763636)
(9000,17.9672727)
(10000,20.3190909)};
\Saucy coordinates {
(1000,0.0026488)
(2000,0.0072101)
(3000,0.0143042)
(4000,0.0227933)
(5000,0.0367325)
(6000,0.0389577)
(7000,0.0671063)
(8000,0.0586646)
(9000,0.0850131)
(10000,0.1107858)};
\Traces coordinates {
(1000,0.0010575)
(2000,0.0029554)
(3000,0.0061240)
(4000,0.0099283)
(5000,0.0135687)
(6000,0.0179273)
(7000,0.0249306)
(8000,0.0301143)
(9000,0.0356169)
(10000,0.0439269)};
\end{semilogyaxis} 
\end{tikzpicture}%

%% file: fams/ran3reg-can.tex
\begin{tikzpicture}[transform shape, scale=\picScale]
\def \timeout{150}
\def \xmin{1000}
\def \xmax{10000}
\def \ymin{-0}
\def \ymax{150}
\begin{semilogyaxis}[
width=9cm,
height=8cm,
xmin=\xmin,  xmax=\xmax, 
ymin=\ymin,  ymax=\ymax, 
enlargelimits=0.01, 
grid=major,
scaled x ticks = false]
\Bliss coordinates {
(1000,0.0118729)
(2000,0.0391173)
(3000,0.1102617)
(4000,0.1893223)
(5000,0.2846419)
(6000,0.3793636)
(7000,0.5232121)
(8000,0.7559091)
(9000,0.8082879)
(10000,0.9566818)};
\Nauty coordinates {
(1000,0.1731060)
(2000,0.9962121)
(3000,2.9200000)
(4000,6.5990909)
(5000,12.2800000)
(6000,20.4963636)
(7000,32.7336364)
(8000,72.4054545)
(9000,70.6327273)
(10000,98.2609091)};
\NautyLight coordinates {
(1000,0.0004547)
(2000,0.0013259)
(3000,0.0025138)
(4000,0.0040721)
(5000,0.0058818)
(6000,0.0080058)
(7000,0.0108414)
(8000,0.0135054)
(9000,0.0168906)
(10000,0.0202746)};
\Traces coordinates {
(1000,0.0010625)
(2000,0.0029702)
(3000,0.0061441)
(4000,0.0099187)
(5000,0.0135919)
(6000,0.0179599)
(7000,0.0249298)
(8000,0.0301628)
(9000,0.0358239)
(10000,0.0439567)};
\end{semilogyaxis} 
\end{tikzpicture}%

%% file: fams/hypercubes.tex
\begin{tikzpicture}[transform shape, scale=\picScale]
\def \timeout{600}
\def \xmin{8}
\def \xmax{2097152}
\def \ymin{-0}
\def \ymax{1700}
\begin{loglogaxis}[
width=9cm,
height=8cm,
xmin=\xmin,  xmax=\xmax, 
ymin=\ymin,  ymax=\ymax, 
enlargelimits=0.01, 
legend style={anchor=north west,
at={(0.01,0.998)},
font=\tiny,
inner xsep=-.5pt,
inner ysep=-.5pt,
fill=white,
draw=none},
grid=major]
\addplot [lightgray, no markers,line width=3pt] coordinates {(\xmin,\timeout) (\xmax,\timeout)};
\addlegendentry{timeout: 600 secs}
\Bliss coordinates {
(8,0.0000160)
(16,0.0000250)
(32,0.0000660)
(64,0.0001770)
(128,0.0005098)
(256,0.0009583)
(512,0.0030303)
(1024,0.0071174)
(2048,0.0152672)
(4096,0.0381132)
(8192,0.0971429)
(16384,0.2662500)
(32768,0.6833333)
(65536,1.7400000)
(131072,6.6500000)
(262144,13.5200000)
(524288,41.9700000)
(1048576,106.6650000)
(2097152,225.1300000)};
\Nauty coordinates {
(8,0.0000026)
(16,0.0000049)
(32,0.0000154)
(64,0.0000479)
(128,0.0001690)
(256,0.0003634)
(512,0.0010460)
(1024,0.0025638)
(2048,0.0061280)
(4096,0.0135526)
(8192,0.0323437)
(16384,0.0883333)
(32768,0.2625000)
(65536,0.6666667)
(131072,2.3600000)
(262144,6.4800000)
(524288,22.0900000)
(1048576,76.7250000)
(2097152,175.8900000)};
\Conauto coordinates {
(8,0.0000133)
(16,0.0000301)
(32,0.0001051)
(64,0.0003149)
(128,0.0013841)
(256,0.0053908)
(512,0.0262338)
(1024,0.0971429)
(2048,0.5025000)
(4096,2.1500000)};
\Saucy coordinates {
(8,0.0000030)
(16,0.0000116)
(32,0.0000359)
(64,0.0000831)
(128,0.0003240)
(256,0.0008777)
(512,0.0022109)
(1024,0.0039650)
(2048,0.0112408)
(4096,0.0354531)
(8192,0.0908475)
(16384,0.1747720)
(32768,0.6317557)
(65536,1.9353360)
(131072,6.3523770)
(262144,20.9249540)
(524288,77.8195480)
(1048576,232.0548520)
(2097152,584.2536780)};
\Traces coordinates {
(8,0.0000163)
(16,0.0000260)
(32,0.0000520)
(64,0.0001029)
(128,0.0002076)
(256,0.0004433)
(512,0.0011905)
(1024,0.0025510)
(2048,0.0056389)
(4096,0.0125625)
(8192,0.0290278)
(16384,0.0709375)
(32768,0.1750000)
(65536,0.5500000)
(131072,2.3400000)
(262144,5.8200000)
(524288,18.8200000)
(1048576,58.9150000)
(2097152,170.7500000)};
\end{loglogaxis} 
\end{tikzpicture}%

%% file: fams/hypercubes-can.tex
\begin{tikzpicture}[transform shape, scale=\picScale]
\def \timeout{600}
\def \xmin{8}
\def \xmax{2097152}
\def \ymin{-0}
\def \ymax{700}
\begin{loglogaxis}[
width=9cm,
height=8cm,
xmin=\xmin,  xmax=\xmax, 
ymin=\ymin,  ymax=\ymax, 
enlargelimits=0.01, 
grid=major]
\Bliss coordinates {
(8,0.0000198)
(16,0.0000310)
(32,0.0000830)
(64,0.0002132)
(128,0.0005732)
(256,0.0011031)
(512,0.0034783)
(1024,0.0078125)
(2048,0.0167500)
(4096,0.0410204)
(8192,0.1030000)
(16384,0.2862500)
(32768,0.6833333)
(65536,1.7700000)
(131072,6.5100000)
(262144,13.5500000)
(524288,42.6900000)
(1048576,106.4550000)
(2097152,233.3000000)};
\Nauty coordinates {
(8,0.0000028)
(16,0.0000052)
(32,0.0000172)
(64,0.0000576)
(128,0.0001919)
(256,0.0004153)
(512,0.0012195)
(1024,0.0029706)
(2048,0.0072143)
(4096,0.0161719)
(8192,0.0373214)
(16384,0.1020833)
(32768,0.3085714)
(65536,0.7633333)
(131072,2.7000000)
(262144,7.0100000)
(524288,26.4100000)
(1048576,81.4900000)
(2097152,201.5900000)};
\Traces coordinates {
(8,0.0000164)
(16,0.0000261)
(32,0.0000524)
(64,0.0001031)
(128,0.0002092)
(256,0.0004472)
(512,0.0011792)
(1024,0.0026172)
(2048,0.0058140)
(4096,0.0127500)
(8192,0.0300000)
(16384,0.0693750)
(32768,0.2060000)
(65536,0.5325000)
(131072,2.2100000)
(262144,5.8300000)
(524288,18.3100000)
(1048576,62.5150000)
(2097152,155.5700000)};
\end{loglogaxis} 
\end{tikzpicture}%

%% file: fams/tran.tex
\begin{tikzpicture}[transform shape, scale=\picScale]
\def \timeout{120}
\def \xmin{6}
\def \xmax{56832}
\def \ymin{-0}
\def \ymax{20}
\begin{loglogaxis}[
width=9cm,
height=8cm,
xmin=\xmin,  xmax=\xmax, 
ymin=\ymin,  ymax=\ymax, 
enlargelimits=0.01, 
grid=major]
\BlissOnlyMarks coordinates {
(6,0.0000093)
(8,0.0000134)
(10,0.0000222)
(12,0.0000273)
(16,0.0000534)
(20,0.0000690)
(24,0.0000615)
(32,0.0000935)
(36,0.0002155)
(40,0.0000723)
(45,0.0002971)
(48,0.0004175)
(52,0.0006156)
(64,0.0003537)
(72,0.0004946)
(80,0.0004694)
(81,0.0009918)
(90,0.0014195)
(99,0.0019149)
(100,0.0005014)
(108,0.0005043)
(110,0.0002432)
(120,0.0016557)
(126,0.0017898)
(130,0.0017281)
(132,0.0019399)
(135,0.0009587)
(144,0.0024973)
(160,0.0025284)
(170,0.0004746)
(176,0.0005816)
(180,0.0032206)
(182,0.0007115)
(187,0.0018727)
(189,0.0035971)
(192,0.0033278)
(198,0.0045558)
(208,0.0044444)
(210,0.0014378)
(216,0.0010494)
(221,0.0007605)
(243,0.0019324)
(264,0.0072464)
(270,0.0010336)
(300,0.0010554)
(308,0.0032051)
(312,0.0095834)
(322,0.0073260)
(340,0.0017841)
(352,0.0140845)
(360,0.0058651)
(375,0.0025873)
(396,0.0042105)
(400,0.0027816)
(405,0.0342373)
(414,0.0082645)
(416,0.0068729)
(429,0.0147816)
(437,0.0077519)
(440,0.0133113)
(464,0.0046296)
(465,0.0067340)
(475,0.0123313)
(480,0.0015886)
(486,0.0312308)
(504,0.0032103)
(513,0.0269333)
(560,0.0183486)
(561,0.0170339)
(575,0.0315625)
(600,0.0394118)
(624,0.0016708)
(640,0.0034722)
(651,0.0046404)
(775,0.1268750)
(783,0.0307692)
(784,0.0719928)
(840,0.0396078)
(900,0.1553846)
(924,0.1592308)
(928,0.0434783)
(1024,0.0051151)
(1089,0.0152672)
(1260,0.0047170)
(1860,0.1094737)
(1920,0.0108108)
(2736,0.0088496)
(2856,0.0251250)
(2880,0.0142857)
(3016,0.0131579)
(3080,0.0242169)
(3840,0.2928571)
(4000,0.4240000)
(4200,0.0129677)
(4224,0.0116959)
(4480,0.0223333)
(4752,0.0242169)
(7296,0.6700000)
(7440,0.0322222)
(7544,0.5425000)
(7680,0.0502500)
(8064,0.0436957)
(11340,1.6100000)
(12288,2.0300000)
(13824,0.1172222)
(36864,0.1972727)
(38400,0.1900000)
(46200,0.5225000)
(46592,1.2100000)
(47104,0.3683333)
(56832,0.9700000)};
\NautyOnlyMarks coordinates {
(6,0.0000012)
(8,0.0000020)
(10,0.0000039)
(12,0.0000064)
(16,0.0000105)
(20,0.0000176)
(24,0.0000212)
(32,0.0000292)
(36,0.0000403)
(40,0.0000123)
(45,0.0000567)
(48,0.0000902)
(52,0.0002350)
(64,0.0001645)
(72,0.0002107)
(80,0.0001220)
(81,0.0005653)
(90,0.0008573)
(99,0.0017116)
(100,0.0006443)
(108,0.0001310)
(110,0.0000630)
(120,0.0005588)
(126,0.0005360)
(130,0.0006574)
(132,0.0003834)
(135,0.0012696)
(144,0.0011428)
(160,0.0009843)
(170,0.0001116)
(176,0.0001197)
(180,0.0006562)
(182,0.0001660)
(187,0.0003356)
(189,0.0004980)
(192,0.0007764)
(198,0.0008741)
(208,0.0008621)
(210,0.0004209)
(216,0.0003116)
(221,0.0001532)
(243,0.0004456)
(264,0.0011416)
(270,0.0003181)
(300,0.0003201)
(308,0.0007645)
(312,0.0033628)
(322,0.0007813)
(340,0.0004480)
(352,0.0015244)
(360,0.0007692)
(375,0.0004125)
(396,0.0009058)
(400,0.0005605)
(405,0.0152941)
(414,0.0016234)
(416,0.0013369)
(429,0.0011199)
(437,0.0006649)
(440,0.0014045)
(464,0.0010040)
(465,0.0009058)
(475,0.0013441)
(480,0.0005400)
(486,0.0223958)
(504,0.0016556)
(513,0.0015432)
(560,0.0027610)
(561,0.0016556)
(575,0.0021739)
(600,0.0114773)
(624,0.0005774)
(640,0.0016667)
(651,0.0008897)
(775,0.0768750)
(783,0.0025253)
(784,0.0031493)
(840,0.0205769)
(900,0.0246591)
(924,0.1172222)
(928,0.0177500)
(1024,0.0016779)
(1089,0.0022841)
(1260,0.0017730)
(1860,0.0970833)
(1920,0.0063125)
(2736,0.0030488)
(2856,0.0061585)
(2880,0.0043534)
(3016,0.0038068)
(3080,0.0050000)
(3840,0.5475000)
(4000,0.3616667)
(4200,0.0046991)
(4224,0.0057955)
(4480,0.0056389)
(4752,0.0108696)
(7296,1.6900000)
(7440,0.0133553)
(7544,1.3150000)
(7680,0.0236364)
(8064,0.0171667)
(11340,2.9800000)
(12288,3.5600000)
(13824,0.0502500)
(36864,0.4080000)
(38400,0.3200000)
(46200,0.1275000)
(46592,0.7133333)
(47104,0.6175000)
(56832,1.1100000)};
\ConautoOnlyMarks coordinates {
(6,0.0000084)
(8,0.0000122)
(10,0.0000180)
(12,0.0000236)
(16,0.0000330)
(20,0.0000452)
(24,0.0000469)
(32,0.0000966)
(36,0.0001083)
(40,0.0001079)
(45,0.0001671)
(48,0.0001050)
(52,0.0002111)
(64,0.0003183)
(72,0.0003574)
(80,0.0004682)
(81,0.0005027)
(90,0.0005739)
(99,0.0006008)
(100,0.0009238)
(108,0.0009153)
(110,0.0006557)
(120,0.0011943)
(126,0.0008358)
(130,0.0010135)
(132,0.0014782)
(135,0.0014355)
(144,0.0012611)
(160,0.0011912)
(170,0.0012788)
(176,0.0016681)
(180,0.0011249)
(182,0.0015201)
(187,0.0012469)
(189,0.0018484)
(192,0.0009940)
(198,0.0032895)
(208,0.0026774)
(210,0.0028902)
(216,0.0019847)
(221,0.0020305)
(243,0.0027248)
(264,0.0025543)
(270,0.0046620)
(300,0.0044843)
(308,0.0020790)
(312,0.0046379)
(322,0.0040404)
(340,0.0026076)
(352,0.0125000)
(360,0.0087336)
(375,0.0066890)
(396,0.0054201)
(400,0.0069686)
(405,0.0092593)
(414,0.0049505)
(416,0.0044150)
(429,0.0088492)
(437,0.0074349)
(440,0.0136986)
(464,0.0118343)
(465,0.0074906)
(475,0.0095694)
(480,0.0058140)
(486,0.0109290)
(504,0.0173276)
(513,0.0162602)
(560,0.0125786)
(561,0.0132450)
(575,0.0147794)
(600,0.0152672)
(624,0.0075472)
(640,0.0187850)
(651,0.0203030)
(775,0.0261039)
(783,0.0310769)
(784,0.0117347)
(840,0.0281690)
(900,0.0294118)
(924,0.0434783)
(928,0.0203030)
(1024,0.0490244)
(1089,0.0414286)
(1260,0.0412245)
(1860,0.1442857)
(1920,0.2333333)
(2736,0.1881818)
(2856,0.2377778)
(2880,0.3350000)
(3016,0.4640000)
(3080,0.4960000)
(3840,0.6300000)
(4000,0.6933333)
(4200,0.2587500)
(4224,0.3650000)
(4480,1.0150000)
(4752,0.8033333)
(7440,0.9366667)
(7544,1.2000000)
(7680,2.3900000)
(8064,2.6200000)};
\SaucyOnlyMarks coordinates {
(6,0.0000019)
(8,0.0000030)
(10,0.0000042)
(12,0.0000056)
(16,0.0000090)
(24,0.0000154)
(32,0.0000278)
(36,0.0000614)
(40,0.0000244)
(45,0.0000570)
(48,0.0001109)
(52,0.0000991)
(64,0.0001165)
(72,0.0001807)
(90,0.0003813)
(99,0.0004735)
(100,0.0001511)
(108,0.0002460)
(110,0.0001174)
(120,0.0006079)
(126,0.0006297)
(130,0.0005950)
(132,0.0006464)
(135,0.0003942)
(144,0.0008846)
(160,0.0009258)
(170,0.0002874)
(176,0.0002285)
(180,0.0010758)
(182,0.0002824)
(187,0.0003895)
(189,0.0005350)
(192,0.0011785)
(198,0.0015909)
(210,0.0006050)
(216,0.0004273)
(221,0.0002565)
(243,0.0005458)
(264,0.0020111)
(270,0.0004502)
(300,0.0004309)
(308,0.0010352)
(312,0.0026163)
(322,0.0026724)
(340,0.0006870)
(352,0.0045873)
(360,0.0017952)
(375,0.0009725)
(396,0.0023704)
(400,0.0009365)
(405,0.0110835)
(414,0.0020445)
(416,0.0029028)
(429,0.0030803)
(437,0.0015097)
(440,0.0049147)
(464,0.0009853)
(465,0.0016529)
(475,0.0037103)
(480,0.0008001)
(486,0.0116697)
(504,0.0014213)
(513,0.0068290)
(560,0.0067513)
(561,0.0059247)
(575,0.0036156)
(600,0.0137052)
(624,0.0009691)
(640,0.0019785)
(775,0.0477434)
(783,0.0085072)
(784,0.0085250)
(840,0.0161761)
(900,0.0279112)
(924,0.0490921)
(1024,0.0021214)
(1089,0.0030520)
(1260,0.0020039)
(1860,0.0455024)
(1920,0.0045052)
(2736,0.0041733)
(2856,0.0090857)
(2880,0.0051072)
(3016,0.0060048)
(3080,0.0081644)
(3840,0.0933741)
(4200,0.0052848)
(4480,0.0087659)
(4752,0.0110014)
(7296,0.2547917)
(7440,0.0091123)
(7544,0.1724707)
(8064,0.0188008)
(11340,0.4606038)
(12288,0.5535155)
(13824,0.0233274)
(36864,0.0670103)
(38400,0.0579284)
(46592,0.1234188)
(47104,0.0741452)
(56832,0.1125106)};
\TracesOnlyMarks coordinates {
(6,0.0000066)
(8,0.0000127)
(10,0.0000163)
(12,0.0000185)
(16,0.0000310)
(20,0.0000365)
(24,0.0000378)
(32,0.0000412)
(36,0.0000790)
(40,0.0000277)
(45,0.0000838)
(48,0.0001475)
(52,0.0002781)
(64,0.0001562)
(72,0.0001947)
(80,0.0001565)
(81,0.0004879)
(90,0.0010381)
(99,0.0019572)
(100,0.0002505)
(108,0.0001867)
(110,0.0000890)
(120,0.0006982)
(126,0.0005473)
(130,0.0007078)
(132,0.0003482)
(135,0.0007566)
(144,0.0009662)
(160,0.0011211)
(170,0.0001409)
(176,0.0001330)
(180,0.0004252)
(182,0.0002432)
(187,0.0002793)
(189,0.0003650)
(192,0.0004587)
(198,0.0008197)
(208,0.0012755)
(210,0.0002938)
(216,0.0003025)
(221,0.0001867)
(243,0.0006579)
(264,0.0007062)
(270,0.0004554)
(300,0.0003378)
(308,0.0005682)
(312,0.0029841)
(322,0.0005092)
(340,0.0003618)
(352,0.0011848)
(360,0.0006667)
(375,0.0004052)
(396,0.0006964)
(400,0.0004970)
(405,0.0147059)
(414,0.0009653)
(416,0.0008224)
(429,0.0008772)
(437,0.0007003)
(440,0.0011312)
(464,0.0009881)
(465,0.0006631)
(475,0.0010549)
(480,0.0004798)
(486,0.0170833)
(504,0.0015244)
(513,0.0013966)
(560,0.0028409)
(561,0.0012195)
(575,0.0015924)
(600,0.0145833)
(624,0.0008503)
(640,0.0018116)
(651,0.0009363)
(775,0.0562500)
(783,0.0018611)
(784,0.0033830)
(840,0.0251250)
(900,0.0122619)
(924,0.1029167)
(928,0.0077273)
(1024,0.0015244)
(1089,0.0019685)
(1260,0.0021113)
(1860,0.1127778)
(1920,0.0049020)
(2736,0.0039648)
(2856,0.0054891)
(2880,0.0037313)
(3016,0.0039881)
(3080,0.0051531)
(3840,0.5275000)
(4000,0.0740625)
(4200,0.0032051)
(4224,0.0041667)
(4480,0.0056667)
(4752,0.0099038)
(7296,0.1538462)
(7440,0.0069257)
(7544,0.1366667)
(7680,0.0192308)
(8064,0.0163281)
(11340,0.3114286)
(12288,0.3616667)
(13824,0.0178571)
(36864,0.0575000)
(38400,0.0431250)
(46200,0.0833333)
(46592,0.1241176)
(47104,0.0787500)
(56832,0.0937500)};
\end{loglogaxis} 
\end{tikzpicture}%

%% file: fams/tran-can.tex
\begin{tikzpicture}[transform shape, scale=\picScale]
\def \timeout{120}
\def \xmin{6}
\def \xmax{56832}
\def \ymin{-0}
\def \ymax{20}
\begin{loglogaxis}[
width=9cm,
height=8cm,
xmin=\xmin,  xmax=\xmax, 
ymin=\ymin,  ymax=\ymax, 
enlargelimits=0.01, 
grid=major]
\BlissOnlyMarks coordinates {
(6,0.0000112)
(8,0.0000169)
(10,0.0000258)
(12,0.0000314)
(16,0.0000621)
(20,0.0000859)
(24,0.0000750)
(32,0.0001135)
(36,0.0002630)
(40,0.0000947)
(45,0.0003313)
(48,0.0005175)
(52,0.0007013)
(64,0.0004083)
(72,0.0005834)
(80,0.0005507)
(81,0.0011466)
(90,0.0016445)
(99,0.0021631)
(100,0.0005919)
(108,0.0005767)
(110,0.0003001)
(120,0.0018744)
(126,0.0020493)
(130,0.0019659)
(132,0.0021954)
(135,0.0011099)
(144,0.0028911)
(160,0.0028409)
(170,0.0006026)
(176,0.0007145)
(180,0.0036765)
(182,0.0008356)
(187,0.0021368)
(189,0.0040900)
(192,0.0038241)
(198,0.0053476)
(208,0.0048900)
(210,0.0017065)
(216,0.0012269)
(221,0.0009294)
(243,0.0022051)
(264,0.0082305)
(270,0.0012658)
(300,0.0012723)
(308,0.0036765)
(312,0.0108768)
(322,0.0082645)
(340,0.0021345)
(352,0.0154615)
(360,0.0064935)
(375,0.0030257)
(396,0.0048077)
(400,0.0032520)
(405,0.0384615)
(414,0.0091324)
(416,0.0077821)
(429,0.0163252)
(437,0.0088889)
(440,0.0148148)
(464,0.0051813)
(465,0.0078125)
(475,0.0140845)
(480,0.0019305)
(486,0.0358929)
(504,0.0036101)
(513,0.0301493)
(560,0.0201000)
(561,0.0191429)
(575,0.0354386)
(600,0.0439130)
(624,0.0020555)
(640,0.0039139)
(651,0.0053050)
(775,0.1464286)
(783,0.0346552)
(784,0.0802254)
(840,0.0476190)
(900,0.1661538)
(924,0.1872727)
(928,0.0523077)
(1024,0.0057971)
(1089,0.0170339)
(1260,0.0055556)
(1860,0.1250000)
(1920,0.0120482)
(2736,0.0104712)
(2856,0.0285714)
(2880,0.0164754)
(3016,0.0152672)
(3080,0.0279167)
(3840,0.3228571)
(4000,0.4840000)
(4200,0.0150000)
(4224,0.0133333)
(4480,0.0239286)
(4752,0.0268000)
(7296,0.7866667)
(7440,0.0377358)
(7544,0.6500000)
(7680,0.0566667)
(8064,0.0492683)
(11340,1.8000000)
(12288,2.3000000)
(13824,0.1241176)
(36864,0.2120000)
(38400,0.2200000)
(46200,0.5175000)
(46592,1.2550000)
(47104,0.4020000)
(56832,1.0200000)};
\NautyOnlyMarks coordinates {
(6,0.0000013)
(8,0.0000022)
(10,0.0000044)
(12,0.0000076)
(16,0.0000125)
(20,0.0000220)
(24,0.0000284)
(32,0.0000416)
(36,0.0000513)
(40,0.0000132)
(45,0.0000752)
(48,0.0001067)
(52,0.0005081)
(64,0.0002746)
(72,0.0005182)
(80,0.0001615)
(81,0.0014946)
(90,0.0019034)
(99,0.0038247)
(100,0.0018519)
(108,0.0001472)
(110,0.0000662)
(120,0.0007906)
(126,0.0007652)
(130,0.0008880)
(132,0.0004386)
(135,0.0052302)
(144,0.0019262)
(160,0.0014368)
(170,0.0001179)
(176,0.0001265)
(180,0.0007418)
(182,0.0001957)
(187,0.0003597)
(189,0.0005543)
(192,0.0008651)
(198,0.0010460)
(208,0.0013228)
(210,0.0005040)
(216,0.0003314)
(221,0.0001605)
(243,0.0006127)
(264,0.0012821)
(270,0.0003511)
(300,0.0003392)
(308,0.0008446)
(312,0.0045763)
(322,0.0008503)
(340,0.0004771)
(352,0.0017092)
(360,0.0008197)
(375,0.0004333)
(396,0.0010417)
(400,0.0005967)
(405,0.0221875)
(414,0.0018611)
(416,0.0015060)
(429,0.0011956)
(437,0.0007042)
(440,0.0015320)
(464,0.0010823)
(465,0.0009615)
(475,0.0014535)
(480,0.0005760)
(486,0.0351562)
(504,0.0022124)
(513,0.0016340)
(560,0.0032372)
(561,0.0017986)
(575,0.0023364)
(600,0.0150000)
(624,0.0006098)
(640,0.0018474)
(651,0.0009398)
(775,0.1386667)
(783,0.0027174)
(784,0.0032856)
(840,0.0287500)
(900,0.0298611)
(924,0.2150000)
(928,0.0191964)
(1024,0.0017819)
(1089,0.0024752)
(1260,0.0018939)
(1860,0.1366667)
(1920,0.0090179)
(2736,0.0032372)
(2856,0.0070833)
(2880,0.0045909)
(3016,0.0039453)
(3080,0.0053191)
(3840,0.8233333)
(4000,0.4080000)
(4200,0.0047642)
(4224,0.0058430)
(4480,0.0058430)
(4752,0.0113636)
(7296,2.1000000)
(7440,0.0134868)
(7544,1.6250000)
(7680,0.0248864)
(8064,0.0178571)
(11340,3.7200000)
(12288,4.5200000)
(13824,0.0505000)
(36864,0.4080000)
(38400,0.3200000)
(46200,0.1306250)
(46592,0.7200000)
(47104,0.6150000)
(56832,1.1100000)};
\TracesOnlyMarks coordinates {
(6,0.0000066)
(8,0.0000128)
(10,0.0000160)
(12,0.0000186)
(16,0.0000313)
(20,0.0000369)
(24,0.0000365)
(32,0.0000419)
(36,0.0000800)
(40,0.0000308)
(45,0.0000838)
(48,0.0001487)
(52,0.0002796)
(64,0.0001581)
(72,0.0001955)
(80,0.0001578)
(81,0.0004907)
(90,0.0010342)
(99,0.0019722)
(100,0.0002580)
(108,0.0001885)
(110,0.0001008)
(120,0.0006942)
(126,0.0005540)
(130,0.0007121)
(132,0.0003506)
(135,0.0007537)
(144,0.0009948)
(160,0.0011261)
(170,0.0001800)
(176,0.0001690)
(180,0.0004274)
(182,0.0002560)
(187,0.0003053)
(189,0.0004417)
(192,0.0004990)
(198,0.0008251)
(208,0.0012755)
(210,0.0003463)
(216,0.0003146)
(221,0.0002296)
(243,0.0006596)
(264,0.0008865)
(270,0.0004587)
(300,0.0003704)
(308,0.0005800)
(312,0.0030076)
(322,0.0008143)
(340,0.0003864)
(352,0.0015924)
(360,0.0008591)
(375,0.0006313)
(396,0.0008143)
(400,0.0005995)
(405,0.0147794)
(414,0.0010163)
(416,0.0009191)
(429,0.0014521)
(437,0.0010776)
(440,0.0015625)
(464,0.0010331)
(465,0.0010000)
(475,0.0014620)
(480,0.0004854)
(486,0.0171667)
(504,0.0015337)
(513,0.0022523)
(560,0.0027778)
(561,0.0019531)
(575,0.0024632)
(600,0.0143750)
(624,0.0008446)
(640,0.0017986)
(651,0.0010204)
(775,0.0565000)
(783,0.0034247)
(784,0.0035350)
(840,0.0251250)
(900,0.0134211)
(924,0.0995833)
(928,0.0079688)
(1024,0.0014793)
(1089,0.0023810)
(1260,0.0021659)
(1860,0.1068421)
(1920,0.0049265)
(2736,0.0039453)
(2856,0.0056389)
(2880,0.0036949)
(3016,0.0043319)
(3080,0.0067905)
(3840,0.5300000)
(4000,0.0768750)
(4200,0.0032630)
(4224,0.0056111)
(4480,0.0056667)
(4752,0.0098077)
(7296,0.1561538)
(7440,0.0089286)
(7544,0.1331250)
(7680,0.0195192)
(8064,0.0163281)
(11340,0.3114286)
(12288,0.3733333)
(13824,0.0178571)
(36864,0.0646875)
(38400,0.0587500)
(46200,0.1016667)
(46592,0.1262500)
(47104,0.0941667)
(56832,0.1211765)};
\end{loglogaxis} 
\end{tikzpicture}%

%% file: fams/tnn.tex
\begin{tikzpicture}[transform shape, scale=\picScale]
\def \timeout{600}
\def \xmin{26}
\def \xmax{1014}
\def \ymin{-0}
\def \ymax{1700}
\begin{semilogyaxis}[
width=9cm,
height=8cm,
xmin=\xmin,  xmax=\xmax, 
ymin=\ymin,  ymax=\ymax, 
enlargelimits=0.01, 
legend style={anchor=north west,
at={(0.01,0.998)},
font=\tiny,
inner xsep=-.5pt,
inner ysep=-.5pt,
fill=white,
draw=none},
grid=major]
\addplot [lightgray, no markers,line width=3pt] coordinates {(\xmin,\timeout) (\xmax,\timeout)};
\addlegendentry{timeout: 600 secs}
\Bliss coordinates {
(26,0.0000363)
(52,0.0001630)
(78,0.0004068)
(104,0.0009867)
(182,0.0033956)
(286,0.0112360)
(416,0.0574286)
(598,0.1528571)
(806,0.2957143)
(1014,0.4960000)};
\Nauty coordinates {
(26,0.0000118)
(52,0.0001084)
(78,0.0003497)
(104,0.0007937)
(182,0.0053989)
(286,0.0283333)
(416,0.1406667)
(598,0.5250000)
(806,1.1650000)
(1014,2.6300000)};
\Conauto coordinates {
(26,0.0000589)
(52,0.0001808)
(78,0.0003905)
(104,0.0006618)
(182,0.0021882)
(286,0.0061728)
(416,0.0160000)
(598,0.0453333)
(806,0.1089474)
(1014,0.2737500)};
\Saucy coordinates {
(26,0.0000125)
(52,0.0000440)
(78,0.0013219)
(104,0.0001932)
(182,1.0527605)
(286,\timeout)
};
\Traces coordinates {
(26,0.0000554)
(52,0.0001143)
(78,0.0004845)
(104,0.0010504)
(182,0.0069792)
(286,0.0460417)
(416,0.2140000)
(598,1.3400000)
(806,45.8200000)
(1014,\timeout)};
\end{semilogyaxis} 
\end{tikzpicture}%

%% file: fams/tnn-can.tex
\begin{tikzpicture}[transform shape, scale=\picScale]
\def \timeout{600}
\def \xmin{26}
\def \xmax{1014}
\def \ymin{-0}
\def \ymax{1700}
\begin{semilogyaxis}[
width=9cm,
height=8cm,
xmin=\xmin,  xmax=\xmax, 
ymin=\ymin,  ymax=\ymax, 
enlargelimits=0.01, 
legend style={anchor=north west,
at={(0.01,0.998)},
font=\tiny,
inner xsep=-.5pt,
inner ysep=-.5pt,
fill=white,
draw=none},
grid=major]
\addplot [lightgray, no markers,line width=3pt] coordinates {(\xmin,\timeout) (\xmax,\timeout)};
\addlegendentry{timeout: 600 secs}
\Bliss coordinates {
(26,0.0000635)
(52,0.0006633)
(78,0.0023229)
(104,0.1366667)
(182,9.3300000)
(286,\timeout)
};
\Nauty coordinates {
(26,0.0000350)
(52,0.0006053)
(78,0.0052344)
(104,0.0191071)
(182,2.0400000)
(286,\timeout)
};
\Traces coordinates {
(26,0.0000793)
(52,0.0001630)
(78,0.0006944)
(104,0.0014881)
(182,0.0093519)
(286,0.0656250)
(416,0.3142857)
(598,2.2200000)
(806,55.6300000)
(1014,\timeout)};
\end{semilogyaxis} 
\end{tikzpicture}%

%% file: fams/srg-small.tex
\begin{tikzpicture}[transform shape, scale=\picScale]
\def \timeout{120}
\def \xmin{4}
\def \xmax{500}
\def \ymin{-0}
\def \ymax{10}
\begin{semilogyaxis}[
width=9cm,
height=8cm,
xmin=\xmin,  xmax=\xmax, 
ymin=\ymin,  ymax=\ymax, 
enlargelimits=0.01, 
grid=major,
scaled x ticks = false,
x tick label style={/pgf/number format/fixed},
xtick={100,200,300,400,500}]
\BlissOnlyMarks coordinates {
(4,0.0000107)
(5,0.0000090)
(6,0.0000150)
(7,0.0000230)
(9,0.0000189)
(10,0.0000238)
(12,0.0000344)
(13,0.0000173)
(15,0.0000399)
(16,0.0000354)
(17,0.0000244)
(21,0.0000750)
(25,0.0000791)
(26,0.0000846)
(28,0.0001417)
(29,0.0000612)
(35,0.0001620)
(36,0.0001982)
(37,0.0001050)
(41,0.0002139)
(45,0.0003618)
(49,0.0003128)
(53,0.0002844)
(55,0.0005817)
(57,0.0006860)
(61,0.0003999)
(64,0.0005177)
(66,0.0007816)
(70,0.0009728)
(73,0.0004405)
(78,0.0010466)
(81,0.0009183)
(89,0.0006787)
(91,0.0014641)
(97,0.0007984)
(100,0.0031650)
(101,0.0014388)
(105,0.0020284)
(109,0.0010220)
(113,0.0014837)
(117,0.0024603)
(120,0.0024570)
(121,0.0098619)
(125,0.0026525)
(136,0.0031949)
(137,0.0016461)
(144,0.0110193)
(149,0.0019802)
(153,0.0040404)
(155,0.0046142)
(157,0.0029499)
(169,0.0216300)
(171,0.0053763)
(173,0.0026596)
(176,0.0058099)
(181,0.0047847)
(190,0.0063291)
(193,0.0045351)
(196,0.0181967)
(197,0.0046512)
(210,0.0075188)
(222,0.0101958)
(225,0.0344926)
(229,0.0063492)
(231,0.0093023)
(233,0.0050378)
(241,0.0053763)
(247,0.0128722)
(253,0.0105820)
(256,0.0568594)
(257,0.0061538)
(269,0.0070671)
(276,0.0136986)
(277,0.0073260)
(281,0.0101010)
(289,0.0769586)
(293,0.0082305)
(300,0.0162602)
(301,0.0216472)
(313,0.0157480)
(317,0.0128205)
(324,0.0718540)
(325,0.0202020)
(330,0.0263682)
(337,0.0112921)
(349,0.0157480)
(351,0.0217391)
(353,0.0161290)
(361,0.0729119)
(373,0.0182727)
(378,0.0264474)
(389,0.0150376)
(392,0.0468474)
(397,0.0156250)
(400,0.1367646)
(401,0.0159524)
(406,0.0281690)
(409,0.0225843)
(421,0.0238095)
(425,0.0605074)
(433,0.0248148)
(435,0.0352632)
(441,0.1431982)
(449,0.0398039)
(457,0.0279167)
(461,0.0291304)
(484,0.3092971)
(495,0.0956794)
(529,0.2015170)
(532,0.1221728)
(576,0.4973632)
(610,0.1788034)
(625,0.2417207)
(651,0.2164053)
(676,1.0292793)
(729,0.4868268)
(737,0.3655238)
(782,0.4230595)
(784,0.2978175)
(841,0.2967319)
(876,0.6035669)
(900,0.8417557)
(925,0.6993483)
(1024,3.5000000)
(1027,0.9934722)
(1156,3.3500000)
(1225,0.2737500)
(1296,0.3566667)
(1444,0.4380000)
(1600,12.9300000)
(2025,1.2450000)
(2500,2.0100000)};
\NautyOnlyMarks coordinates {
(4,0.0000013)
(5,0.0000011)
(6,0.0000020)
(7,0.0000042)
(9,0.0000036)
(10,0.0000034)
(12,0.0000090)
(13,0.0000031)
(15,0.0000082)
(16,0.0000075)
(17,0.0000046)
(21,0.0000111)
(25,0.0000136)
(26,0.0000815)
(28,0.0000229)
(29,0.0000086)
(35,0.0001073)
(36,0.0000544)
(37,0.0000149)
(41,0.0000143)
(45,0.0000638)
(49,0.0000616)
(53,0.0000210)
(55,0.0001154)
(57,0.0083790)
(61,0.0000350)
(64,0.0001271)
(66,0.0001635)
(70,0.0145874)
(73,0.0000344)
(78,0.0002246)
(81,0.0001780)
(89,0.0000633)
(91,0.0003016)
(97,0.0000583)
(100,0.0376722)
(101,0.0000796)
(105,0.0003660)
(109,0.0000703)
(113,0.0001023)
(117,0.0566086)
(120,0.0004630)
(121,0.0387231)
(125,0.0004259)
(136,0.0006219)
(137,0.0001436)
(144,0.0658219)
(149,0.0001232)
(153,0.0007837)
(155,0.1206771)
(157,0.0001741)
(169,0.1041361)
(171,0.0009690)
(173,0.0002183)
(176,0.1646547)
(181,0.0001862)
(190,0.0010638)
(193,0.0003205)
(196,0.1859098)
(197,0.0002662)
(210,0.0012887)
(222,0.3172262)
(225,0.2779809)
(229,0.0002651)
(231,0.0016779)
(233,0.0004488)
(241,0.0003693)
(247,0.3957360)
(253,0.0019531)
(256,0.3327400)
(257,0.0004000)
(269,0.0003566)
(276,0.0022635)
(277,0.0004647)
(281,0.0003731)
(289,0.4805874)
(293,0.0005102)
(300,0.0025902)
(301,0.6877778)
(313,0.0004647)
(317,0.0006010)
(324,0.5408433)
(325,0.0030271)
(330,0.8339497)
(337,0.0007003)
(349,0.0007225)
(351,0.0035069)
(353,0.0005580)
(361,0.7550666)
(373,0.0006068)
(378,0.0037132)
(389,0.0006925)
(392,1.3382500)
(397,0.0009191)
(400,1.5272185)
(401,0.0007184)
(406,0.0043103)
(409,0.0007310)
(421,0.0009542)
(425,1.5994288)
(433,0.0008039)
(435,0.0049020)
(441,1.4943815)
(449,0.0010823)
(457,0.0013889)
(461,0.0011416)
(484,2.0139247)
(495,2.5279167)
(529,1.9528806)
(532,2.8931920)
(576,2.4218590)
(610,4.2950000)
(625,3.9329628)
(651,4.8960208)
(676,3.6533782)
(729,4.4840294)
(737,7.1183333)
(782,7.9440278)
(784,6.1913104)
(841,9.7835493)
(876,11.0700000)
(900,35.3987768)
(925,12.4030903)
(1024,57.9600000)
(1027,16.8033333)
(1156,78.4700000)
(1225,91.9200000)
(1296,108.6900000)
(1444,140.8500000)
(1600,185.9000000)
(2025,348.3600000)
(2500,198.5400000)};
\ConautoOnlyMarks coordinates {
(4,0.0000077)
(5,0.0000066)
(6,0.0000088)
(7,0.0000145)
(9,0.0000129)
(10,0.0000141)
(12,0.0000202)
(13,0.0000141)
(15,0.0000248)
(16,0.0000251)
(17,0.0000166)
(21,0.0000382)
(25,0.0000571)
(26,0.0000333)
(28,0.0000867)
(29,0.0000249)
(35,0.0000413)
(36,0.0001158)
(37,0.0000371)
(41,0.0000355)
(45,0.0001778)
(49,0.0001673)
(53,0.0000545)
(55,0.0003281)
(57,0.0006403)
(61,0.0000757)
(64,0.0003083)
(66,0.0004410)
(70,0.0008525)
(73,0.0000931)
(78,0.0005447)
(81,0.0004348)
(89,0.0001137)
(91,0.0008580)
(97,0.0001200)
(100,0.0033857)
(101,0.0001568)
(105,0.0010560)
(109,0.0001315)
(113,0.0001974)
(117,0.0033963)
(120,0.0012658)
(121,0.0057134)
(125,0.0005137)
(136,0.0018382)
(137,0.0002010)
(144,0.0089241)
(149,0.0002144)
(153,0.0022962)
(155,0.0084293)
(157,0.0002734)
(169,0.0155345)
(171,0.0027174)
(173,0.0003275)
(176,0.0126089)
(181,0.0002865)
(190,0.0036969)
(193,0.0004517)
(196,0.0209548)
(197,0.0004184)
(210,0.0044248)
(222,0.0325490)
(225,0.0290674)
(229,0.0004736)
(231,0.0051414)
(233,0.0006127)
(241,0.0005900)
(247,0.0326994)
(253,0.0069444)
(256,0.0523919)
(257,0.0007186)
(269,0.0005741)
(276,0.0078740)
(277,0.0008471)
(281,0.0006315)
(289,0.0666729)
(293,0.0007752)
(300,0.0087719)
(301,0.0575635)
(313,0.0007479)
(317,0.0009809)
(324,0.0799161)
(325,0.0113636)
(330,0.0834172)
(337,0.0012731)
(349,0.0011448)
(351,0.0128025)
(353,0.0009046)
(361,0.1141851)
(373,0.0010293)
(378,0.0150376)
(389,0.0011370)
(392,0.1304598)
(397,0.0013708)
(400,0.1650376)
(401,0.0010977)
(406,0.0179464)
(409,0.0012188)
(421,0.0015420)
(425,0.1494596)
(433,0.0013793)
(435,0.0210526)
(441,0.2798547)
(449,0.0016667)
(457,0.0020202)
(461,0.0016181)
(484,0.2403240)
(495,0.2350390)
(529,0.2529637)
(532,0.3401702)
(576,0.3185322)
(610,0.3942673)
(625,0.4414899)
(651,0.5686338)
(676,0.5848018)
(729,0.6394672)
(737,0.9470238)
(782,0.9592577)
(784,0.6965011)
(841,1.1388875)
(876,1.0584087)
(900,1.3036194)
(925,1.5052193)
(1024,3.1200000)
(1027,2.1545833)
(1156,2.0400000)
(1225,2.1100000)
(1296,5.3700000)
(1444,3.6000000)
(1600,6.7700000)
(2025,19.6100000)
(2500,18.0800000)};
\SaucyOnlyMarks coordinates {
(4,0.0000014)
(5,0.0000018)
(6,0.0000021)
(7,0.0000028)
(9,0.0000042)
(10,0.0000062)
(12,0.0000070)
(13,0.0000073)
(15,0.0000134)
(16,0.0000132)
(17,0.0000094)
(21,0.0000266)
(25,0.0000315)
(26,0.0001048)
(28,0.0000514)
(29,0.0000223)
(35,0.0001384)
(36,0.0000764)
(37,0.0000341)
(41,0.0000434)
(45,0.0001634)
(49,0.0001226)
(53,0.0000570)
(55,0.0002636)
(57,0.0045742)
(61,0.0000596)
(64,0.0002381)
(66,0.0003754)
(70,0.0074938)
(73,0.0001237)
(78,0.0005453)
(81,0.0002608)
(89,0.0001404)
(91,0.0007445)
(97,0.0002051)
(100,0.0188738)
(101,0.0001632)
(105,0.0009769)
(109,0.0002330)
(113,0.0002703)
(117,0.0290950)
(120,0.0012600)
(121,0.0281765)
(125,0.0007134)
(136,0.0016176)
(137,0.0003915)
(144,0.0461942)
(149,0.0003066)
(153,0.0021207)
(155,0.0636673)
(157,0.0005623)
(169,0.0645987)
(171,0.0026302)
(173,0.0004877)
(176,0.0855271)
(181,0.0005938)
(190,0.0032505)
(193,0.0004763)
(196,0.1028202)
(197,0.0006553)
(210,0.0039466)
(222,0.1699156)
(225,0.1507036)
(229,0.0008371)
(231,0.0046323)
(233,0.0007157)
(241,0.0009461)
(247,0.2145693)
(253,0.0056773)
(256,0.2130940)
(257,0.0011455)
(269,0.0008833)
(276,0.0067192)
(277,0.0015390)
(281,0.0013630)
(289,0.2716951)
(293,0.0013317)
(300,0.0080825)
(301,0.3968688)
(313,0.0012158)
(317,0.0015286)
(324,0.4000675)
(325,0.0094399)
(330,0.4831930)
(337,0.0017887)
(349,0.0018842)
(351,0.0110192)
(353,0.0023749)
(361,0.4993886)
(373,0.0021331)
(378,0.0122732)
(389,0.0022717)
(392,0.7986853)
(397,0.0029359)
(400,0.7271645)
(401,0.0024287)
(406,0.0148768)
(409,0.0031888)
(421,0.0026674)
(425,0.9459675)
(433,0.0022312)
(435,0.0167324)
(441,0.9464554)
(449,0.0023300)
(457,0.0031580)
(461,0.0023351)
(484,1.1911063)
(495,1.5354388)
(529,1.5993643)
(532,1.7316616)
(576,1.8685151)
(610,2.7233505)
(625,2.4074672)
(651,2.9661871)
(676,2.8532499)
(729,3.6020125)
(737,4.3720867)
(782,4.9143778)
(784,4.4071738)
(841,5.3046040)
(876,7.1319578)
(900,6.3174103)
(925,7.8199195)
(1024,10.6327270)
(1027,10.8898882)
(1156,15.4666320)
(1225,18.4860870)
(1296,21.5835800)
(1444,29.4309730)
(1600,37.7234300)
(2025,80.3928880)
(2500,144.1079350)};
\TracesOnlyMarks coordinates {
(4,0.0000063)
(5,0.0000066)
(6,0.0000118)
(7,0.0000065)
(9,0.0000154)
(10,0.0000156)
(12,0.0000253)
(13,0.0000154)
(15,0.0000233)
(16,0.0000263)
(17,0.0000172)
(21,0.0000356)
(25,0.0000363)
(26,0.0000334)
(28,0.0000581)
(29,0.0000212)
(35,0.0000511)
(36,0.0000798)
(37,0.0000266)
(41,0.0000386)
(45,0.0001394)
(49,0.0001006)
(53,0.0000418)
(55,0.0001952)
(57,0.0003039)
(61,0.0000521)
(64,0.0001873)
(66,0.0002884)
(70,0.0004574)
(73,0.0000538)
(78,0.0003788)
(81,0.0002172)
(89,0.0000686)
(91,0.0005308)
(97,0.0000782)
(100,0.0010196)
(101,0.0001174)
(105,0.0006757)
(109,0.0000943)
(113,0.0001139)
(117,0.0013019)
(120,0.0009191)
(121,0.0014389)
(125,0.0003639)
(136,0.0011161)
(137,0.0001229)
(144,0.0027896)
(149,0.0001404)
(153,0.0014368)
(155,0.0023818)
(157,0.0001875)
(169,0.0116642)
(171,0.0017946)
(173,0.0001685)
(176,0.0029820)
(181,0.0002897)
(190,0.0022124)
(193,0.0002561)
(196,0.0171129)
(197,0.0002723)
(210,0.0025638)
(222,0.0049972)
(225,0.0206328)
(229,0.0003429)
(231,0.0031250)
(233,0.0002969)
(241,0.0003145)
(247,0.0060061)
(253,0.0039881)
(256,0.0284742)
(257,0.0003234)
(269,0.0003618)
(276,0.0045455)
(277,0.0003671)
(281,0.0004655)
(289,0.0336154)
(293,0.0004188)
(300,0.0052344)
(301,0.0092506)
(313,0.0006667)
(317,0.0005938)
(324,0.0458085)
(325,0.0064744)
(330,0.0109000)
(337,0.0005198)
(349,0.0006925)
(351,0.0074265)
(353,0.0007082)
(361,0.0535224)
(373,0.0007599)
(378,0.0083065)
(389,0.0006494)
(392,0.0161227)
(397,0.0006868)
(400,0.0662425)
(401,0.0007022)
(406,0.0092593)
(409,0.0008929)
(421,0.0009728)
(425,0.0183449)
(433,0.0010204)
(435,0.0106771)
(441,0.0856113)
(449,0.0014451)
(457,0.0011111)
(461,0.0011062)
(484,0.1024716)
(495,0.0261307)
(529,0.1185699)
(532,0.0293056)
(576,0.0364380)
(610,0.0422321)
(625,0.0430473)
(651,0.0462326)
(676,0.0497566)
(729,0.0575949)
(737,0.0616719)
(782,0.0695759)
(784,0.0687095)
(841,0.0802956)
(876,0.0898351)
(900,0.0911963)
(925,0.0989514)
(1024,0.1366667)
(1027,0.1269329)
(1156,0.1733333)
(1225,0.1990909)
(1296,0.2266667)
(1444,0.2787500)
(1600,0.3500000)
(2025,0.5850000)
(2500,0.9333333)};
\end{semilogyaxis} 
\end{tikzpicture}%

%% file: fams/srg-small-can.tex
\begin{tikzpicture}[transform shape, scale=\picScale]
\def \timeout{120}
\def \xmin{4}
\def \xmax{500}
\def \ymin{-0}
\def \ymax{10}
\begin{semilogyaxis}[
width=9cm,
height=8cm,
xmin=\xmin,  xmax=\xmax, 
ymin=\ymin,  ymax=\ymax, 
enlargelimits=0.01, 
grid=major,
scaled x ticks = false,
x tick label style={/pgf/number format/fixed},
xtick={100,200,300,400,500}]
\BlissOnlyMarks coordinates {
(4,0.0000130)
(5,0.0000110)
(6,0.0001194)
(7,0.0000280)
(9,0.0000238)
(10,0.0001790)
(12,0.0000511)
(13,0.0000236)
(15,0.0003109)
(16,0.0000540)
(17,0.0000340)
(21,0.0004955)
(25,0.0000996)
(26,0.0001203)
(28,0.0009901)
(29,0.0000847)
(35,0.0002896)
(36,0.0006864)
(37,0.0001435)
(41,0.0002720)
(45,0.0003986)
(49,0.0003783)
(53,0.0003717)
(55,0.0006487)
(57,0.0072604)
(61,0.0005054)
(64,0.0006478)
(66,0.0008482)
(70,0.0128471)
(73,0.0005889)
(78,0.0011325)
(81,0.0010622)
(89,0.0008711)
(91,0.0015773)
(97,0.0010325)
(100,0.0369658)
(101,0.0017376)
(105,0.0021858)
(109,0.0013132)
(113,0.0017778)
(117,0.0585876)
(120,0.0026212)
(121,0.0590761)
(125,0.0031008)
(136,0.0032733)
(137,0.0020899)
(144,0.1099728)
(149,0.0025063)
(153,0.0042283)
(155,0.1407175)
(157,0.0035149)
(169,0.1623786)
(171,0.0186111)
(173,0.0033670)
(176,0.1987289)
(181,0.0056818)
(190,0.0390385)
(193,0.0053050)
(196,0.2619277)
(197,0.0056180)
(210,0.0436957)
(222,0.4159545)
(225,0.4368480)
(229,0.0076923)
(231,0.0548649)
(233,0.0062696)
(241,0.0066890)
(247,0.5546678)
(253,0.0561111)
(256,0.5904214)
(257,0.0075758)
(269,0.0086580)
(276,0.0676667)
(277,0.0090909)
(281,0.0117647)
(289,0.8222386)
(293,0.0100000)
(300,0.0788462)
(301,0.9981944)
(313,0.0178571)
(317,0.0152672)
(324,1.2460984)
(325,0.1116667)
(330,1.2924949)
(337,0.0136054)
(349,0.0186111)
(351,0.1205882)
(353,0.0190476)
(361,1.6123216)
(373,0.0215054)
(378,0.1340000)
(389,0.0183486)
(392,2.2488889)
(397,0.0190476)
(400,2.2761979)
(401,0.0198020)
(406,0.1306250)
(409,0.0263158)
(421,0.0280556)
(425,2.7872994)
(433,0.0297059)
(435,0.1542857)
(441,3.0052080)
(449,0.0451111)
(457,0.0327869)
(461,0.0335000)
(484,3.9417521)
(495,4.5779167)
(529,5.2525789)
(532,5.5220270)
(576,7.0065077)
(610,8.8133333)
(625,8.8033295)
(651,10.2770076)
(676,10.7207730)
(729,13.2586325)
(737,15.6608333)
(782,17.8313172)
(784,16.7798352)
(841,20.9255271)
(876,26.2691667)
(900,26.0196702)
(925,30.3591667)
(1024,43.1100000)
(1027,42.8983333)
(1156,62.9600000)
(1225,75.9100000)
(1296,92.2300000)
(1444,127.2800000)
(1600,176.9600000)
(2025,356.4200000)
(2500,\timeout)};
\NautyOnlyMarks coordinates {
(4,0.0000014)
(5,0.0000011)
(6,0.0000021)
(7,0.0000045)
(9,0.0000043)
(10,0.0000036)
(12,0.0000108)
(13,0.0000031)
(15,0.0000093)
(16,0.0000083)
(17,0.0000046)
(21,0.0000125)
(25,0.0000151)
(26,0.0000813)
(28,0.0000288)
(29,0.0000086)
(35,0.0000790)
(36,0.0000581)
(37,0.0000152)
(41,0.0000149)
(45,0.0000918)
(49,0.0000749)
(53,0.0000205)
(55,0.0001658)
(57,0.0080631)
(61,0.0000338)
(64,0.0001581)
(66,0.0002374)
(70,0.0140630)
(73,0.0000351)
(78,0.0003383)
(81,0.0002177)
(89,0.0000631)
(91,0.0004744)
(97,0.0000580)
(100,0.0330984)
(101,0.0000813)
(105,0.0005593)
(109,0.0000702)
(113,0.0001012)
(117,0.0558491)
(120,0.0007396)
(121,0.0398384)
(125,0.0004596)
(136,0.0010204)
(137,0.0001451)
(144,0.0680388)
(149,0.0001255)
(153,0.0012887)
(155,0.1194853)
(157,0.0001753)
(169,0.1074674)
(171,0.0015924)
(173,0.0002137)
(176,0.1614969)
(181,0.0001797)
(190,0.0017819)
(193,0.0003129)
(196,0.1921905)
(197,0.0002668)
(210,0.0021848)
(222,0.3045238)
(225,0.2203997)
(229,0.0002660)
(231,0.0028230)
(233,0.0004281)
(241,0.0003823)
(247,0.3956783)
(253,0.0033500)
(256,0.3433514)
(257,0.0004160)
(269,0.0003582)
(276,0.0039453)
(277,0.0004717)
(281,0.0003805)
(289,0.4988366)
(293,0.0005263)
(300,0.0046136)
(301,0.6817708)
(313,0.0004638)
(317,0.0006039)
(324,0.5589598)
(325,0.0053723)
(330,0.8389139)
(337,0.0006793)
(349,0.0007082)
(351,0.0062500)
(353,0.0005841)
(361,0.7778131)
(373,0.0006443)
(378,0.0067905)
(389,0.0006964)
(392,1.3456667)
(397,0.0009124)
(400,1.2792957)
(401,0.0007310)
(406,0.0078125)
(409,0.0007553)
(421,0.0010040)
(425,1.6148275)
(433,0.0008333)
(435,0.0089224)
(441,1.5411408)
(449,0.0011161)
(457,0.0014045)
(461,0.0011628)
(484,2.0622054)
(495,2.4727778)
(529,1.9931524)
(532,2.8872222)
(576,2.4543163)
(610,4.1862500)
(625,3.1283850)
(651,4.8897083)
(676,3.7484458)
(729,4.6576549)
(737,6.9808333)
(782,7.8930937)
(784,6.3707807)
(841,7.4596524)
(876,11.0258333)
(900,8.1302056)
(925,12.7375260)
(1024,13.2600000)
(1027,16.6583333)
(1156,18.0800000)
(1225,20.6000000)
(1296,24.1200000)
(1444,31.8700000)
(1600,41.7700000)
(2025,78.6900000)
(2500,136.4400000)};
\TracesOnlyMarks coordinates {
(4,0.0000064)
(5,0.0000068)
(6,0.0000120)
(7,0.0000066)
(9,0.0000154)
(10,0.0000158)
(12,0.0000263)
(13,0.0000146)
(15,0.0000236)
(16,0.0000266)
(17,0.0000181)
(21,0.0000366)
(25,0.0000372)
(26,0.0000641)
(28,0.0000606)
(29,0.0000257)
(35,0.0000907)
(36,0.0000758)
(37,0.0000349)
(41,0.0000516)
(45,0.0001438)
(49,0.0001021)
(53,0.0000582)
(55,0.0002024)
(57,0.0015699)
(61,0.0000750)
(64,0.0001900)
(66,0.0002874)
(70,0.0025074)
(73,0.0000831)
(78,0.0003846)
(81,0.0002185)
(89,0.0001143)
(91,0.0005423)
(97,0.0001314)
(100,0.0060434)
(101,0.0001857)
(105,0.0007022)
(109,0.0001605)
(113,0.0001935)
(117,0.0096300)
(120,0.0009225)
(121,0.0078539)
(125,0.0003788)
(136,0.0011364)
(137,0.0002287)
(144,0.0132799)
(149,0.0002599)
(153,0.0014620)
(155,0.0203950)
(157,0.0003303)
(169,0.0186078)
(171,0.0018248)
(173,0.0003268)
(176,0.0266446)
(181,0.0004950)
(190,0.0022124)
(193,0.0004655)
(196,0.0286902)
(197,0.0004892)
(210,0.0026172)
(222,0.0510436)
(225,0.0417843)
(229,0.0006313)
(231,0.0031646)
(233,0.0005556)
(241,0.0006024)
(247,0.0641396)
(253,0.0041189)
(256,0.0581619)
(257,0.0006427)
(269,0.0007331)
(276,0.0045682)
(277,0.0007440)
(281,0.0008897)
(289,0.0814112)
(293,0.0008278)
(300,0.0053457)
(301,0.1140037)
(313,0.0012376)
(317,0.0011364)
(324,0.1166627)
(325,0.0065705)
(330,0.1402265)
(337,0.0010870)
(349,0.0013441)
(351,0.0076894)
(353,0.0013736)
(361,0.1402735)
(373,0.0014970)
(378,0.0083065)
(389,0.0013812)
(392,0.2305750)
(397,0.0014535)
(400,0.1943956)
(401,0.0014793)
(406,0.0093056)
(409,0.0018116)
(421,0.0019380)
(425,0.2723133)
(433,0.0019841)
(435,0.0115909)
(441,0.2590226)
(449,0.0026882)
(457,0.0022433)
(461,0.0022235)
(484,0.3151886)
(495,0.4302727)
(529,0.4172742)
(532,0.5023681)
(576,0.4915291)
(610,0.7655556)
(625,0.6130718)
(651,0.8892393)
(676,0.7508001)
(729,0.9382786)
(737,1.3119444)
(782,1.4869737)
(784,1.1128398)
(841,1.3406538)
(876,2.1000000)
(900,1.5811321)
(925,2.3620703)
(1024,2.6100000)
(1027,3.3975000)
(1156,3.8100000)
(1225,4.2100000)
(1296,4.9300000)
(1444,7.0700000)
(1600,8.8400000)
(2025,16.6700000)
(2500,28.7800000)};
\end{semilogyaxis} 
\end{tikzpicture}%

%% file: fams/srg.tex
\begin{tikzpicture}[transform shape, scale=\picScale]
\def \timeout{250}
\def \xmin{500}
\def \xmax{2500}
\def \ymin{-0}
\def \ymax{350}
\begin{semilogyaxis}[
width=9cm,
height=8cm,
xmin=\xmin,  xmax=\xmax, 
ymin=\ymin,  ymax=\ymax, 
enlargelimits=0.01, 
grid=major,
scaled x ticks = false,
x tick label style={/pgf/number format/fixed},
xtick={500,1000,1500,2000,2500}]
\Bliss coordinates {
(4,0.0000107)
(5,0.0000090)
(6,0.0000150)
(7,0.0000230)
(9,0.0000189)
(10,0.0000238)
(12,0.0000344)
(13,0.0000173)
(15,0.0000399)
(16,0.0000354)
(17,0.0000244)
(21,0.0000750)
(25,0.0000791)
(26,0.0000846)
(28,0.0001417)
(29,0.0000612)
(35,0.0001620)
(36,0.0001982)
(37,0.0001050)
(41,0.0002139)
(45,0.0003618)
(49,0.0003128)
(53,0.0002844)
(55,0.0005817)
(57,0.0006860)
(61,0.0003999)
(64,0.0005177)
(66,0.0007816)
(70,0.0009728)
(73,0.0004405)
(78,0.0010466)
(81,0.0009183)
(89,0.0006787)
(91,0.0014641)
(97,0.0007984)
(100,0.0031650)
(101,0.0014388)
(105,0.0020284)
(109,0.0010220)
(113,0.0014837)
(117,0.0024603)
(120,0.0024570)
(121,0.0098619)
(125,0.0026525)
(136,0.0031949)
(137,0.0016461)
(144,0.0110193)
(149,0.0019802)
(153,0.0040404)
(155,0.0046142)
(157,0.0029499)
(169,0.0216300)
(171,0.0053763)
(173,0.0026596)
(176,0.0058099)
(181,0.0047847)
(190,0.0063291)
(193,0.0045351)
(196,0.0181967)
(197,0.0046512)
(210,0.0075188)
(222,0.0101958)
(225,0.0344926)
(229,0.0063492)
(231,0.0093023)
(233,0.0050378)
(241,0.0053763)
(247,0.0128722)
(253,0.0105820)
(256,0.0568594)
(257,0.0061538)
(269,0.0070671)
(276,0.0136986)
(277,0.0073260)
(281,0.0101010)
(289,0.0769586)
(293,0.0082305)
(300,0.0162602)
(301,0.0216472)
(313,0.0157480)
(317,0.0128205)
(324,0.0718540)
(325,0.0202020)
(330,0.0263682)
(337,0.0112921)
(349,0.0157480)
(351,0.0217391)
(353,0.0161290)
(361,0.0729119)
(373,0.0182727)
(378,0.0264474)
(389,0.0150376)
(392,0.0468474)
(397,0.0156250)
(400,0.1367646)
(401,0.0159524)
(406,0.0281690)
(409,0.0225843)
(421,0.0238095)
(425,0.0605074)
(433,0.0248148)
(435,0.0352632)
(441,0.1431982)
(449,0.0398039)
(457,0.0279167)
(461,0.0291304)
(484,0.3092971)
(495,0.0956794)
(529,0.2015170)
(532,0.1221728)
(576,0.4973632)
(610,0.1788034)
(625,0.2417207)
(651,0.2164053)
(676,1.0292793)
(729,0.4868268)
(737,0.3655238)
(782,0.4230595)
(784,0.2978175)
(841,0.2967319)
(876,0.6035669)
(900,0.8417557)
(925,0.6993483)
(1024,3.5000000)
(1027,0.9934722)
(1156,3.3500000)
(1225,0.2737500)
(1296,0.3566667)
(1444,0.4380000)
(1600,12.9300000)
(2025,1.2450000)
(2500,2.0100000)};
\Nauty coordinates {
(4,0.0000013)
(5,0.0000011)
(6,0.0000020)
(7,0.0000042)
(9,0.0000036)
(10,0.0000034)
(12,0.0000090)
(13,0.0000031)
(15,0.0000082)
(16,0.0000075)
(17,0.0000046)
(21,0.0000111)
(25,0.0000136)
(26,0.0000815)
(28,0.0000229)
(29,0.0000086)
(35,0.0001073)
(36,0.0000544)
(37,0.0000149)
(41,0.0000143)
(45,0.0000638)
(49,0.0000616)
(53,0.0000210)
(55,0.0001154)
(57,0.0083790)
(61,0.0000350)
(64,0.0001271)
(66,0.0001635)
(70,0.0145874)
(73,0.0000344)
(78,0.0002246)
(81,0.0001780)
(89,0.0000633)
(91,0.0003016)
(97,0.0000583)
(100,0.0376722)
(101,0.0000796)
(105,0.0003660)
(109,0.0000703)
(113,0.0001023)
(117,0.0566086)
(120,0.0004630)
(121,0.0387231)
(125,0.0004259)
(136,0.0006219)
(137,0.0001436)
(144,0.0658219)
(149,0.0001232)
(153,0.0007837)
(155,0.1206771)
(157,0.0001741)
(169,0.1041361)
(171,0.0009690)
(173,0.0002183)
(176,0.1646547)
(181,0.0001862)
(190,0.0010638)
(193,0.0003205)
(196,0.1859098)
(197,0.0002662)
(210,0.0012887)
(222,0.3172262)
(225,0.2779809)
(229,0.0002651)
(231,0.0016779)
(233,0.0004488)
(241,0.0003693)
(247,0.3957360)
(253,0.0019531)
(256,0.3327400)
(257,0.0004000)
(269,0.0003566)
(276,0.0022635)
(277,0.0004647)
(281,0.0003731)
(289,0.4805874)
(293,0.0005102)
(300,0.0025902)
(301,0.6877778)
(313,0.0004647)
(317,0.0006010)
(324,0.5408433)
(325,0.0030271)
(330,0.8339497)
(337,0.0007003)
(349,0.0007225)
(351,0.0035069)
(353,0.0005580)
(361,0.7550666)
(373,0.0006068)
(378,0.0037132)
(389,0.0006925)
(392,1.3382500)
(397,0.0009191)
(400,1.5272185)
(401,0.0007184)
(406,0.0043103)
(409,0.0007310)
(421,0.0009542)
(425,1.5994288)
(433,0.0008039)
(435,0.0049020)
(441,1.4943815)
(449,0.0010823)
(457,0.0013889)
(461,0.0011416)
(484,2.0139247)
(495,2.5279167)
(529,1.9528806)
(532,2.8931920)
(576,2.4218590)
(610,4.2950000)
(625,3.9329628)
(651,4.8960208)
(676,3.6533782)
(729,4.4840294)
(737,7.1183333)
(782,7.9440278)
(784,6.1913104)
(841,9.7835493)
(876,11.0700000)
(900,35.3987768)
(925,12.4030903)
(1024,57.9600000)
(1027,16.8033333)
(1156,78.4700000)
(1225,91.9200000)
(1296,108.6900000)
(1444,140.8500000)
(1600,185.9000000)
(2025,348.3600000)
(2500,198.5400000)};
\Conauto coordinates {
(4,0.0000077)
(5,0.0000066)
(6,0.0000088)
(7,0.0000145)
(9,0.0000129)
(10,0.0000141)
(12,0.0000202)
(13,0.0000141)
(15,0.0000248)
(16,0.0000251)
(17,0.0000166)
(21,0.0000382)
(25,0.0000571)
(26,0.0000333)
(28,0.0000867)
(29,0.0000249)
(35,0.0000413)
(36,0.0001158)
(37,0.0000371)
(41,0.0000355)
(45,0.0001778)
(49,0.0001673)
(53,0.0000545)
(55,0.0003281)
(57,0.0006403)
(61,0.0000757)
(64,0.0003083)
(66,0.0004410)
(70,0.0008525)
(73,0.0000931)
(78,0.0005447)
(81,0.0004348)
(89,0.0001137)
(91,0.0008580)
(97,0.0001200)
(100,0.0033857)
(101,0.0001568)
(105,0.0010560)
(109,0.0001315)
(113,0.0001974)
(117,0.0033963)
(120,0.0012658)
(121,0.0057134)
(125,0.0005137)
(136,0.0018382)
(137,0.0002010)
(144,0.0089241)
(149,0.0002144)
(153,0.0022962)
(155,0.0084293)
(157,0.0002734)
(169,0.0155345)
(171,0.0027174)
(173,0.0003275)
(176,0.0126089)
(181,0.0002865)
(190,0.0036969)
(193,0.0004517)
(196,0.0209548)
(197,0.0004184)
(210,0.0044248)
(222,0.0325490)
(225,0.0290674)
(229,0.0004736)
(231,0.0051414)
(233,0.0006127)
(241,0.0005900)
(247,0.0326994)
(253,0.0069444)
(256,0.0523919)
(257,0.0007186)
(269,0.0005741)
(276,0.0078740)
(277,0.0008471)
(281,0.0006315)
(289,0.0666729)
(293,0.0007752)
(300,0.0087719)
(301,0.0575635)
(313,0.0007479)
(317,0.0009809)
(324,0.0799161)
(325,0.0113636)
(330,0.0834172)
(337,0.0012731)
(349,0.0011448)
(351,0.0128025)
(353,0.0009046)
(361,0.1141851)
(373,0.0010293)
(378,0.0150376)
(389,0.0011370)
(392,0.1304598)
(397,0.0013708)
(400,0.1650376)
(401,0.0010977)
(406,0.0179464)
(409,0.0012188)
(421,0.0015420)
(425,0.1494596)
(433,0.0013793)
(435,0.0210526)
(441,0.2798547)
(449,0.0016667)
(457,0.0020202)
(461,0.0016181)
(484,0.2403240)
(495,0.2350390)
(529,0.2529637)
(532,0.3401702)
(576,0.3185322)
(610,0.3942673)
(625,0.4414899)
(651,0.5686338)
(676,0.5848018)
(729,0.6394672)
(737,0.9470238)
(782,0.9592577)
(784,0.6965011)
(841,1.1388875)
(876,1.0584087)
(900,1.3036194)
(925,1.5052193)
(1024,3.1200000)
(1027,2.1545833)
(1156,2.0400000)
(1225,2.1100000)
(1296,5.3700000)
(1444,3.6000000)
(1600,6.7700000)
(2025,19.6100000)
(2500,18.0800000)};
\Saucy coordinates {
(4,0.0000014)
(5,0.0000018)
(6,0.0000021)
(7,0.0000028)
(9,0.0000042)
(10,0.0000062)
(12,0.0000070)
(13,0.0000073)
(15,0.0000134)
(16,0.0000132)
(17,0.0000094)
(21,0.0000266)
(25,0.0000315)
(26,0.0001048)
(28,0.0000514)
(29,0.0000223)
(35,0.0001384)
(36,0.0000764)
(37,0.0000341)
(41,0.0000434)
(45,0.0001634)
(49,0.0001226)
(53,0.0000570)
(55,0.0002636)
(57,0.0045742)
(61,0.0000596)
(64,0.0002381)
(66,0.0003754)
(70,0.0074938)
(73,0.0001237)
(78,0.0005453)
(81,0.0002608)
(89,0.0001404)
(91,0.0007445)
(97,0.0002051)
(100,0.0188738)
(101,0.0001632)
(105,0.0009769)
(109,0.0002330)
(113,0.0002703)
(117,0.0290950)
(120,0.0012600)
(121,0.0281765)
(125,0.0007134)
(136,0.0016176)
(137,0.0003915)
(144,0.0461942)
(149,0.0003066)
(153,0.0021207)
(155,0.0636673)
(157,0.0005623)
(169,0.0645987)
(171,0.0026302)
(173,0.0004877)
(176,0.0855271)
(181,0.0005938)
(190,0.0032505)
(193,0.0004763)
(196,0.1028202)
(197,0.0006553)
(210,0.0039466)
(222,0.1699156)
(225,0.1507036)
(229,0.0008371)
(231,0.0046323)
(233,0.0007157)
(241,0.0009461)
(247,0.2145693)
(253,0.0056773)
(256,0.2130940)
(257,0.0011455)
(269,0.0008833)
(276,0.0067192)
(277,0.0015390)
(281,0.0013630)
(289,0.2716951)
(293,0.0013317)
(300,0.0080825)
(301,0.3968688)
(313,0.0012158)
(317,0.0015286)
(324,0.4000675)
(325,0.0094399)
(330,0.4831930)
(337,0.0017887)
(349,0.0018842)
(351,0.0110192)
(353,0.0023749)
(361,0.4993886)
(373,0.0021331)
(378,0.0122732)
(389,0.0022717)
(392,0.7986853)
(397,0.0029359)
(400,0.7271645)
(401,0.0024287)
(406,0.0148768)
(409,0.0031888)
(421,0.0026674)
(425,0.9459675)
(433,0.0022312)
(435,0.0167324)
(441,0.9464554)
(449,0.0023300)
(457,0.0031580)
(461,0.0023351)
(484,1.1911063)
(495,1.5354388)
(529,1.5993643)
(532,1.7316616)
(576,1.8685151)
(610,2.7233505)
(625,2.4074672)
(651,2.9661871)
(676,2.8532499)
(729,3.6020125)
(737,4.3720867)
(782,4.9143778)
(784,4.4071738)
(841,5.3046040)
(876,7.1319578)
(900,6.3174103)
(925,7.8199195)
(1024,10.6327270)
(1027,10.8898882)
(1156,15.4666320)
(1225,18.4860870)
(1296,21.5835800)
(1444,29.4309730)
(1600,37.7234300)
(2025,80.3928880)
(2500,144.1079350)};
\Traces coordinates {
(4,0.0000063)
(5,0.0000066)
(6,0.0000118)
(7,0.0000065)
(9,0.0000154)
(10,0.0000156)
(12,0.0000253)
(13,0.0000154)
(15,0.0000233)
(16,0.0000263)
(17,0.0000172)
(21,0.0000356)
(25,0.0000363)
(26,0.0000334)
(28,0.0000581)
(29,0.0000212)
(35,0.0000511)
(36,0.0000798)
(37,0.0000266)
(41,0.0000386)
(45,0.0001394)
(49,0.0001006)
(53,0.0000418)
(55,0.0001952)
(57,0.0003039)
(61,0.0000521)
(64,0.0001873)
(66,0.0002884)
(70,0.0004574)
(73,0.0000538)
(78,0.0003788)
(81,0.0002172)
(89,0.0000686)
(91,0.0005308)
(97,0.0000782)
(100,0.0010196)
(101,0.0001174)
(105,0.0006757)
(109,0.0000943)
(113,0.0001139)
(117,0.0013019)
(120,0.0009191)
(121,0.0014389)
(125,0.0003639)
(136,0.0011161)
(137,0.0001229)
(144,0.0027896)
(149,0.0001404)
(153,0.0014368)
(155,0.0023818)
(157,0.0001875)
(169,0.0116642)
(171,0.0017946)
(173,0.0001685)
(176,0.0029820)
(181,0.0002897)
(190,0.0022124)
(193,0.0002561)
(196,0.0171129)
(197,0.0002723)
(210,0.0025638)
(222,0.0049972)
(225,0.0206328)
(229,0.0003429)
(231,0.0031250)
(233,0.0002969)
(241,0.0003145)
(247,0.0060061)
(253,0.0039881)
(256,0.0284742)
(257,0.0003234)
(269,0.0003618)
(276,0.0045455)
(277,0.0003671)
(281,0.0004655)
(289,0.0336154)
(293,0.0004188)
(300,0.0052344)
(301,0.0092506)
(313,0.0006667)
(317,0.0005938)
(324,0.0458085)
(325,0.0064744)
(330,0.0109000)
(337,0.0005198)
(349,0.0006925)
(351,0.0074265)
(353,0.0007082)
(361,0.0535224)
(373,0.0007599)
(378,0.0083065)
(389,0.0006494)
(392,0.0161227)
(397,0.0006868)
(400,0.0662425)
(401,0.0007022)
(406,0.0092593)
(409,0.0008929)
(421,0.0009728)
(425,0.0183449)
(433,0.0010204)
(435,0.0106771)
(441,0.0856113)
(449,0.0014451)
(457,0.0011111)
(461,0.0011062)
(484,0.1024716)
(495,0.0261307)
(529,0.1185699)
(532,0.0293056)
(576,0.0364380)
(610,0.0422321)
(625,0.0430473)
(651,0.0462326)
(676,0.0497566)
(729,0.0575949)
(737,0.0616719)
(782,0.0695759)
(784,0.0687095)
(841,0.0802956)
(876,0.0898351)
(900,0.0911963)
(925,0.0989514)
(1024,0.1366667)
(1027,0.1269329)
(1156,0.1733333)
(1225,0.1990909)
(1296,0.2266667)
(1444,0.2787500)
(1600,0.3500000)
(2025,0.5850000)
(2500,0.9333333)};
\end{semilogyaxis} 
\end{tikzpicture}%

%% file: fams/srg-can.tex
\begin{tikzpicture}[transform shape, scale=\picScale]
\def \timeout{600}
\def \xmin{500}
\def \xmax{2500}
\def \ymin{-0}
\def \ymax{900}
\begin{semilogyaxis}[
width=9cm,
height=8cm,
xmin=\xmin,  xmax=\xmax, 
ymin=\ymin,  ymax=\ymax, 
enlargelimits=0.01, 
legend style={anchor=north west,
at={(0.01,0.998)},
font=\tiny,
inner xsep=-.5pt,
inner ysep=-.5pt,
fill=white,
draw=none},
grid=major,
scaled x ticks = false,
x tick label style={/pgf/number format/fixed},
xtick={500,1000,1500,2000,2500}]
\addplot [lightgray, no markers,line width=3pt] coordinates {(0,\timeout) (3000,\timeout)};
\addlegendentry{timeout: 600 secs}
\Bliss coordinates {
(4,0.0000130)
(5,0.0000110)
(6,0.0001194)
(7,0.0000280)
(9,0.0000238)
(10,0.0001790)
(12,0.0000511)
(13,0.0000236)
(15,0.0003109)
(16,0.0000540)
(17,0.0000340)
(21,0.0004955)
(25,0.0000996)
(26,0.0001203)
(28,0.0009901)
(29,0.0000847)
(35,0.0002896)
(36,0.0006864)
(37,0.0001435)
(41,0.0002720)
(45,0.0003986)
(49,0.0003783)
(53,0.0003717)
(55,0.0006487)
(57,0.0072604)
(61,0.0005054)
(64,0.0006478)
(66,0.0008482)
(70,0.0128471)
(73,0.0005889)
(78,0.0011325)
(81,0.0010622)
(89,0.0008711)
(91,0.0015773)
(97,0.0010325)
(100,0.0369658)
(101,0.0017376)
(105,0.0021858)
(109,0.0013132)
(113,0.0017778)
(117,0.0585876)
(120,0.0026212)
(121,0.0590761)
(125,0.0031008)
(136,0.0032733)
(137,0.0020899)
(144,0.1099728)
(149,0.0025063)
(153,0.0042283)
(155,0.1407175)
(157,0.0035149)
(169,0.1623786)
(171,0.0186111)
(173,0.0033670)
(176,0.1987289)
(181,0.0056818)
(190,0.0390385)
(193,0.0053050)
(196,0.2619277)
(197,0.0056180)
(210,0.0436957)
(222,0.4159545)
(225,0.4368480)
(229,0.0076923)
(231,0.0548649)
(233,0.0062696)
(241,0.0066890)
(247,0.5546678)
(253,0.0561111)
(256,0.5904214)
(257,0.0075758)
(269,0.0086580)
(276,0.0676667)
(277,0.0090909)
(281,0.0117647)
(289,0.8222386)
(293,0.0100000)
(300,0.0788462)
(301,0.9981944)
(313,0.0178571)
(317,0.0152672)
(324,1.2460984)
(325,0.1116667)
(330,1.2924949)
(337,0.0136054)
(349,0.0186111)
(351,0.1205882)
(353,0.0190476)
(361,1.6123216)
(373,0.0215054)
(378,0.1340000)
(389,0.0183486)
(392,2.2488889)
(397,0.0190476)
(400,2.2761979)
(401,0.0198020)
(406,0.1306250)
(409,0.0263158)
(421,0.0280556)
(425,2.7872994)
(433,0.0297059)
(435,0.1542857)
(441,3.0052080)
(449,0.0451111)
(457,0.0327869)
(461,0.0335000)
(484,3.9417521)
(495,4.5779167)
(529,5.2525789)
(532,5.5220270)
(576,7.0065077)
(610,8.8133333)
(625,8.8033295)
(651,10.2770076)
(676,10.7207730)
(729,13.2586325)
(737,15.6608333)
(782,17.8313172)
(784,16.7798352)
(841,20.9255271)
(876,26.2691667)
(900,26.0196702)
(925,30.3591667)
(1024,43.1100000)
(1027,42.8983333)
(1156,62.9600000)
(1225,75.9100000)
(1296,92.2300000)
(1444,127.2800000)
(1600,176.9600000)
(2025,356.4200000)
(2500,\timeout)};
\Nauty coordinates {
(4,0.0000014)
(5,0.0000011)
(6,0.0000021)
(7,0.0000045)
(9,0.0000043)
(10,0.0000036)
(12,0.0000108)
(13,0.0000031)
(15,0.0000093)
(16,0.0000083)
(17,0.0000046)
(21,0.0000125)
(25,0.0000151)
(26,0.0000813)
(28,0.0000288)
(29,0.0000086)
(35,0.0000790)
(36,0.0000581)
(37,0.0000152)
(41,0.0000149)
(45,0.0000918)
(49,0.0000749)
(53,0.0000205)
(55,0.0001658)
(57,0.0080631)
(61,0.0000338)
(64,0.0001581)
(66,0.0002374)
(70,0.0140630)
(73,0.0000351)
(78,0.0003383)
(81,0.0002177)
(89,0.0000631)
(91,0.0004744)
(97,0.0000580)
(100,0.0330984)
(101,0.0000813)
(105,0.0005593)
(109,0.0000702)
(113,0.0001012)
(117,0.0558491)
(120,0.0007396)
(121,0.0398384)
(125,0.0004596)
(136,0.0010204)
(137,0.0001451)
(144,0.0680388)
(149,0.0001255)
(153,0.0012887)
(155,0.1194853)
(157,0.0001753)
(169,0.1074674)
(171,0.0015924)
(173,0.0002137)
(176,0.1614969)
(181,0.0001797)
(190,0.0017819)
(193,0.0003129)
(196,0.1921905)
(197,0.0002668)
(210,0.0021848)
(222,0.3045238)
(225,0.2203997)
(229,0.0002660)
(231,0.0028230)
(233,0.0004281)
(241,0.0003823)
(247,0.3956783)
(253,0.0033500)
(256,0.3433514)
(257,0.0004160)
(269,0.0003582)
(276,0.0039453)
(277,0.0004717)
(281,0.0003805)
(289,0.4988366)
(293,0.0005263)
(300,0.0046136)
(301,0.6817708)
(313,0.0004638)
(317,0.0006039)
(324,0.5589598)
(325,0.0053723)
(330,0.8389139)
(337,0.0006793)
(349,0.0007082)
(351,0.0062500)
(353,0.0005841)
(361,0.7778131)
(373,0.0006443)
(378,0.0067905)
(389,0.0006964)
(392,1.3456667)
(397,0.0009124)
(400,1.2792957)
(401,0.0007310)
(406,0.0078125)
(409,0.0007553)
(421,0.0010040)
(425,1.6148275)
(433,0.0008333)
(435,0.0089224)
(441,1.5411408)
(449,0.0011161)
(457,0.0014045)
(461,0.0011628)
(484,2.0622054)
(495,2.4727778)
(529,1.9931524)
(532,2.8872222)
(576,2.4543163)
(610,4.1862500)
(625,3.1283850)
(651,4.8897083)
(676,3.7484458)
(729,4.6576549)
(737,6.9808333)
(782,7.8930937)
(784,6.3707807)
(841,7.4596524)
(876,11.0258333)
(900,8.1302056)
(925,12.7375260)
(1024,13.2600000)
(1027,16.6583333)
(1156,18.0800000)
(1225,20.6000000)
(1296,24.1200000)
(1444,31.8700000)
(1600,41.7700000)
(2025,78.6900000)
(2500,136.4400000)};
\Traces coordinates {
(4,0.0000064)
(5,0.0000068)
(6,0.0000120)
(7,0.0000066)
(9,0.0000154)
(10,0.0000158)
(12,0.0000263)
(13,0.0000146)
(15,0.0000236)
(16,0.0000266)
(17,0.0000181)
(21,0.0000366)
(25,0.0000372)
(26,0.0000641)
(28,0.0000606)
(29,0.0000257)
(35,0.0000907)
(36,0.0000758)
(37,0.0000349)
(41,0.0000516)
(45,0.0001438)
(49,0.0001021)
(53,0.0000582)
(55,0.0002024)
(57,0.0015699)
(61,0.0000750)
(64,0.0001900)
(66,0.0002874)
(70,0.0025074)
(73,0.0000831)
(78,0.0003846)
(81,0.0002185)
(89,0.0001143)
(91,0.0005423)
(97,0.0001314)
(100,0.0060434)
(101,0.0001857)
(105,0.0007022)
(109,0.0001605)
(113,0.0001935)
(117,0.0096300)
(120,0.0009225)
(121,0.0078539)
(125,0.0003788)
(136,0.0011364)
(137,0.0002287)
(144,0.0132799)
(149,0.0002599)
(153,0.0014620)
(155,0.0203950)
(157,0.0003303)
(169,0.0186078)
(171,0.0018248)
(173,0.0003268)
(176,0.0266446)
(181,0.0004950)
(190,0.0022124)
(193,0.0004655)
(196,0.0286902)
(197,0.0004892)
(210,0.0026172)
(222,0.0510436)
(225,0.0417843)
(229,0.0006313)
(231,0.0031646)
(233,0.0005556)
(241,0.0006024)
(247,0.0641396)
(253,0.0041189)
(256,0.0581619)
(257,0.0006427)
(269,0.0007331)
(276,0.0045682)
(277,0.0007440)
(281,0.0008897)
(289,0.0814112)
(293,0.0008278)
(300,0.0053457)
(301,0.1140037)
(313,0.0012376)
(317,0.0011364)
(324,0.1166627)
(325,0.0065705)
(330,0.1402265)
(337,0.0010870)
(349,0.0013441)
(351,0.0076894)
(353,0.0013736)
(361,0.1402735)
(373,0.0014970)
(378,0.0083065)
(389,0.0013812)
(392,0.2305750)
(397,0.0014535)
(400,0.1943956)
(401,0.0014793)
(406,0.0093056)
(409,0.0018116)
(421,0.0019380)
(425,0.2723133)
(433,0.0019841)
(435,0.0115909)
(441,0.2590226)
(449,0.0026882)
(457,0.0022433)
(461,0.0022235)
(484,0.3151886)
(495,0.4302727)
(529,0.4172742)
(532,0.5023681)
(576,0.4915291)
(610,0.7655556)
(625,0.6130718)
(651,0.8892393)
(676,0.7508001)
(729,0.9382786)
(737,1.3119444)
(782,1.4869737)
(784,1.1128398)
(841,1.3406538)
(876,2.1000000)
(900,1.5811321)
(925,2.3620703)
(1024,2.6100000)
(1027,3.3975000)
(1156,3.8100000)
(1225,4.2100000)
(1296,4.9300000)
(1444,7.0700000)
(1600,8.8400000)
(2025,16.6700000)
(2500,28.7800000)};
\end{semilogyaxis} 
\end{tikzpicture}%

%% file: fams/had.tex
\begin{tikzpicture}[transform shape, scale=\picScale]
\def \timeout{600}
\def \xmin{0}
\def \xmax{1030}
\def \ymin{-0}
\def \ymax{1700}
\begin{semilogyaxis}[
width=9cm,
height=8cm,
xmin=\xmin,  xmax=\xmax, 
ymin=\ymin,  ymax=\ymax, 
enlargelimits=0.01, 
legend style={anchor=north west,
at={(0.01,0.998)},
font=\tiny,
inner xsep=-.5pt,
inner ysep=-.5pt,
fill=white,
draw=none},
grid=major]
\addplot [lightgray, no markers,line width=3pt] coordinates {(-10,\timeout) (1100,\timeout)};
\addlegendentry{timeout: 600 secs}
\BlissOnlyMarks coordinates {
(4,0.0000588)
(8,0.0000497)
(16,0.0000289)
(32,0.0000985)
(48,0.0015349)
(64,0.0029630)
(80,0.0239286)
(96,0.0030769)
(112,0.0112994)
(128,0.0066225)
(144,0.0024420)
(160,0.0243902)
(176,0.0019627)
(192,0.0026810)
(208,0.0232558)
(224,0.0065574)
(240,0.0048309)
(256,0.0125786)
(272,0.0057471)
(288,0.0094340)
(304,0.0093897)
(320,0.0085106)
(336,0.0078125)
(352,0.0153435)
(368,0.0507500)
(384,0.0238095)
(400,0.4320000)
(416,0.0683333)
(432,0.0966667)
(448,0.4300000)
(464,0.4840000)
(480,0.1766667)
(496,0.1528571)
(512,0.3000000)
(528,0.1105263)
(544,0.3866667)
(560,0.1783333)
(576,0.1500000)
(592,0.2020000)
(608,0.1366667)
(624,0.6766667)
(640,0.2885714)
(656,0.1927273)
(672,0.1881818)
(688,0.7566667)
(704,0.4940000)
(720,0.2455556)
(736,1.9550000)
(752,110.3700000)
(768,0.2344444)
(784,0.2675000)
(800,0.2400000)
(816,0.3242857)
(832,5.3100000)
(848,0.2775000)
(864,0.2455556)
(880,0.3716667)
(896,0.2587500)
(912,0.3314286)
(928,3.7200000)
(944,161.9500000)
(960,0.0693103)
(976,0.1661538)
(992,0.1116667)
(1008,0.0976190)
(1024,0.2957143)};
\NautyOnlyMarks coordinates {
(4,0.0000009)
(8,0.0000016)
(16,0.0000055)
(32,0.0000161)
(48,0.0000495)
(64,0.0001029)
(80,0.0035211)
(96,0.0008961)
(112,0.0074632)
(128,0.0158594)
(144,0.0010460)
(160,0.0343750)
(176,0.0019841)
(192,0.0007418)
(208,0.7766667)
(224,0.0186607)
(240,0.0075758)
(256,0.0020833)
(272,0.0093981)
(288,0.0119643)
(304,0.0097596)
(320,0.0020161)
(336,0.0147059)
(352,0.0318750)
(368,2.6800000)
(384,0.0445833)
(400,6.0700000)
(416,0.0247727)
(432,0.0261250)
(448,1.4850000)
(464,7.1800000)
(480,0.1182353)
(496,0.0793750)
(512,0.0072857)
(528,0.0205769)
(544,1.4600000)
(560,0.0051786)
(576,0.2150000)
(592,0.0085000)
(608,0.0281944)
(624,6.7800000)
(640,0.2837500)
(656,0.0342187)
(672,0.0045909)
(688,18.7600000)
(704,0.0696875)
(720,0.1105263)
(736,89.1600000)
(752,\timeout)
(768,0.0064103)
(784,0.1800000)
(800,0.0068243)
(816,0.0668750)
(832,101.0200000)
(848,0.0582500)
(864,0.3583333)
(880,0.0850000)
(896,0.1827273)
(912,0.1972727)
(928,568.2500000)
(944,\timeout)
(960,0.0061585)
(976,2.9900000)
(992,1.3650000)
(1008,0.0912500)
(1024,0.0376786)};
\ConautoOnlyMarks coordinates {
(4,0.0000059)
(8,0.0000105)
(16,0.0000287)
(32,0.0000863)
(48,0.0001946)
(64,0.0002873)
(80,0.0015186)
(96,0.0010035)
(112,0.0009574)
(128,0.0014399)
(144,0.0013280)
(160,0.0047506)
(176,0.0014306)
(192,0.0017197)
(208,0.0673333)
(224,0.0044944)
(240,0.0030395)
(256,0.0044944)
(272,0.0035714)
(288,0.0047059)
(304,0.0046189)
(320,0.0050125)
(336,0.0050378)
(352,0.0074349)
(368,0.2060000)
(384,0.0178571)
(400,0.3466667)
(416,0.0070671)
(432,0.0087336)
(448,0.0248148)
(464,0.4680000)
(480,0.0116959)
(496,0.0103627)
(512,0.0173913)
(528,0.0114943)
(544,0.0824000)
(560,0.0160800)
(576,0.0153435)
(592,0.0262338)
(608,0.0139860)
(624,0.4940000)
(640,0.0229885)
(656,0.0184404)
(672,0.0203030)
(688,1.6600000)
(704,0.0404000)
(720,0.0233721)
(736,1.0800000)
(752,228.9900000)
(768,0.0275342)
(784,0.0295588)
(800,0.0288571)
(816,0.0335000)
(832,3.5900000)
(848,0.0320635)
(864,0.0327869)
(880,0.0404000)
(896,0.0327419)
(912,0.0374074)
(928,6.3700000)
(944,\timeout)
(960,0.0390385)
(976,0.0700000)
(992,0.0528947)
(1008,0.0500000)
(1024,0.0882609)};
\SaucyOnlyMarks coordinates {
(4,0.0000012)
(8,0.0000038)
(16,0.0000093)
(32,0.0000514)
(48,0.0001442)
(64,0.0002871)
(80,0.0115391)
(96,0.0054706)
(112,0.0073316)
(128,0.0520682)
(144,0.0187669)
(160,0.0519178)
(176,0.0750400)
(192,0.1062483)
(208,2.0414460)
(224,0.1010344)
(240,0.2399152)
(256,0.0061256)
(272,0.4331496)
(288,0.5021273)
(304,0.2004147)
(320,0.7852737)
(336,0.9347913)
(352,0.4282558)
(368,20.5983820)
(384,0.6452810)
(400,27.2043590)
(416,2.3443700)
(432,2.5824110)
(448,3.5693670)
(464,80.0669970)
(480,1.1458360)
(496,1.0899930)
(512,0.0273683)
(528,6.2764870)
(544,2.3948000)
(560,8.0387260)
(576,2.6581020)
(592,1.2463815)
(608,11.7267520)
(624,293.2123370)
(640,2.9052300)
(656,16.1550630)
(688,\timeout)
(704,5.4032480)
(720,24.4830980)
(736,\timeout)
(752,\timeout)
(768,32.4930320)
(784,2.4564590)
(800,40.1920980)
(816,5.2176720)
(832,\timeout)
(848,49.3869880)
(864,9.7506540)
(880,4.8504770)
(896,63.8004810)
(912,69.0142870)
(928,\timeout)
(944,\timeout)
(960,88.4347670)
(976,11.9730570)
(992,11.1558200)
(1008,114.5289100)
(1024,0.1260576)};
\TracesOnlyMarks coordinates {
(4,0.0000064)
(8,0.0000128)
(16,0.0000229)
(32,0.0000512)
(48,0.0001034)
(64,0.0001908)
(80,0.0006757)
(96,0.0007911)
(112,0.0005910)
(128,0.0018116)
(144,0.0008475)
(160,0.0030271)
(176,0.0008503)
(192,0.0010204)
(208,0.0114205)
(224,0.0028736)
(240,0.0018519)
(256,0.0032051)
(272,0.0018116)
(288,0.0020100)
(304,0.0022235)
(320,0.0024159)
(336,0.0023364)
(352,0.0033784)
(368,0.0429167)
(384,0.0044518)
(400,0.0550000)
(416,0.0037500)
(432,0.0035714)
(448,0.0152206)
(464,0.0756250)
(480,0.0065064)
(496,0.0048317)
(512,0.0131579)
(528,0.0054620)
(544,0.0086250)
(560,0.0074632)
(576,0.0072143)
(592,0.0065789)
(608,0.0067568)
(624,0.1008333)
(640,0.0115341)
(656,0.0075758)
(672,0.0100500)
(688,0.1406667)
(704,0.0316667)
(720,0.0102000)
(736,0.2333333)
(752,8.4800000)
(768,0.0127500)
(784,0.0119048)
(800,0.0164063)
(816,0.0122619)
(832,0.4480000)
(848,0.0169167)
(864,0.0155882)
(880,0.0147794)
(896,0.0181250)
(912,0.0197115)
(928,0.5425000)
(944,25.2700000)
(960,0.0169167)
(976,0.0460417)
(992,0.0326563)
(1008,0.0260000)
(1024,0.0595000)};
\end{semilogyaxis} 
\end{tikzpicture}%

%% file: fams/had-can.tex
\begin{tikzpicture}[transform shape, scale=\picScale]
\def \timeout{600}
\def \xmin{0}
\def \xmax{1030}
\def \ymin{-0}
\def \ymax{1700}
\begin{semilogyaxis}[
width=9cm,
height=8cm,
xmin=\xmin,  xmax=\xmax, 
ymin=\ymin,  ymax=\ymax, 
enlargelimits=0.01, 
legend style={anchor=north west,
at={(0.01,0.998)},
font=\tiny,
inner xsep=-.5pt,
inner ysep=-.5pt,
fill=white,
draw=none},
grid=major]
\addplot [lightgray, no markers,line width=3pt] coordinates {(-10,\timeout) (1060,\timeout)};
\addlegendentry{timeout: 600 secs}
\NautyOnlyMarks coordinates {
(4,0.0000009)
(8,0.0000017)
(16,0.0000062)
(32,0.0000178)
(48,0.0000649)
(64,0.0001236)
(80,0.0007082)
(96,0.0017819)
(112,0.0021659)
(128,0.0160156)
(144,0.0010549)
(160,0.0483333)
(176,0.0004394)
(192,0.0005519)
(208,0.5950000)
(224,0.0154412)
(240,0.0075000)
(256,0.0025510)
(272,0.0010204)
(288,0.0011013)
(304,0.0025000)
(320,0.0018750)
(336,0.0145833)
(352,0.0066118)
(368,2.0600000)
(384,0.0517500)
(400,5.0800000)
(416,0.0250000)
(432,0.0266250)
(448,0.2712500)
(464,6.4400000)
(480,0.0165625)
(496,0.0458333)
(512,0.0087069)
(528,0.0023264)
(544,1.4550000)
(560,0.0052865)
(576,0.2150000)
(592,0.0067568)
(608,0.0023148)
(624,5.9000000)
(640,0.1068421)
(656,0.0339062)
(672,0.0037687)
(688,8.6000000)
(704,0.0547500)
(720,0.0038654)
(736,82.2900000)
(752,\timeout)
(768,0.0062500)
(784,0.0140972)
(800,0.0068581)
(816,0.0097596)
(832,64.5000000)
(848,0.0065789)
(864,0.3633333)
(880,0.0828125)
(896,0.1836364)
(912,0.0056667)
(928,543.6600000)
(944,\timeout)
(960,0.0051276)
(976,1.2600000)
(992,1.5650000)
(1008,0.0912500)
(1024,0.0468750)};
\BlissOnlyMarks coordinates {
(4,0.0000099)
(8,0.0000140)
(16,0.0000356)
(32,0.0001212)
(48,0.0002967)
(64,0.0006221)
(80,0.0018282)
(96,0.0042283)
(112,0.0055402)
(128,0.0142857)
(144,0.0025543)
(160,0.0420833)
(176,0.0040404)
(192,0.0032626)
(208,0.2675000)
(224,0.0280556)
(240,0.0125786)
(256,0.0132237)
(272,0.0147059)
(288,0.0658065)
(304,0.0268000)
(320,0.0116959)
(336,0.0131373)
(352,0.0434043)
(368,0.8733333)
(384,0.0609091)
(400,1.8400000)
(416,0.0176316)
(432,0.0220879)
(448,0.2900000)
(464,1.5650000)
(480,0.1615385)
(496,0.0517949)
(512,0.0551351)
(528,0.1138889)
(544,0.3650000)
(560,0.0577143)
(576,0.1576923)
(592,0.0980952)
(608,0.0561111)
(624,3.4400000)
(640,0.2433333)
(656,0.1741667)
(672,0.0873913)
(688,8.3700000)
(704,0.9533333)
(720,0.0645161)
(736,5.3200000)
(768,0.1340000)
(784,0.2377778)
(800,0.1450000)
(816,0.1963636)
(832,22.0500000)
(848,0.1700000)
(864,0.3257143)
(880,0.3683333)
(896,0.4440000)
(912,0.1045000)
(928,47.0800000)
(960,0.1340000)
(976,0.3766667)
(992,0.8433333)
(1008,0.4580000)
(1024,0.2985714)};
\BlissMed coordinates {
(752,\timeout)
(944,\timeout)};
\TracesOnlyMarks coordinates {
(4,0.0000066)
(8,0.0000111)
(16,0.0000231)
(32,0.0000507)
(48,0.0001053)
(64,0.0001901)
(80,0.0005531)
(96,0.0007911)
(112,0.0008929)
(128,0.0024159)
(144,0.0012953)
(160,0.0040079)
(176,0.0009091)
(192,0.0011111)
(208,0.0419643)
(224,0.0062500)
(240,0.0021008)
(256,0.0032051)
(272,0.0020262)
(288,0.0023050)
(304,0.0037879)
(320,0.0027747)
(336,0.0025902)
(352,0.0078030)
(368,0.2050000)
(384,0.0082258)
(400,0.2850000)
(416,0.0043103)
(432,0.0040524)
(448,0.0405357)
(464,0.4420000)
(480,0.0147059)
(496,0.0095833)
(512,0.0136842)
(528,0.0063437)
(544,0.0172500)
(560,0.0086207)
(576,0.0154412)
(592,0.0143056)
(608,0.0079297)
(624,0.4100000)
(640,0.0284722)
(656,0.0086638)
(672,0.0117045)
(688,0.5975000)
(704,0.0634375)
(720,0.0119048)
(736,2.2600000)
(768,0.0144444)
(784,0.0261250)
(800,0.0179464)
(816,0.0300000)
(832,1.4150000)
(848,0.0203846)
(864,0.0412500)
(880,0.0340625)
(896,0.0216667)
(912,0.0231818)
(928,3.3200000)
(960,0.0204808)
(976,0.0562500)
(992,0.1194118)
(1008,0.0297222)
(1024,0.0595000)};
\TracesMed coordinates {
(752,\timeout)
(944,\timeout)};
\end{semilogyaxis} 
\end{tikzpicture}%

%% file: fams/rantree.tex
\begin{tikzpicture}[transform shape, scale=\picScale]
\def \timeout{600}
\def \xmin{0}
\def \xmax{100000}
\def \ymin{-0}
\def \ymax{1700}
\begin{loglogaxis}[
width=9cm,
height=8cm,
xmin=\xmin,  xmax=\xmax, 
ymin=\ymin,  ymax=\ymax, 
enlargelimits=0.01, 
legend style={anchor=north west,
at={(0.01,0.998)},
font=\tiny,
inner xsep=-.5pt,
inner ysep=-.5pt,
fill=white,
draw=none},
grid=major]
\addplot [lightgray, no markers,line width=3pt] coordinates {(10,\timeout) (\xmax,\timeout)};
\addlegendentry{timeout: 600 secs}
\Bliss coordinates {
(10,0.0000067)
(20,0.0000078)
(50,0.0000198)
(100,0.0000281)
(200,0.0000481)
(300,0.0000825)
(400,0.0001332)
(500,0.0002000)
(600,0.0002372)
(700,0.0002826)
(800,0.0004091)
(900,0.0004406)
(1000,0.0004894)
(2000,0.0016892)
(5000,0.0076046)
(10000,0.0307692)
(20000,0.1205882)
(50000,0.7966667)
(100000,3.1200000)};
\Nauty coordinates {
(10,0.0000009)
(20,0.0000015)
(50,0.0000097)
(100,0.0000242)
(200,0.0000672)
(300,0.0001853)
(400,0.0005308)
(500,0.0011792)
(600,0.0013812)
(700,0.0018657)
(800,0.0031563)
(900,0.0041189)
(1000,0.0051531)
(2000,0.0450000)
(5000,0.5175000)
(10000,5.4300000)
(20000,49.2100000)
(50000,\timeout)};
\Conauto coordinates {
(10,0.0000092)
(20,0.0000126)
(50,0.0000441)
(100,0.0000966)
(200,0.0002204)
(300,0.0004518)
(400,0.0006335)
(500,0.0008969)
(600,0.0013689)
(700,0.0016667)
(800,0.0021030)
(900,0.0023781)
(1000,0.0033727)
(2000,0.0173276)
(5000,0.1127778)
(10000,0.4680000)
(20000,1.8450000)
(50000,\timeout)};
\Saucy coordinates {
(10,0.0000013)
(50,0.0000057)
(100,0.0000097)
(200,0.0000141)
(300,0.0000235)
(400,0.0000371)
(500,0.0000458)
(600,0.0000513)
(700,0.0000627)
(800,0.0000754)
(900,0.0000934)
(1000,0.0001038)
(2000,0.0002329)
(5000,0.0006200)
(10000,0.0013398)
(20000,0.0028944)
(50000,0.0094720)
(100000,0.0254946)};
\Traces coordinates {
(10,0.0000068)
(20,0.0000077)
(50,0.0000102)
(100,0.0000136)
(200,0.0000201)
(300,0.0000293)
(400,0.0000361)
(500,0.0000440)
(600,0.0000688)
(700,0.0000611)
(800,0.0000734)
(900,0.0000996)
(1000,0.0000952)
(2000,0.0002137)
(5000,0.0008651)
(10000,0.0025902)
(20000,0.0068581)
(50000,0.0416667)
(100000,0.1546154)};
\end{loglogaxis} 
\end{tikzpicture}%

%% file: fams/rantree-can.tex
\begin{tikzpicture}[transform shape, scale=\picScale]
\def \timeout{600}
\def \xmin{0}
\def \xmax{100000}
\def \ymin{-0}
\def \ymax{1700}
\begin{loglogaxis}[
width=9cm,
height=8cm,
xmin=\xmin,  xmax=\xmax, 
ymin=\ymin,  ymax=\ymax, 
enlargelimits=0.01, 
legend style={anchor=north west,
at={(0.01,0.998)},
font=\tiny,
inner xsep=-.5pt,
inner ysep=-.5pt,
fill=white,
draw=none},
grid=major]
\addplot [lightgray, no markers,line width=3pt] coordinates {(10,\timeout) (110000,\timeout)};
\addlegendentry{timeout: 600 secs}
\Bliss coordinates {
(10,0.0000095)
(20,0.0000129)
(50,0.0000308)
(100,0.0000501)
(200,0.0000926)
(300,0.0001523)
(400,0.0002254)
(500,0.0003121)
(600,0.0003684)
(700,0.0004298)
(800,0.0005631)
(900,0.0006319)
(1000,0.0007189)
(2000,0.0020790)
(5000,0.0088889)
(10000,0.0327869)
(20000,0.1256250)
(50000,0.8000000)
(100000,3.2500000)};
\Nauty coordinates {
(10,0.0000010)
(20,0.0000015)
(50,0.0000106)
(100,0.0000272)
(200,0.0000767)
(300,0.0002315)
(400,0.0007205)
(500,0.0016129)
(600,0.0020594)
(700,0.0026882)
(800,0.0045909)
(900,0.0062500)
(1000,0.0075368)
(2000,0.0612500)
(5000,0.5900000)
(10000,5.6600000)
(20000,50.3300000)
(50000,\timeout)};
\Traces coordinates {
(10,0.0000071)
(20,0.0000077)
(50,0.0000104)
(100,0.0000140)
(200,0.0000207)
(300,0.0000304)
(400,0.0000370)
(500,0.0000456)
(600,0.0000700)
(700,0.0000648)
(800,0.0000769)
(900,0.0001024)
(1000,0.0000996)
(2000,0.0002254)
(5000,0.0008834)
(10000,0.0026447)
(20000,0.0072857)
(50000,0.0416667)
(100000,0.1576923)};
\end{loglogaxis} 
\end{tikzpicture}%

%% file: fams/cfi.tex
\begin{tikzpicture}[transform shape, scale=\picScale]
\def \timeout{10}
\def \xmin{200}
\def \xmax{2000}
\def \ymin{-0}
\def \ymax{10}
\begin{semilogyaxis}[
width=9cm,
height=8cm,
xmin=\xmin,  xmax=\xmax, 
ymin=\ymin,  ymax=\ymax, 
enlargelimits=0.01, 
grid=major,
/pgfplots/xtick={200,400,600,800,1200,1600,2000}]
\BlissOnlyMarks coordinates {
(200,0.0010315)
(220,0.0010947)
(240,0.0013841)
(260,0.0017138)
(280,0.0013360)
(300,0.0015552)
(320,0.0023810)
(340,0.0018100)
(360,0.0027360)
(380,0.0024125)
(400,0.0031496)
(420,0.0031847)
(440,0.0040650)
(460,0.0050000)
(480,0.0038314)
(500,0.0039526)
(520,0.0057307)
(540,0.0048426)
(560,0.0065359)
(580,0.0078125)
(600,0.0078431)
(620,0.0060976)
(640,0.0093897)
(660,0.0108696)
(680,0.0097561)
(700,0.0116279)
(720,0.0125786)
(740,0.0138889)
(760,0.0111732)
(780,0.0147794)
(800,0.0166116)
(820,0.0157480)
(840,0.0180180)
(860,0.0192308)
(880,0.0203030)
(900,0.0188679)
(920,0.0184404)
(940,0.0231034)
(960,0.0239286)
(980,0.0250000)
(1000,0.0261039)
(1020,0.0285714)
(1040,0.0257692)
(1060,0.0246914)
(1080,0.0291304)
(1100,0.0224719)
(1120,0.0317187)
(1140,0.0272973)
(1160,0.0356140)
(1180,0.0265789)
(1200,0.0300000)
(1220,0.0408163)
(1240,0.0398039)
(1260,0.0418750)
(1280,0.0461364)
(1300,0.0369091)
(1320,0.0465116)
(1340,0.0356140)
(1360,0.0370370)
(1380,0.0369091)
(1400,0.0540541)
(1420,0.0582857)
(1440,0.0648387)
(1460,0.0441304)
(1480,0.0648387)
(1500,0.0648387)
(1520,0.0693103)
(1540,0.0490244)
(1560,0.0523077)
(1580,0.0792308)
(1600,0.0755556)
(1620,0.0561111)
(1640,0.0588235)
(1660,0.0600000)
(1680,0.0625000)
(1700,0.0664516)
(1720,0.0658065)
(1740,0.0966667)
(1760,0.1015000)
(1780,0.1035000)
(1800,0.0735714)
(1820,0.1100000)
(1840,0.1100000)
(1860,0.1155556)
(1880,0.1116667)
(1900,0.1300000)
(1920,0.1318750)
(1940,0.0918182)
(1960,0.1155556)
(1980,0.1000000)
(2000,0.1035000)};
\NautyOnlyMarks coordinates {
(200,0.0017361)
(220,0.0021113)
(240,0.0025000)
(260,0.0030120)
(280,0.0033784)
(300,0.0039453)
(320,0.0044079)
(340,0.0050750)
(360,0.0056818)
(380,0.0063750)
(400,0.0071181)
(420,0.0080078)
(440,0.0116477)
(460,0.0094444)
(480,0.0175833)
(500,0.0110870)
(520,0.0200000)
(540,0.0130625)
(560,0.0141667)
(580,0.0151471)
(600,0.0163281)
(620,0.0174167)
(640,0.0184821)
(660,0.0200000)
(680,0.0212500)
(700,0.0226042)
(720,0.0238636)
(740,0.0251250)
(760,0.0267500)
(780,0.0286111)
(800,0.0297222)
(820,0.0312500)
(840,0.0331250)
(860,0.0350000)
(880,0.0373214)
(900,0.0394643)
(920,0.0420833)
(940,0.0431250)
(960,0.0452083)
(980,0.1287500)
(1000,0.0483333)
(1020,0.0502500)
(1040,0.1223529)
(1060,0.1041667)
(1080,0.0557500)
(1100,0.2388889)
(1120,0.1836364)
(1140,0.0653125)
(1160,0.0650000)
(1180,0.0675000)
(1200,0.0715625)
(1220,0.0731250)
(1240,0.2477778)
(1260,0.0790625)
(1280,0.2160000)
(1300,0.0850000)
(1320,0.2942857)
(1340,0.0904167)
(1360,0.1478571)
(1380,0.3616667)
(1400,0.4200000)
(1420,4.9900000)
(1440,0.4900000)
(1460,0.1111111)
(1480,0.4320000)
(1500,0.1182353)
(1520,0.1211765)
(1540,0.5250000)
(1560,0.1360000)
(1580,0.1380000)
(1600,0.3650000)
(1620,0.4160000)
(1640,0.1485714)
(1660,0.1538462)
(1680,0.1576923)
(1700,0.1638462)
(1720,4.2400000)
(1740,0.2575000)
(1760,2.3500000)
(1780,0.1808333)
(1800,1.1900000)
(1820,0.9900000)
(1840,0.3700000)
(1860,1.3300000)
(1880,0.4360000)
(1900,1.4150000)
(1920,0.2170000)
(1940,0.8233333)
(1960,0.2311111)
(1980,0.3866667)
(2000,0.6100000)};
\ConautoOnlyMarks coordinates {
(200,0.0046189)
(220,0.0044543)
(240,0.0059701)
(260,0.0078125)
(280,0.0061920)
(300,0.0071685)
(320,0.0117647)
(340,0.0089286)
(360,0.0132450)
(380,0.0119048)
(400,0.0150376)
(420,0.0149254)
(440,0.0211579)
(460,0.0300000)
(480,0.0208333)
(500,0.0213830)
(520,0.0340678)
(540,0.0266667)
(560,0.0350877)
(580,0.0497561)
(600,0.0502500)
(620,0.0338983)
(640,0.0600000)
(660,0.0714286)
(680,0.0594118)
(700,0.0755556)
(720,0.0854167)
(740,0.0952381)
(760,0.0571429)
(780,0.0966667)
(800,0.1063158)
(820,0.1015000)
(840,0.1366667)
(860,0.1333333)
(880,0.1464286)
(900,0.1235294)
(920,0.1005000)
(940,0.1675000)
(960,0.1646154)
(980,0.1675000)
(1000,0.1825000)
(1020,0.2120000)
(1040,0.1733333)
(1060,0.1428571)
(1080,0.2040000)
(1100,0.1464286)
(1120,0.2140000)
(1140,0.1457143)
(1160,0.2650000)
(1180,0.1691667)
(1200,0.1750000)
(1220,0.3171429)
(1240,0.2675000)
(1260,0.2857143)
(1280,0.3633333)
(1300,0.1900000)
(1320,0.3333333)
(1340,0.2300000)
(1360,0.2388889)
(1380,0.2090000)
(1400,0.3983333)
(1420,0.4300000)
(1440,0.4980000)
(1460,0.2700000)
(1480,0.4840000)
(1500,0.4940000)
(1520,0.5075000)
(1540,0.2957143)
(1560,0.3200000)
(1580,0.5975000)
(1600,0.5075000)
(1620,0.3550000)
(1640,0.3983333)
(1660,0.3228571)
(1680,0.4200000)
(1700,0.4200000)
(1720,0.4200000)
(1740,0.6050000)
(1760,0.7433333)
(1780,0.7666667)
(1800,0.5000000)
(1820,0.8266667)
(1840,0.6766667)
(1860,0.8200000)
(1880,0.7366667)
(1900,0.9800000)
(1920,1.0000000)
(1940,0.6050000)
(1960,0.6733333)
(1980,0.5750000)
(2000,0.6200000)};
\SaucyOnlyMarks coordinates {
(200,0.0004638)
(220,0.0003542)
(240,0.0006471)
(260,0.0006499)
(280,0.0008618)
(300,0.0006372)
(320,0.0007275)
(340,0.0008870)
(360,0.0007373)
(380,0.0011581)
(400,0.0013477)
(420,0.0010685)
(440,0.0010070)
(460,0.0017628)
(480,0.0014004)
(500,0.0017785)
(520,0.0013293)
(540,0.0016165)
(560,0.0014989)
(580,0.0016525)
(600,0.0019407)
(620,0.0029942)
(640,0.0030996)
(660,0.0019235)
(680,0.0019447)
(700,0.0029508)
(720,0.0036005)
(740,0.0025612)
(760,0.0037193)
(780,0.0028244)
(800,0.0030025)
(820,0.0034183)
(840,0.0027998)
(860,0.0045271)
(880,0.0034337)
(900,0.0040033)
(920,0.0040874)
(940,0.0035005)
(960,0.0037207)
(980,0.0056156)
(1000,0.0059339)
(1020,0.0044981)
(1040,0.0047716)
(1060,0.0042857)
(1080,0.0058489)
(1100,0.0065743)
(1120,0.0066363)
(1140,0.0048048)
(1160,0.0053430)
(1180,0.0080873)
(1200,0.0078115)
(1220,0.0084142)
(1240,0.0069086)
(1260,0.0082415)
(1280,0.0068465)
(1300,0.0064959)
(1320,0.0093587)
(1340,0.0073120)
(1360,0.0069770)
(1380,0.0104079)
(1400,0.0072643)
(1420,0.0088784)
(1440,0.0112507)
(1460,0.0102748)
(1480,0.0115410)
(1500,0.0079878)
(1520,0.0108677)
(1540,0.0103533)
(1560,0.0086559)
(1580,0.0086414)
(1600,0.0120923)
(1620,0.0127290)
(1640,0.0114673)
(1660,0.0147256)
(1680,0.0099860)
(1700,0.0097409)
(1720,0.0113192)
(1740,0.0151126)
(1760,0.0150994)
(1780,0.0112710)
(1800,0.0159878)
(1820,0.0123727)
(1840,0.0127121)
(1860,0.0119331)
(1880,0.0123105)
(1900,0.0134234)
(1920,0.0121923)
(1940,0.0175033)
(1960,0.0202116)
(1980,0.0142000)
(2000,0.0208897)};
\TracesOnlyMarks coordinates {
(200,0.0006158)
(220,0.0006410)
(240,0.0005708)
(260,0.0008446)
(280,0.0007375)
(300,0.0009653)
(320,0.0011574)
(340,0.0010460)
(360,0.0012563)
(380,0.0014970)
(400,0.0013736)
(420,0.0014286)
(440,0.0015060)
(460,0.0016422)
(480,0.0018474)
(500,0.0017857)
(520,0.0020661)
(540,0.0021930)
(560,0.0023364)
(580,0.0028409)
(600,0.0026729)
(620,0.0022841)
(640,0.0024393)
(660,0.0026172)
(680,0.0031646)
(700,0.0033224)
(720,0.0033784)
(740,0.0040079)
(760,0.0032372)
(780,0.0040079)
(800,0.0047170)
(820,0.0040726)
(840,0.0040323)
(860,0.0043319)
(880,0.0042373)
(900,0.0058430)
(920,0.0048077)
(940,0.0053457)
(960,0.0048798)
(980,0.0057386)
(1000,0.0051000)
(1020,0.0057670)
(1040,0.0063437)
(1060,0.0061585)
(1080,0.0070833)
(1100,0.0065064)
(1120,0.0068919)
(1140,0.0066118)
(1160,0.0070486)
(1180,0.0067105)
(1200,0.0074265)
(1220,0.0083750)
(1240,0.0076515)
(1260,0.0084583)
(1280,0.0080469)
(1300,0.0083750)
(1320,0.0076894)
(1340,0.0088793)
(1360,0.0081048)
(1380,0.0087069)
(1400,0.0099519)
(1420,0.0108333)
(1440,0.0098077)
(1460,0.0106250)
(1480,0.0125625)
(1500,0.0117614)
(1520,0.0108333)
(1540,0.0107812)
(1560,0.0113636)
(1580,0.0125000)
(1600,0.0120833)
(1620,0.0120833)
(1640,0.0128125)
(1660,0.0126250)
(1680,0.0126250)
(1700,0.0136184)
(1720,0.0287500)
(1740,0.0128125)
(1760,0.0145139)
(1780,0.0138889)
(1800,0.0152941)
(1820,0.0140278)
(1840,0.0157031)
(1860,0.0173333)
(1880,0.0158594)
(1900,0.0175000)
(1920,0.0169167)
(1940,0.0184821)
(1960,0.0169167)
(1980,0.0162500)
(2000,0.0192308)};
\end{semilogyaxis} 
\end{tikzpicture}%

%% file: fams/cfi-can.tex
\begin{tikzpicture}[transform shape, scale=\picScale]
\def \timeout{100}
\def \xmin{200}
\def \xmax{2000}
\def \ymin{-0}
\def \ymax{100}
\begin{semilogyaxis}[
width=9cm,
height=8cm,
xmin=\xmin,  xmax=\xmax, 
ymin=\ymin,  ymax=\ymax, 
enlargelimits=0.01, 
grid=major,
/pgfplots/xtick={200,400,600,800,1200,1600,2000}]
\BlissOnlyMarks coordinates {
(200,0.0020921)
(220,0.0022099)
(240,0.0017286)
(260,0.0036563)
(280,0.0018215)
(300,0.0028450)
(320,0.0046838)
(340,0.0031201)
(360,0.0048309)
(380,0.0044944)
(400,0.0051414)
(420,0.0052770)
(440,0.0058651)
(460,0.0088889)
(480,0.0076923)
(500,0.0067568)
(520,0.0082645)
(540,0.0094340)
(560,0.0136735)
(580,0.0151515)
(600,0.0159524)
(620,0.0101523)
(640,0.0128025)
(660,0.0155814)
(680,0.0203030)
(700,0.0216129)
(720,0.0220879)
(740,0.0295588)
(760,0.0168067)
(780,0.0277778)
(800,0.0331148)
(820,0.0320635)
(840,0.0314062)
(860,0.0360714)
(880,0.0370370)
(900,0.0357143)
(920,0.0291304)
(940,0.0502500)
(960,0.0478571)
(980,0.0410204)
(1000,0.0427660)
(1020,0.0643750)
(1040,0.0683333)
(1060,0.0606061)
(1080,0.0431915)
(1100,0.0689655)
(1120,0.0680000)
(1140,0.0444444)
(1160,0.0717857)
(1180,0.0472093)
(1200,0.1127778)
(1220,0.0858333)
(1240,0.1035000)
(1260,0.0966667)
(1280,0.0824000)
(1300,0.1010000)
(1320,0.0728571)
(1340,0.1094737)
(1360,0.0927273)
(1380,0.0878261)
(1400,0.1750000)
(1420,0.2388889)
(1440,0.1546154)
(1460,0.1538462)
(1480,0.2322222)
(1500,0.2300000)
(1520,0.1546154)
(1540,0.1393333)
(1560,0.1413333)
(1580,0.2537500)
(1600,0.1791667)
(1620,0.1025000)
(1640,0.1800000)
(1660,0.1155556)
(1680,0.2277778)
(1700,0.2687500)
(1720,0.3600000)
(1740,0.1716667)
(1760,0.3228571)
(1780,0.2750000)
(1800,0.1818182)
(1820,0.2800000)
(1840,0.4040000)
(1860,0.4280000)
(1880,0.4380000)
(1900,0.5300000)
(1920,0.4900000)
(1940,0.4400000)
(1960,0.4420000)
(1980,0.2150000)
(2000,0.5475000)};
\NautyOnlyMarks coordinates {
(200,0.0048317)
(220,0.0084167)
(240,0.0035069)
(260,0.0085417)
(280,0.0095833)
(300,0.0146528)
(320,0.0128125)
(340,0.0178571)
(360,0.0288889)
(380,0.0291667)
(400,0.0362500)
(420,0.0317187)
(440,0.1111111)
(460,0.0681250)
(480,0.1105263)
(500,0.0257500)
(520,0.1428571)
(540,0.1845455)
(560,0.0809375)
(580,0.1972727)
(600,0.1872727)
(620,0.3533333)
(640,0.2650000)
(660,0.1666667)
(680,0.2210000)
(700,0.4000000)
(720,0.3683333)
(740,0.2222222)
(760,0.4600000)
(780,0.1836364)
(800,0.5250000)
(820,0.7933333)
(840,0.3300000)
(860,0.4160000)
(880,0.4960000)
(900,0.8966667)
(920,0.8600000)
(940,1.5300000)
(960,1.5300000)
(980,0.8066667)
(1000,0.4320000)
(1020,0.4600000)
(1040,0.9300000)
(1060,1.4650000)
(1080,1.8600000)
(1100,2.3000000)
(1120,3.7000000)
(1140,1.3400000)
(1160,1.1000000)
(1180,1.0800000)
(1200,2.7300000)
(1220,1.3150000)
(1240,3.8900000)
(1260,4.2200000)
(1280,3.5100000)
(1300,4.8800000)
(1320,4.6600000)
(1340,3.6600000)
(1360,5.5400000)
(1380,5.5300000)
(1400,4.4000000)
(1420,8.0600000)
(1440,4.7400000)
(1460,3.8900000)
(1480,7.8800000)
(1500,11.0600000)
(1520,3.5400000)
(1540,2.0600000)
(1560,3.0100000)
(1580,7.7400000)
(1600,12.0200000)
(1620,5.6800000)
(1640,10.1200000)
(1660,14.7000000)
(1680,12.8200000)
(1700,8.4400000)
(1720,13.2700000)
(1740,9.1100000)
(1760,7.1900000)
(1780,7.6000000)
(1800,13.6800000)
(1820,17.2900000)
(1840,19.9800000)
(1860,16.1100000)
(1880,35.6300000)
(1900,15.5100000)
(1920,18.6000000)
(1940,12.6900000)
(1960,15.0300000)
(1980,31.3700000)
(2000,23.4300000)};
\TracesOnlyMarks coordinates {
(200,0.0006068)
(220,0.0006667)
(240,0.0005787)
(260,0.0008418)
(280,0.0007102)
(300,0.0009690)
(320,0.0011905)
(340,0.0010331)
(360,0.0012255)
(380,0.0014867)
(400,0.0013661)
(420,0.0014205)
(440,0.0015060)
(460,0.0016667)
(480,0.0017946)
(500,0.0017694)
(520,0.0020594)
(540,0.0022039)
(560,0.0022936)
(580,0.0028736)
(600,0.0026882)
(620,0.0023481)
(640,0.0023810)
(660,0.0025773)
(680,0.0032212)
(700,0.0033667)
(720,0.0032895)
(740,0.0039683)
(760,0.0032051)
(780,0.0039258)
(800,0.0045312)
(820,0.0040984)
(840,0.0040323)
(860,0.0044079)
(880,0.0042083)
(900,0.0059524)
(920,0.0048317)
(940,0.0054348)
(960,0.0049020)
(980,0.0058140)
(1000,0.0050000)
(1020,0.0056667)
(1040,0.0063437)
(1060,0.0061585)
(1080,0.0069257)
(1100,0.0064103)
(1120,0.0065385)
(1140,0.0066447)
(1160,0.0070139)
(1180,0.0068581)
(1200,0.0077652)
(1220,0.0082258)
(1240,0.0076515)
(1260,0.0086207)
(1280,0.0079688)
(1300,0.0083065)
(1320,0.0076136)
(1340,0.0090625)
(1360,0.0081855)
(1380,0.0087500)
(1400,0.0098558)
(1420,0.0108696)
(1440,0.0098077)
(1460,0.0107812)
(1480,0.0126875)
(1500,0.0118750)
(1520,0.0107812)
(1540,0.0106250)
(1560,0.0112500)
(1580,0.0125000)
(1600,0.0118182)
(1620,0.0117614)
(1640,0.0125000)
(1660,0.0125000)
(1680,0.0127500)
(1700,0.0134868)
(1720,0.0286111)
(1740,0.0131250)
(1760,0.0144444)
(1780,0.0137500)
(1800,0.0152941)
(1820,0.0141667)
(1840,0.0155882)
(1860,0.0170833)
(1880,0.0157813)
(1900,0.0180357)
(1920,0.0167500)
(1940,0.0183036)
(1960,0.0171667)
(1980,0.0161719)
(2000,0.0190179)};
\end{semilogyaxis} 
\end{tikzpicture}%

%% file: fams/mz-aug2.tex
\begin{tikzpicture}[transform shape, scale=\picScale]
\def \timeout{600}
\def \xmin{90}
\def \xmax{1200}
\def \ymin{-0}
\def \ymax{1620}
\begin{semilogyaxis}[
width=9cm,
height=8cm,
xmin=\xmin,  xmax=\xmax, 
ymin=\ymin,  ymax=\ymax, 
enlargelimits=0.01, 
legend style={anchor=north west,
at={(0.01,0.998)},
font=\tiny,
inner xsep=-.5pt,
inner ysep=-.5pt,
fill=white,
draw=none},
grid=major]
\addplot [lightgray, no markers,line width=3pt] coordinates {(0,\timeout) (1300,\timeout)};
\addlegendentry{timeout: 600 secs}
\Bliss coordinates {
(96,0.0001207)
(144,0.0002201)
(192,0.0003721)
(240,0.0005848)
(288,0.0008726)
(336,0.0012626)
(384,0.0017778)
(432,0.0024010)
(480,0.0031847)
(528,0.0041494)
(576,0.0054496)
(624,0.0070175)
(672,0.0085106)
(720,0.0104712)
(768,0.0129032)
(816,0.0154615)
(864,0.0185185)
(912,0.0219780)
(960,0.0259740)
(1008,0.0307576)
(1056,0.0356140)
(1104,0.0416667)
(1152,0.0469767)
(1200,0.0536842)};
\Nauty coordinates {
(96,0.0003472)
(144,0.0036232)
(192,0.0315625)
(240,0.2344444)
(288,1.4900000)
(336,7.3300000)
(384,51.7100000)
(432,253.9700000)
(480,\timeout)};
\Conauto coordinates {
(96,0.0006581)
(144,0.0012812)
(192,0.0020855)
(240,0.0030120)
(288,0.0041667)
(336,0.0054496)
(384,0.0069444)
(432,0.0087336)
(480,0.0108108)
(528,0.0125786)
(576,0.0150376)
(624,0.0170940)
(672,0.0198020)
(720,0.0227273)
(768,0.0259740)
(816,0.0307692)
(864,0.0344828)
(912,0.0383019)
(960,0.0412245)
(1008,0.0459091)
(1056,0.0502500)
(1104,0.0566667)
(1152,0.0609091)
(1200,0.0666667)};
\Saucy coordinates {
(96,0.0000679)
(144,0.0000838)
(192,0.0001023)
(240,0.0001218)
(288,0.0001430)
(336,0.0001642)
(384,0.0001812)
(432,0.0002069)
(480,0.0002270)
(528,0.0002428)
(576,0.0002593)
(624,0.0002817)
(672,0.0002959)
(720,0.0003202)
(768,0.0003346)
(816,0.0003661)
(864,0.0003786)
(912,0.0004124)
(960,0.0004337)
(1008,0.0004402)
(1056,0.0004649)
(1104,0.0004769)
(1152,0.0004944)
(1200,0.0005373)};
\Traces coordinates {
(96,0.0000429)
(144,0.0000655)
(192,0.0000968)
(240,0.0001303)
(288,0.0001813)
(336,0.0002220)
(384,0.0002700)
(432,0.0003201)
(480,0.0003956)
(528,0.0004579)
(576,0.0005112)
(624,0.0005981)
(672,0.0006757)
(720,0.0007716)
(768,0.0008091)
(816,0.0008865)
(864,0.0010373)
(912,0.0011364)
(960,0.0012500)
(1008,0.0013812)
(1056,0.0014368)
(1104,0.0015625)
(1152,0.0016340)
(1200,0.0017986)};
\end{semilogyaxis} 
\end{tikzpicture}%

%% file: fams/mz-aug2-can.tex
\begin{tikzpicture}[transform shape, scale=\picScale]
\def \timeout{600}
\def \xmin{90}
\def \xmax{1200}
\def \ymin{-0}
\def \ymax{1620}
\begin{semilogyaxis}[
width=9cm,
height=8cm,
xmin=\xmin,  xmax=\xmax, 
ymin=\ymin,  ymax=\ymax, 
enlargelimits=0.01, 
legend style={anchor=north west,
at={(0.01,0.998)},
font=\tiny,
inner xsep=-.5pt,
inner ysep=-.5pt,
fill=white,
draw=none},
grid=major]
\addplot [lightgray, no markers,line width=3pt] coordinates {(0,\timeout) (1300,\timeout)};
\addlegendentry{timeout: 600 secs}
\Bliss coordinates {
(96,0.0002213)
(144,0.0004515)
(192,0.0008435)
(240,0.0014749)
(288,0.0026316)
(336,0.0051948)
(384,0.0102564)
(432,0.0228409)
(480,0.0384615)
(528,0.0913636)
(576,0.1854545)
(624,0.3900000)
(672,0.6625000)
(720,1.5200000)
(768,2.7100000)
(816,6.5400000)
(864,11.9200000)
(912,27.2700000)
(960,47.0300000)
(1008,98.7800000)
(1056,230.1900000)
(1104,420.0600000)
(1152,\timeout)
(1200,\timeout)};
\Nauty coordinates {
(96,0.0012821)
(144,0.0326562)
(192,0.1211765)
(240,1.6300000)
(288,11.7700000)
(336,42.8600000)
(384,277.9200000)
(432,\timeout)};
\Traces coordinates {
(96,0.0000429)
(144,0.0000684)
(192,0.0000975)
(240,0.0001310)
(288,0.0001820)
(336,0.0002258)
(384,0.0002657)
(432,0.0003161)
(480,0.0003949)
(528,0.0004537)
(576,0.0005319)
(624,0.0005967)
(672,0.0006631)
(720,0.0007463)
(768,0.0008143)
(816,0.0008803)
(864,0.0010246)
(912,0.0011312)
(960,0.0012500)
(1008,0.0013736)
(1056,0.0014440)
(1104,0.0016447)
(1152,0.0016026)
(1200,0.0017730)};
\end{semilogyaxis} 
\end{tikzpicture}%

%% file: fams/planes.tex
\begin{tikzpicture}[transform shape, scale=\picScale]
\def \timeout{3600}
\def \xmin{$P_1$}
\def \xmax{$P_{12}$}
\def \ymin{-0}
\def \ymax{7000}
\begin{semilogyaxis}[%
width=10cm,
height=8cm,
xmin=\xmin,  xmax=\xmax, 
ymin=\ymin,  ymax=\ymax, 
enlargelimits=0.025, 
legend style={anchor=north west,
at={(0.01,0.998)},
font=\tiny,
inner xsep=-.5pt,
inner ysep=-.5pt,
fill=white,
draw=none},
grid=major,
/pgfplots/xtick={$P_1$,$P_2$,$P_3$,$P_4$,$P_5$,$P_6$,$P_7$,$P_8$,$P_9$,$P_{10}$,$P_{11}$,$P_{12}$},
symbolic x coords={A,$P_1$,$P_2$,$P_3$,$P_4$,$P_5$,$P_6$,$P_7$,$P_8$,$P_9$,$P_{10}$,$P_{11}$,$P_{12}$,B}]
\addplot [lightgray, no markers,line width=3pt] coordinates {(A,\timeout) (B,\timeout)};
\addlegendentry{timeout: 3600 secs}
\NautyLightOnlyMarks coordinates {
($P_1$,0.0084583)
($P_2$,28.66)
($P_3$,1.62)
($P_4$,108.95)
($P_5$,11.92)
($P_6$,89.13)
($P_7$,185.95)
($P_8$,33.24)
($P_9$,\timeout)
($P_{10}$,965.35)
($P_{11}$,202.61)
($P_{12}$,1276.55)};
\BlissOnlyMarks coordinates {
($P_1$,0.0054201)
($P_2$,256.38)
($P_3$,40.48)
($P_4$,71.03)
($P_5$,38.43)
($P_6$,1864.01)
($P_7$,185.95)
($P_8$,89.20)
($P_9$,2347.88)
($P_{10}$,327.76)
($P_{11}$,1249.03)
($P_{12}$,2643.41)};
\ConautoOnlyMarks coordinates {
($P_1$,0.0137931)
($P_2$,118.85)
($P_3$,24.03)
($P_4$,28.6)
($P_5$,2.59)
($P_6$,86.05)
($P_7$,45.79)
($P_8$,99.55)
($P_9$,370.99)
($P_{10}$,150.12)
($P_{11}$,342.19)
($P_{12}$,32.65)};
\SaucyOnlyMarks coordinates {
($P_1$,0.0019061)
($P_2$,\timeout)
($P_3$,\timeout)
($P_4$,\timeout)
($P_5$,\timeout)
($P_6$,\timeout)
($P_7$,\timeout)
($P_8$,\timeout)
($P_9$,\timeout)
($P_{10}$,\timeout)
($P_{11}$,\timeout)
($P_{12}$,\timeout)};
\TracesOnlyMarks coordinates {
($P_1$,0.0054348)
($P_2$,0.0382143)
($P_3$,0.0441667)
($P_4$,0.05775)
($P_5$,0.0778125)
($P_6$,0.0831250)
($P_7$,0.11)
($P_8$,0.1836364)
($P_9$,0.2971429)
($P_{10}$,0.6125)
($P_{11}$,0.9166667)
($P_{12}$,0.8466667)};
\end{semilogyaxis}
\end{tikzpicture}

%% file: fams/planestable.tex
\renewcommand{\arraystretch}{1.1}
\scriptsize{
  \begin{tabular}[b]{@{} lrr @{}}
    \# & group size & orbits \\ \hline
    $P_1$ & $3.42171648e10$ & $1$ \\ 
    $P_2$ & $921{,}600$ & $6$ \\ 
    $P_3$ & $884{,}736$ & $3$ \\ 
    $P_4$ & $258{,}048$ & $6$ \\ 
    $P_5$ & $147{,}456$ & $3$ \\ 
    $P_6$ & $92{,}160$ & $8$ \\ 
    $P_7$ & $55{,}296$ & $8$ \\ 
    $P_8$ & $18{,}432$ & $5$ \\ 
    $P_9$ & $12{,}288$ & $6$ \\ 
    $P_{10}$ & $3{,}840$ & $10$ \\ 
    $P_{11}$ & $3{,}456$ & $12$ \\ 
    $P_{12}$ & $2{,}304$ & $14$ \\ 
     &  & \\ 
  \end{tabular}
}

%% file: fams/comb.tex
\begin{tikzpicture}[transform shape, scale=\picScale]
\def \timeout{3600}
\def \xmin{$C_1$}
\def \xmax{$C_{12}$}
\def \ymin{-0}
\def \ymax{7000}
\begin{semilogyaxis}[%
width=10cm,
height=8cm,
xmin=\xmin,  xmax=\xmax, 
ymin=\ymin,  ymax=\ymax, 
enlargelimits=0.025, 
legend style={anchor=north west,
at={(0.01,0.998)},
font=\tiny,
inner xsep=-.5pt,
inner ysep=-.5pt,
fill=white,
draw=none},
grid=major,
/pgfplots/xtick={$C_1$,$C_2$,$C_3$,$C_4$,$C_5$,$C_6$,$C_7$,$C_8$,$C_9$,$C_{10}$,$C_{11}$,$C_{12}$},
symbolic x coords={A,$C_1$,$C_2$,$C_3$,$C_4$,$C_5$,$C_6$,$C_7$,$C_8$,$C_9$,$C_{10}$,$C_{11}$,$C_{12}$,B}]
\addplot [lightgray, no markers,line width=3pt] coordinates {(A,\timeout) (B,\timeout)};
\addlegendentry{timeout: 3600 secs}
\NautyOnlyMarks coordinates {
($C_1$,3052.61)
($C_2$,\timeout)
($C_3$,\timeout)
($C_4$,1371.82)
($C_5$,26.10)
($C_6$,\timeout)
($C_9$,\timeout)
($C_{10}$,0.486)
($C_{11}$,0.6425)};
\NautyLightOnlyMarks coordinates {
($C_7$,78.20)
($C_8$,\timeout)
($C_{12}$,909.59)};
\BlissOnlyMarks coordinates {
($C_1$,1.21)
($C_2$,91.90)
($C_3$,10.68)
($C_4$,0.7633333)
($C_5$,4.57)
($C_6$,81.98)
($C_7$,442.21)
($C_8$,\timeout)
($C_9$,2.16)
($C_{10}$,0.6766667)
($C_{11}$,0.0536842)
($C_{12}$,664.69)};
\ConautoOnlyMarks coordinates {
($C_1$,3.12)
($C_2$,1084.23)
($C_3$,438.52)
($C_4$,10.91)
($C_6$,704.59)
($C_7$,50.16)
($C_9$,16.41)
($C_{10}$,0.5725)
($C_{11}$,0.0640625)
($C_{12}$,957.98)};
\ConautoSmall coordinates {
($C_8$,\timeout)};
\SaucyOnlyMarks coordinates {
($C_1$,\timeout)
($C_2$,\timeout)
($C_3$,\timeout)
($C_4$,\timeout)
($C_5$,11.929546)
($C_6$,\timeout)
($C_7$,\timeout)
($C_9$,\timeout)
($C_{10}$,316.921072)
($C_{11}$,105.586394)
($C_{12}$,\timeout)};
\SaucySmall coordinates {
($C_8$,\timeout)};
\TracesOnlyMarks coordinates {
($C_1$,1.135)
($C_2$,13.07)
($C_3$,2.65)
($C_4$,0.2433333)
($C_5$,3.83)
($C_6$,35.11)
($C_7$,0.345)
($C_8$,25.37)
($C_9$,1.74)
($C_{10}$,0.0929167)
($C_{11}$,0.0081048)
($C_{12}$,1.5)};
\end{semilogyaxis}
\end{tikzpicture}

%% file: fams/combtable.tex
\renewcommand{\arraystretch}{1.1}
\scriptsize{
  \begin{tabular}[b]{@{\hspace{3pt}} l@{\hspace{3pt}}r@{\hspace{3pt}}r@{\hspace{3pt}}r @{}}
    \# & $(V,E)$ & group size & orbits \\ \hline
$C_1$ & $(3650,598600)$ & $324$ & $30$ \\
$C_2$ & $(15984,10725264)$ & $2{,}125{,}873{,}200$ & $2$ \\
$C_3$ & $(7300,2693700)$ & $188{,}956{,}800$ & $2$ \\
$C_4$ & $(2752,481600)$ & $38{,}723{,}328$ & $2$ \\
$C_5$ & $(900000,1200000)$ & $600{,}000$ & $2$ \\
$C_6$ & $(15984,10725264)$ & $231{,}913{,}440$ & $2$ \\ 
$C_7$ & $(1302,16926)$ & $1{,}488{,}000$ & $2$ \\
$C_8$ & $(8322,270465)$ & $43{,}352{,}064$ & $8$ \\
$C_9$ & $(3650,598600)$ & $72$ & $65$ \\
$C_{10}$ & $(3276,245700)$ & $9{,}000{,}000$ & $3$ \\
$C_{11}$ & $(756,49140)$ & $9{,}000{,}000$ & $2$ \\ 
$C_{12}$ & $(1514,21196)$ & $122{,}472$ & $4$ \\ 
\\
 \end{tabular}
}

%% file: fams/planescomb-can.tex
\begin{tikzpicture}[transform shape, scale=\picScale]
\def \timeout{3600}
\def \xmin{$P_1$}
\def \xmax{$C_{12}$}
\def \ymin{-0}
\def \ymax{7000}
\begin{semilogyaxis}[%
width=18cm,
height=8cm,
xmin=\xmin,  xmax=\xmax, 
ymin=\ymin,  ymax=\ymax, 
enlargelimits=0.015, 
grid=major,
legend style={anchor=north west,
at={(0.01,0.998)},
font=\tiny,
inner xsep=-.5pt,
inner ysep=-.5pt,
fill=white,
draw=none},
/pgfplots/xtick={$P_1$,$P_2$,$P_3$,$P_4$,$P_5$,$P_6$,$P_7$,$P_8$,$P_9$,$P_{10}$,$P_{11}$,$P_{12}$,$C_1$,$C_2$,$C_3$,$C_4$,$C_5$,$C_6$,$C_7$,$C_8$,$C_9$,$C_{10}$,$C_{11}$,$C_{12}$},
symbolic x coords={A,$P_1$,$P_2$,$P_3$,$P_4$,$P_5$,$P_6$,$P_7$,$P_8$,$P_9$,$P_{10}$,$P_{11}$,$P_{12}$,$C_1$,$C_2$,$C_3$,$C_4$,$C_5$,$C_6$,$C_7$,$C_8$,$C_9$,$C_{10}$,$C_{11}$,$C_{12}$,B}]
\addplot [lightgray, no markers,line width=3pt] coordinates {(A,\timeout) (B,\timeout)};
\addlegendentry{timeout: 3600 secs}
\NautyLightOnlyMarks coordinates {
($P_1$,0.0054620)
($P_2$,28.15)
($P_3$,1.560000)
($P_4$,8.47)
($P_5$,8.05)
($P_6$,12.12)
($P_7$,23.09)
($P_8$,31.84)
($P_9$,\timeout)
($P_{10}$,984.77)
($P_{11}$,191.06)
($P_{12}$,1048.77)
($C_7$,58.83)
($C_8$,\timeout)
($C_{12}$,890.97)};
\NautyOnlyMarks coordinates {
($C_1$,1984.60)
($C_2$,\timeout)
($C_3$,2010.96)
($C_4$,1045.60)
($C_5$,27.00)
($C_6$,\timeout)
($C_9$,\timeout)
($C_{10}$,0.40)
($C_{11}$,0.575)};
\BlissOnlyMarks coordinates {
($P_1$,0.008658)
($P_2$,438.18)
($P_3$,139.66)
($P_4$,76.04)
($P_5$,90.43)
($P_6$,837.93)
($P_7$,398.32)
($P_8$,278.45)
($P_9$,551.38)
($P_{10}$,616.29)
($P_{11}$,1169.17)
($P_{12}$,1191.51)
($C_1$,101.3700000)
($C_2$,\timeout)
($C_3$,\timeout)
($C_4$,32.7000000)
($C_5$,2.1000000)
($C_6$,\timeout)
($C_7$,2576.58)
($C_8$,\timeout)
($C_9$,76.1200000)
($C_{10}$,0.83)
($C_{11}$,4.7600000)
($C_{12}$,\timeout)};
\TracesOnlyMarks coordinates {
($P_1$,0.0031019)
($P_2$,0.125)
($P_3$,0.115)
($P_4$,0.3328571)
($P_5$,0.42)
($P_6$,0.6733333)
($P_7$,1.105)
($P_8$,2.53)
($P_9$,3.94)
($P_{10}$,11.57)
($P_{11}$,13.68)
($P_{12}$,20.21)
($C_1$,3.03)
($C_2$,10.47)
($C_3$,2.40)
($C_4$,0.454)
($C_5$,4.17)
($C_6$,28.53)
($C_7$,1.85)
($C_8$,356.68)
($C_9$,5.92)
($C_{10}$,0.1194118)
($C_{11}$,0.0165625)
($C_{12}$,38.33)};
\end{semilogyaxis}
\end{tikzpicture}